\documentclass[runningheads]{llncs} 
\usepackage[T1]{fontenc}
\usepackage{graphicx}
\usepackage{graphicx}
\usepackage{xcolor}

\usepackage{subfig}
\usepackage{wrapfig}

\usepackage{tikz}
\usetikzlibrary{arrows,calc,shapes,shadows,decorations.pathmorphing,decorations.pathreplacing,automata,shapes.multipart,positioning}

\usepackage{mathtools}
\usepackage{amsmath,latexsym,amssymb,amsfonts,amscd,stmaryrd,textcomp} 
\usepackage{pifont}
\usepackage{mathrsfs}
\usepackage{upgreek}

\usepackage{booktabs}

\usepackage{multicol}
\usepackage{makecell}

\usepackage[ruled,vlined,linesnumbered]{algorithm2e}

\usepackage{csquotes}

\usepackage{hyperref}
\hypersetup{%
	colorlinks=true,           
	allcolors=blue!70!black,   
	pdfstartview=Fit,          
	breaklinks=true}

\usepackage{xspace} 
\usepackage{xargs} 

\usepackage[disable]{todonotes}
\usepackage{etoolbox} 

\usepackage{soul}

\makeatletter
\@ifclassloaded{lipics-v2021}{%
}{
	\usepackage{thmtools}
	\usepackage[shortlabels]{enumitem}
	\usepackage[nameinlink,noabbrev]{cleveref} 
	\renewcommand{\cref}{\Cref} 
}
\makeatother 
\spnewtheorem*{openproblem}{Open Problem}{\itshape}{}

\newrobustcmd{\christof}[2][]{{\color{black}\todo[color=green!20,#1]{{\bf Christof:} #2}}\ignorespaces} 
\newrobustcmd{\lina}[2][]{{\color{black}\todo[color=cyan!20,#1]{{\bf Lina:} #2}}\ignorespaces} 

\definecolor{green100}{RGB}{87,171,39} 

\newcommand\R{\mathbb R}

\newcommand\N{\mathbb N}

\renewcommand{\iff}{\Leftrightarrow}

\newcommand{\Max}[1]{\text{max} \{#1\}}

\newcommand{\Dom}{\textit{Dom}}

\makeatletter
\newsavebox{\@brx}
\newcommand{\llangle}[1][]{\savebox{\@brx}{\(\m@th{#1\langle}\)}%
	\mathopen{\copy\@brx\kern-0.5\wd\@brx\usebox{\@brx}}}
\newcommand{\rrangle}[1][]{\savebox{\@brx}{\(\m@th{#1\rangle}\)}%
	\mathclose{\copy\@brx\kern-0.5\wd\@brx\usebox{\@brx}}}
\newcommand{\llcombined}{\mathopen{|\mkern-3mu\langle}}
\newcommand{\rrcombined}{\mathclose{\rangle\mkern-3mu|}}
\makeatother

\newcommand\superalpha[1][]{{#1}^{\strat}}
\newcommand\superbeta[1][]{{#1}^{\strattwo}}
\newcommand\superalphabeta[1][]{{#1}^{\strat {\leftarrow} \strattwo}}

\newcommand\indexi[1][]{{#1}_{i}}

\newcommand\superstar[1][]{{#1}^*}
\newcommand\superstaropt[1][]{{#1}^*_{\opt}}
\newcommand\superstars[1][]{{#1}^*_{\state}}

\tikzset{st/.style={draw,circle}} 
\tikzset{ap/.style={font={\footnotesize},align=center}}
\tikzstyle{dist}=[circle, inner sep=0.8pt, solid, draw=black,fill=black]
\tikzset{trans/.style={font={\footnotesize},thick}}

\newcommand{\states}{S}
\newcommand\Act{\textit{Act}}
\newcommand\Trans{\mathbf{P}}
\newcommand\AP{\mathsf{AP}}
\newcommand{\labelingfct}{L}

\newcommand\Paths[1][\tbg]{\textit{Paths}^{#1}}
\newcommandx{\Pathsfrom}[2][1=\tbg, 2=\state, usedefault]{\textit{Paths}^{#1}_{#2}}
\newcommand\finPaths[1][\tbg]{\textit{Paths}_{\textit{fin}}^{#1}}
\newcommandx{\finPathsfrom}[2][1=\tbg, 2=\state, usedefault]{\textit{Paths}_{\textit{fin}, #2}^{#1}}

\newcommand{\state}{s}
\newcommand{\action}{a}
\newcommand{\actiontwo}{b}
\newcommand{\actionthree}{c}
\newcommand{\ap}{\textsf{p}}
\newcommand{\aptwo}{\textsf{q}}

\newcommand\dtmc{\mathcal{D}}
\newcommand\dtmctup[1] [\refl]{(#1[\states], \allowbreak #1[\Trans], \allowbreak #1[\AP], \allowbreak #1[L] )}

\newcommand{\tbg}{\mathcal{G}}
\newcommand{\tbgtup}[1] [\refl] {(#1[\agents], \allowbreak #1[\states], \allowbreak #1[\Act], \allowbreak #1[\Trans], \allowbreak #1[\AP], \allowbreak #1[\labelingfct])}
\newcommand{\agents}{\textit{Ags}}

\newcommand{\agent}{g}
\newcommand{\setofagents}{A}
\newcommand{\setofagentstwo}{B}

\newcommand\sched{\sigma}
\newcommand{\modes}{Q}
\newcommand{\mode}{q}
\newcommand{\modef}{\textit{mode}}
\newcommand{\start}{\textit{init}}
\newcommand{\act}{\textit{act}}

\newcommand{\schedtup}[1] [\refl]{(#1[\modes], \allowbreak #1[\start], \allowbreak #1[\modef], \allowbreak #1[\act])}

\newcommand{\strat}{\alpha} 
\newcommand{\strattwo}{\beta} 

\newcommand{\strattup}[1] [\refl]{\schedtup[{#1}]}
\newcommandx{\Strats}[2][1=\tbg, 2=\setofagents, usedefault] {\textit{Str}^{#1}_{#2}} 

\newcommand{\partelt}[1][i]{S_{#1}} 

\newcommand{\swspepsne}{SW-SP-$\epsilon$-NE\xspace}
\newcommand{\spepsne}{subgame perfect $\epsilon$-NE\xspace} 
\newcommand{\objone}{X^C}
\newcommand{\objtwo}{X^{\overline{C}}}

\renewcommand{\phi}{\varphi}
\renewcommand\implies{\Rightarrow}
\newcommand{\true}{\top} 

\newcommand{\variable}[1]{\hat{#1}}
\newcommand\varstate{\variable{\state}}
\newcommand{\VarsState}{\variable{\states}} 

\newcommand{\existsStrat}[2]{\llangle #1 \rrangle_{#2} \ }
\newcommand{\forallStrat}[2]{\| #1 \|_{#2} \ }
\newcommand{\setofvariables}{\variable{R}} 

\newcommand{\mapStrat}{\textit{map}_{\Strats[ ][ ]}}

\newcommand{\statetup}{\vec{\state}}
\newcommand{\statetupMap}[1][\state]{\vec{#1}}

\newcommand{\Ctx}{\Gamma}
\newcommandx{\Ctxtup}[3][1=\tbg, 2=\mapStrat, 3=\statetup, usedefault]{(#1, #2, #3)}

\newcommand{\modelsMD}{\models_{\textit MD}}

\newcommand{\modelsHyperSGL}{\models_{\HyperSGL}}
\newcommand{\modelsHyperPCTL}{\models_{\HyperPCTL}}
\newcommand{\modelsSGL}{\models_{\SGL}}
\newcommand{\modelsHyperSGLHD}{\models_{\HyperSGL, \textit HD}}
\newcommand{\modelsHyperSL}{\models_{\HyperSL}}
\newcommand{\modelsrPATL}{\models_{\rPATL}}

\newcommand{\projS}[1]{{#1}_\states}

\newcommand{\emptymap}{\{\}}

\newcommand{\HyperPCTL}{\textsf{\smaller HyperPCTL}\xspace}

\newcommand{\LTL}{\textsf{\smaller LTL}\xspace}
\newcommand{\CTL}{\textsf{\smaller CTL}\xspace}
\newcommand{\CTLstar}{\textsf{\smaller CTL$^*$}\xspace}
\newcommand{\ATL}{\textsf{\smaller ATL}\xspace}
\newcommand{\ATLstar}{\textsf{\smaller ATL$^*$}\xspace}

\newcommand{\PCTL}{\textsf{\smaller PCTL}\xspace}
\newcommand{\PCTLstar}{\textsf{\smaller PCTL$^*$}\xspace}

\newcommand{\HyperLTL}{\textsf{\smaller HyperLTL}\xspace}

\newcommand{\HyperCTLstar}{\textsf{\smaller HyperCTL$^*$}\xspace}

\newcommand{\HyperATLstar}{\textsf{\smaller HyperATL$^*$}\xspace}

\newcommand{\SL}{\textsf{\smaller SL}\xspace}
\newcommand{\HyperSL}{\textsf{\smaller HyperSL}\xspace}
\newcommand{\GL}{\textsf{\smaller GL}\xspace}

\newcommand{\HyperPCTLstar}{\textsf{\smaller HyperPCTL$^*$}\xspace}

\newcommand{\AHyperPCTL}{\textsf{\smaller AHyperPCTL}\xspace}

\newcommand{\PHL}{\textsf{\smaller PHL}\xspace}

\newcommand{\SGL}{\textsf{\smaller SGL}\xspace}
\newcommand{\PATL}{\textsf{\smaller PATL}\xspace}
\newcommand{\PATLstar}{\textsf{\smaller PATL$^*$}\xspace}
\newcommand{\rPATL}{\textsf{\smaller rPATL}\xspace}
\newcommand{\rPATLstar}{\textsf{\smaller rPATL$^*$}\xspace}
\newcommand{\PATLNE}{\textsf{\smaller PATL$_{\text{NE}}$}\xspace}

\newcommand{\HyperSGL}{\textsf{\smaller HyperSt$^{\mathsf{2}}$}\xspace}

\newcommand\Next{\mathbin{\bigcirc}\,}
\newcommand\Until{\mathbin{\mathsf{U}}}
\newcommand\Globally{\mathbin{\square}}
\newcommand\Finally{\mathbin{\lozenge}}

\newcommand{\Prob}{\mathbb{P}}

\newcommand\pschedq{\varphi^{str}}
\newcommand\pstateq{\varphi^{sta}} 
\newcommand\pnonquant{\varphi^{nq}}
\newcommand\pprob{\varphi^{pr}}
\newcommand\ppath{\varphi^{path}}

\newcommand\varsched{\hat{\sched}}

\newcommand{\varstrat}{\variable{x}}
\newcommand{\varstratt}{\variable{y}}

\newcommand{\dra}{\mathcal{A}}

\newcommand{\translation}{T}
\newcommand{\init}{\textit{init}}
\newcommand{\sinit}{\state_{\init}}

\newcommand{\PATLfrag}{\textsf{\smaller PATL$_{\text{1-NE}}$}\xspace}

\newcommand{\nasheq}{\strat^*}
\newcommand{\schedcomp}{\strat}

\newcommand{\varfixed}{\varstate_0}
\newcommand{\varfixedp}{\varstate'_0}
\newcommand{\varopt}{\varstate_1} 
\newcommand{\varoptp}{\varstate'_1} 
\newcommand{\varcomp}{\varstate_2} 
\newcommand{\varcompp}{\varstate'_2} 

\newcommand{\fixed}{\state}
\newcommand{\opt}{\state_1}
\newcommand{\comp}{\state_2}

\newcommand{\phinested}{\phi_{\textit{nested}}}

\newcommand{\ProbSGL}{\mathbb{P}}
\newcommand{\SGLfrag}{$\SGL_{\mathsf{\scriptstyle LTL}}$\xspace}
\newcommand{\tbgpq}{\tbg_{\ap\aptwo}}

\newcommand{\pathmap}{\Uppi}
\newcommand{\HyperSLfrag}{$\HyperSL_{uni}$\xspace} 
\newcommand{\varpath}{\hat{\pi}}
\newcommand{\VarsPath}{\hat{\Pi}}
\newcommand{\VarsStrat}{\variable{X}}

\newcommand{\Vars}{\textit{Vars}}
\newcommand{\Agt}{\textit{Agt}}
\newcommand{\RefVar}{\textit{RefVar}}
\newcommand{\valundefined}{\texttt{undef}}

\newcommand{\phipath}{\phi_{\text{path}}} 

\newcommand{\Deltafixed}{\Delta^{\vec{\strat}}} 
\newcommand{\pathmapfixed}{\pathmap^{\vec{\strat}}} 
\newcommand{\statetupMapfixed}{\statetupMap^{\vec{\strat}}} 
\newcommand{\mapStratfixed}{\mapStrat^{\vec{\strat}}}

\newcommand{\istate}{r}

\newcommand{\phiex}{\phi_{ex}}

\newcommand{\pif}{\pi}

\newcommand{\Type}{\textsf{Type}}

\renewcommand{\P}{\textsf{P}\xspace}
\newcommand{\NP}{\textsf{NP}\xspace}
\newcommand{\PSPACE}{\textsf{PSPACE}\xspace}
\newcommand{\coNP}{\textsf{coNP}\xspace}
\newcommand{\EXPTIME}{\textsf{EXPTIME}\xspace}

\SetKwInOut{Input}{Input}
\SetKwInOut{Output}{Output}
\SetKwProg{Fn}{Function}{}{}
\DontPrintSemicolon

\SetNlSkip{0.3ex} 
\SetAlgoSkip{}
\SetKwFunction{FCheck}{check}
\SetKwFunction{FCalc}{calcProb}
\SetKwFunction{FSucc}{succ}

\newcommand{\lIfElse}[3]{\lIf{#1}{#2 \textbf{else}~#3}}

\newcommand{\statedict}[1][\state]{\statetupMap[#1]}
\newcommand{\algorithmicnot}{\normalfont{\textbf{not}} }
\newcommand{\algorithmicand}{\normalfont{\textbf{and}} }

\newcommand{\Sat}{\textit{Sat}}

\SetKwFunction{FMain}{Main}
\SetKwFunction{FSem}{Sem}
\SetKwFunction{FSemProb}{SemProb}
\SetKwFunction{FSemProbUntil}{SemProbUntil}
\SetKwFunction{FTruth}{Truth}
\SetKwFunction{FSemUUntil}{SemUnboundedUntil}

\newcommand{\numstatequant}{n}

\newcommand{\phifull}{Q_1 \varstate_1 \ldots Q_\numstatequant \varstate_\numstatequant .\ \pschedq}

\newcommand{\countQuant}{\textit{ctr}}
\newcommand{\currMap}{\textit{map}}

\newcommand{\actiontup}{\vec{\action}}
\newcommand{\modetup}{\vec{\mode}}
\newcommand{\fsucc}{\textit{succ}}
\newcommand{\currInd}{\textit{curr}}

\newcommand{\stratenc}{\strat}
\newcommand{\modefenc}{\modef}
\newcommand{\holds}{\textit{h}}
\newcommand{\holdsToInt}{\textit{hInt}}
\newcommand{\prob}{\textit{pr}}
\newcommandx{\stratencWith}[3][1=j, 2=\state, 3=\action, usedefault]{\stratenc_{#1, #2, #3}}
\newcommandx{\modefencWith}[4][1=j, 2=\state, 3=\mode, 4=\mode', usedefault]{\modefenc_{#1, #2, #3, #4}}
\newcommandx{\holdsWith}[2][1=\statetup, 2=\phi, usedefault]{\holds_{#1, #2}}
\newcommandx\holdsToIntWith[2][1=\statetup, 2=\phi, usedefault]{\holdsToInt_{#1, #2}}
\newcommandx\probWith[2][1=\statetup, 2=\phi, usedefault]{\prob_{#1, #2}}

\newcommand{\orcid}[1]{\href{https://orcid.org/#1}{\includegraphics[width=3mm]{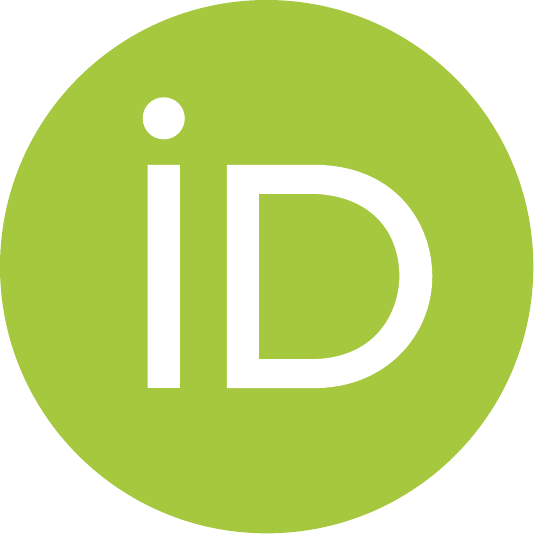}}}

\newtoggle{extended} 
\toggletrue{extended} 

\begin{document}
\title{A Hyperlogic for Strategies in Stochastic Games
	\iftoggle{extended}{(Extended Version)}{}
} 
%
%
\author{Lina Gerlach\orcid{0009-0002-5506-6181} \and
Christof L{\"o}ding\orcid{0000-0002-1529-2806} \and
Erika {\'A}brah{\'a}m\orcid{0000-0002-5647-6134}}
\authorrunning{L. Gerlach et al.}
%
\institute{RWTH Aachen University, Aachen, Germany\\
	\email{\{gerlach,loeding,abraham\}@cs.rwth-aachen.de}}
\maketitle              

\begin{abstract}
	We propose a probabilistic hyperlogic called \HyperSGL that can express hyperproperties of strategies in turn-based stochastic games. 
	To the best of our knowledge, \HyperSGL is the first hyperlogic for stochastic games. 
	\HyperSGL can relate probabilities of several independent executions of strategies in a stochastic game.
	For example, in \HyperSGL it is natural to formalize optimality, i.e., to express that some strategy is better than all other strategies, or to express the existence of Nash equilibria.
	We investigate the expressivity of \HyperSGL by comparing it to existing logics for stochastic games, as well as existing hyperlogics.
	Though the model-checking problem for \HyperSGL is in general undecidable, 
	we show that it becomes decidable for bounded memory and is in \EXPTIME and \PSPACE-hard over memoryless deterministic strategies, and we identify a fragment for which the model-checking problem is \PSPACE-complete.
\end{abstract}

\section{Introduction}
\label{sec:intro}

\emph{Hyperproperties}, formally defined as properties of sets of traces \cite{clarksonHyperproperties2010}, can relate different executions of a system. 
Hyperlogics like \HyperLTL and \HyperCTLstar \cite{clarksonTemporalLogics2014} have been proposed to express hyperproperties of discrete nondeterministic systems. 
In the presence of random components in the system behavior, \emph{probabilistic hyperproperties} can express probabilistic relations between several independent executions of a probabilistic, potentially nondeterministic, system.
Probabilistic hyperlogics for \emph{Markov decision processes} (\emph{MDPs}) or \emph{discrete-time Markov chains} (\emph{DTMCs}) have been proposed in \cite{abrahamHyperPCTLTemporal2018,abrahamProbabilisticHyperproperties2020,dimitrovaProbabilisticHyperproperties2020,wangStatisticalModel2021}. 

Since MDPs can be viewed as single-player stochastic games, it is natural to extend this setting to multiple players (also called \emph{agents}).
%
In \emph{multi-player stochastic games}, several agents may influence the choice of which action is taken, and each action is associated with a probabilistic distribution over successor states. 
In the class of turn-based games, at every step exactly one of the agents can make a choice, 
while in concurrent games, every agent may make a choice at every step. 

On the one hand, temporal logics like \rPATL \cite{kwiatkowskaAutomatedVerification2018,kwiatkowskaEquilibriabasedProbabilistic2019} and \SGL \cite{baierStochasticGame2012} allow to specify properties of stochastic games, but not at the level of hyperproperties.
On the other hand, temporal logics like \SL \cite{chatterjeeStrategyLogic2007} or \ATLstar \cite{alurAlternatingtimeTemporal2002} have been lifted to hyperlogics, \HyperSL \cite{beutnerHyperStrategy2024} and \HyperATLstar \cite{beutnerHyperATLLogic2023}, to express hyperproperties of concurrent non-stochastic games, but without considering probabilistic behavior.
However, to the best of our knowledge, there does not exist any probabilistic hyperlogic for stochastic games yet.

In this paper, we propose the logic \HyperSGL for \textbf{Hyper}properties of \textbf{St}rategies in turn-based \textbf{St}ochastic games.
\HyperSGL can relate strategies and compare induced probabilities in the resulting probabilistic computation trees.
For example, we can 
relate several independent executions of a strategy for some player against different strategies of other players.
It is thus natural to formalize in \HyperSGL a notion of ``optimality'' for strategies, as illustrated by the following example.

\begin{example}
	\label{ex:optimal}
	Consider a two-player turn-based stochastic game with a unique initial state labeled with $\init$, where agent 1 wants to reach some set of states $T$ (labeled $t$). 
	We wish to express that agent 1 has a strategy that is optimal in the sense that for any fixed behavior of agent 2, it achieves the highest winning probability under all of its other strategies.  
	The below \HyperSGL formula compares two plays of the game, starting in states $\varstate$ and $\varstate'$, respectively, (``$\forall \varstate \forall \varstate'$'') that are required to be initial (``$\init_{\varstate} \wedge \init_{\varstate'}$'').
	The formula states that there exists a strategy for agent 1 in the play from $\varstate$ (``$\existsStrat{1}{\varstate}\!$''), such that under any strategy of agent 1 in the play from $\varstate'$ (``$\forallStrat{1}{\varstate'}\!$''), and any strategy that agent 2 follows in both plays (``$\forallStrat{2}{\varstate, \varstate'}\!$''), 
	the probability to reach $T$ in the first play (``$\Prob(\Finally t_{\varstate})$'') is at least as high as in the second play (``$\geq \Prob(\Finally t_{\varstate'})$''). 
	\[
	\forall \varstate \
	\forall \varstate' \
	\existsStrat{1}{\varstate} 
	\forallStrat{1}{\varstate'}
	\forallStrat{2}{\varstate, \varstate'}
	(\init_{\varstate} \wedge \init_{\varstate'})
	\implies 
	\Prob(\Finally t_{\varstate}) \geq \Prob(\Finally t_{\varstate'})
	\]	
\end{example}

As another example, in \HyperSGL we can also express the existence of a Nash equilibrium, i.e., the existence of a strategy for each agent such that all strategies together are optimal in the sense that no agent can improve their outcome by unilaterally deviating from it.
In \cref{sec:logic}, we give more examples of relevant properties that can be specified in \HyperSGL.

\paragraph*{Related Work.}
Alternating-time Temporal Logic (\ATLstar) and its fragment \ATL~\cite{alurAlternatingtimeTemporal2002} are temporal logics for concurrent non-stochastic games, extending \CTLstar and \CTL~\cite{emersonSometimesNot1986} to multiple agents.
Game logic (\GL)~\cite{alurAlternatingtimeTemporal2002} generalizes \ATLstar by separating path and strategy quantification. \GL subsumes \ATLstar.
Strategy Logic (\SL)~\cite{chatterjeeStrategyLogic2007} allows to explicitly quantify over strategies of turn-based non-stochastic games. 
\SL subsumes \ATLstar and \GL on turn-based non-stochastic games.

Stochastic Game Logic (\SGL)~\cite{baierStochasticGame2012} 
is a probabilistic variant of \ATLstar for turn-based stochastic games, 
and employs deterministic Rabin automata to reason about the probability that some $\omega$-regular objective is satisfied.
Another temporal logic for turn-based stochastic games is \rPATL \cite{chenAutomaticVerification2013}. It allows to specify probabilistic or reward-based zero-sum objectives by combining \ATL~\cite{alurAlternatingtimeTemporal2002}, \PCTL~\cite{hanssonLogicReasoning1994} and the reward operator. 
\rPATL has been extended to concurrent stochastic games~\cite{kwiatkowskaAutomatedVerification2018}, and to non-zero-sum objectives \cite{kwiatkowskaEquilibriabasedProbabilistic2019}, namely subgame perfect social welfare optimal Nash equilibria.

\HyperATLstar \cite{beutnerHyperATLLogic2023} is a hyperlogic for strategies in concurrent non-stochastic games combining \ATLstar and \HyperCTLstar \cite{clarksonTemporalLogics2014}.
\HyperATLstar is subsumed by 
\HyperSL \cite{beutnerHyperStrategy2024}, which is an extension of \SL and is also defined for concurrent non-stochastic games.
Like \SL, \HyperSL allows explicit first-order quantification over strategies. 
Strategy variables are not associated with a specific agent when they are quantified, and may be assigned to several different agents in a strategy profile.
The outcomes of strategy profiles can then be bound to named path variables and compared.

\HyperPCTL \cite{dobeModelChecking2022} and \PHL \cite{dimitrovaProbabilisticHyperproperties2020} are probabilistic hyperlogics for MDPs, which can be viewed as single-player stochastic games.
\HyperPCTL over MDPs extends \PCTL with quantification over states and schedulers.
\PHL extends \HyperCTLstar with the probability operator and quantification over schedulers. 
\HyperPCTLstar \cite{wangStatisticalModel2021} is defined only over DTMCs and extends \PCTLstar \cite{hanssonLogicReasoning1994} with quantification over paths.

\paragraph*{Contributions.} 
In this work, we propose the logic \HyperSGL to express hyperproperties for strategies in turn-based stochastic games.
%
\HyperSGL combines elements of \HyperPCTL \cite{dobeModelChecking2022} and \SGL \cite{baierStochasticGame2012}. 
Like \HyperPCTL, \HyperSGL extends \PCTL by quantification over game executions (through quantification over the game's starting states) and quantification over strategies. 
However, \HyperSGL allows a more flexible strategy quantification structure than \HyperPCTL, adapted to the setting of stochastic games.
\HyperSGL restricts quantification over states to the beginning of the formula, but, like in \SGL, quantification over strategies may be nested and thus may also occur inside probability operators. 
The relationship between state and strategy variables mirrors the \HyperPCTL setting: Several state variables may be associated with the same strategy, but at every point in time, each state variable is associated with at most one strategy for every agent.

We investigate the expressiveness of \HyperSGL by comparing it to 
temporal logics over turn-based stochastic games, 
as well as hyperlogics for stochastic systems or non-stochastic games.
There is a close relationship between \HyperSGL and logics for games: 
\HyperSGL embeds the fragment of \rPATL consisting of those formulas that do not contain reward-based objectives or Nash equilibria nested inside strategy quantifiers,
and subsumes the fragment of \HyperSL where each strategy variable is associated with a unique agent, thus also subsuming \HyperATLstar.

While the \HyperSGL model-checking problem over memoryful, probabilistic strategies is undecidable,
we show that the model-checking problem over strategies with bounded memory is decidable, and in particular that it is \PSPACE-hard and in \EXPTIME over memoryless deterministic strategies. 
For a fixed number of state quantifiers, we even prove the model-checking problem over memoryless deterministic strategies to be \PSPACE-complete.

\paragraph*{Organization.}
We recall preliminary concepts in \cref{sec:prelim},
before we present in \cref{sec:logic} the syntax and semantics of \HyperSGL and illustrate its expressiveness with several examples.
In \cref{sec:comp-hyper,sec:comp-stoch-games} we investigate the relationship between \HyperSGL and different logics for stochastic games, as well as hyperlogics.
We present our results on the complexity of the \HyperSGL model-checking problem in \cref{sec:mc}.
We conclude in \cref{sec:conclusion} and give an outlook on future work.
\iftoggle{extended}{}{For details on the syntax and semantics of the compared logics, model-checking algorithms, and proofs we refer to the extended version~\cite{extendedVersion}.}

\section{Preliminaries}
\label{sec:prelim}

Let $\N$ denote the natural numbers (including 0) and $\R$ the reals. 
For a set $S$, and some subset $R \subseteq S$, let $\overline{R}=S\setminus R$ denote the complement of $R$ in $S$.

\begin{definition} 
	A \emph{turn-based (stochastic) game} (TSG) is a tuple $\tbg = \tbgtup$ where 
	(1) $\agents$ is a finite set of \emph{agents},
	(2) $\states$ is a countable set of \emph{states} partitioned into $\states = \bigcup_{\agent \in \agents} \states_\agent$ with $\states_\agent\cap\states_{\agent'}=\emptyset$ for all
        $\agent\not=\agent'$,
	(3) $\Act$ is a finite set of \emph{actions},
	(4) $\Trans \colon \states \times \Act \times \states \to [0,1]$ is a \emph{transition probability function} such that for all $\state \in \states$ the set
		$
		\Act(\state)= \big\{\action\in\Act \mid \sum_{\state' \in \states} \Trans(\state, 
		\action, \state') =1\big\}
		$
		of \emph{enabled actions}
		is non-empty and 
		$\sum_{\state' \in \states} \Trans(\state, \action, \state') = 0$ for all $\action \in \overline{\Act(s)}$,
	(5)~$\AP$ is a set of \emph{atomic propositions},
	and 
	(6) $L \colon S \to 2^{\AP}$ is a \emph{labeling function}.
\end{definition}

In the following, we always assume that $\agents = \{1, \ldots, k\}$ for some $k \in \N$ and use $\states_{\setofagents} = \bigcup_{\agent \in \setofagents} \states_\agent$ for $\setofagents \subseteq \agents$.
We call a game \emph{non-stochastic} if all transition probabilities are 0 or 1.
An \emph{(infinite) path} of a turn-based game $\tbg$ is a sequence of states $\pi = \state_0\state_1\state_2\ldots \in \states^\omega$ such that for all $i \geq 0$ there exists $\action \in \Act$ with $\Trans(\state_i, \action, \state_{i+1}) > 0$.
A \emph{strategy} resolves the choices of one or several of the agents.

\begin{definition} 
	Let $\tbg = \tbgtup$ be a TSG and $\setofagents \subseteq \agents$ a set of agents.
	An \emph{$\setofagents$-strategy} for $\tbg$ is a tuple $\strat = \schedtup$ where
	(1)~$\modes$ is a countable non-empty set of \emph{modes},
	(2)~$\start \colon \states \to \modes$ selects a \emph{starting~mode} for each state,
	(3)~$\modef \colon \modes \times \states \to \modes$ is a \emph{mode transition function},~and
	(4)~$\act \colon \modes \times \partelt[{\setofagents}]  \times \Act \to [0,1]$ is a probabilistic \emph{action selection function} with $\sum_{\action \in \Act(\state)} \act(\mode, \state, \action) = 1$ and $\sum_{\action \in \overline{\Act(\state)}} \act(\mode, \state, \action) = 0$ for all $\mode \in \modes$, $\state \in \partelt[{\setofagents}]$.
\end{definition}

Let $\Strats$ denote the set of all $\setofagents$-strategies for $\tbg$; we might omit the index $\tbg$ if clear from the context.
We call an $\agents$-strategy $\strat \in \Strats[][\agents]$ a \emph{joint} strategy. 
A strategy is  
\emph{finite-memory} if $\modes$ is finite, 
\emph{$k$-memory} if $|\modes| \leq k$ for $k \in \N$,
\emph{memoryless} if $\modes$ is a singleton, and 
\emph{deterministic} if $\act(\mode, \state, \action) \in \{0, 1\}$ for all $(\mode, \state, \action) \in \modes \times \partelt[{\setofagents}] \times \Act$. 
If a strategy is memoryless, we often neglect its mode.

For an $\setofagents$-strategy $\strat$, and some $\setofagentstwo\subset \setofagents$ (or $i \in \setofagents$), we use $\strat|_{\setofagentstwo}$ (or $\strat|_i$) to denote the strategy $\alpha$ in which the action selection is restricted to the agents $\setofagentstwo$ (or agent $i$).
For $\setofagents, \setofagentstwo \subseteq \agents$, and $\strat \in \Strats$, $\strattwo \in \Strats[][\setofagentstwo]$, we let $\strat \leftarrow \strattwo$ be the $(\setofagents \cup \setofagentstwo)$-strategy obtained by taking the product of $\strat$ and $\strattwo$, in which the action selection for the agents in $\setofagentstwo$ is done according to $\strattwo$, and for the remaining agents in $A$ according to $\alpha$, as formalized in the definition below. 
If $\setofagents$ and $\setofagentstwo$ are disjoint, then $\strat \leftarrow \strattwo = \strattwo \leftarrow \strat$, and we also use $\strat \oplus \strattwo$ to denote this strategy.

\begin{definition} 
	Let $\tbg = \tbgtup$ be a TSG, and ${\strat\in\Strats[][\setofagents]}$, $\strattwo~\in~\Strats[][\setofagentstwo]$ strategies for sets of agents $\setofagents, \setofagentstwo \subseteq \agents$.
	The \emph{update of $\strat$ with $\strattwo$} is the $(\setofagents \cup \setofagentstwo)$-strategy 
	$\strat \leftarrow \strattwo = \schedtup[\superalphabeta]$ where
	\begin{itemize}
		\item $\superalphabeta[\modes] = \superalpha[\modes] \times \superbeta[\modes]$,
		
		\item $\superalphabeta[\init](\state) = (\superalpha[\init](\state), \superbeta[\init](\state))$ for $\state \in \states$,
	
		\item $\superalphabeta[\modef]((\superalpha[\mode], \superbeta[\mode]), \state) = (\superalpha[\modef](\superalpha[\mode], \state), \superbeta[\modef](\superbeta[\mode], \state))$ 
		for $\superalpha[\mode] \in \superalpha[\modes]$, $\superbeta[\mode] \in \superbeta[\modes]$, $\state \in \states$, 
		and
		
		\item $\superalphabeta[\act]((\superalpha[\mode], \superbeta[\mode]), \state) = \begin{cases}
			\superbeta[\act](\superbeta[\mode], \state) & \text{if } \state \in \partelt[\setofagentstwo] \\
			\superalpha[\act](\superalpha[\mode], \state) & \text{else } 
		\end{cases}$ 
		\\for $\superalpha[\mode] \in \superalpha[\modes]$, $\superbeta[\mode] \in \superbeta[\modes]$, $\state \in \states_{A \cup B}$.
	\end{itemize}
\end{definition}

Applying a joint strategy to a turn-based game induces a \emph{discrete-time Markov chain}, 
i.e., a stochastic game with a single agent and without nondeterminism.

\begin{definition} 
	A \emph{discrete-time Markov chain (DTMC)} is
	a tuple $\dtmc= \dtmctup$
	where 
	(1) $S$ is a finite  set of \emph{states},
	(2) $\Trans \colon \states \times \states \to [0,1]$ is a \emph{transition probability function} with $\sum_{\state' \in \states} \Trans(\state,\state') =1$ for all $\state \in \states$,
	(3) $\AP$ is a set of \emph{atomic propositions},
	and
	(4) $L \colon \states \to 2^{\AP}$ is a \emph{labeling function}.
\end{definition}

\begin{definition} 
	For a TSG $\tbg= \tbgtup$ and a joint strategy $\strat = \schedtup\in \Strats[][\agents]$  for $\tbg$, we define the \emph{DTMC induced by $\strat$ in $\tbg$} as $\superalpha[\tbg] = \dtmctup[\superalpha]$ with the following components:
	\begin{itemize}
		\item $\superalpha[\states] = \states \times \modes$,
		
		\item 
		$\superalpha[\Trans]((\state, \mode), (\state', \mode')) = \sum_{\action \in \Act(\state)} \Trans(\state, \action, \state') \cdot \act(\state, \mode, \action) \text{ if } \mode' = \modef(\mode, \state)$ and 0 otherwise,
		for $\state, \state' \in \states$, and $\mode, \mode' \in \modes$,
		
		\item $\superalpha[\AP] = \AP$, and
		
		\item $\superalpha[\labelingfct](\state, \mode) = \labelingfct(\state)$ for $\state \in \states$ and $\mode \in \modes$.
	\end{itemize}
\end{definition}

\begin{definition}
	For $n\in\N$ and DTMCs $\dtmc_1, \ldots, \dtmc_n$ with $\indexi[\dtmc]=\dtmctup[\indexi]$ for $i=1,\ldots,n$, we define their \emph{composition} $\dtmc_1\times\ldots\times\dtmc_n$ as the DTMC $\dtmc=\dtmctup$ with
	(1) $\states=\states_1\times\ldots\times\states_n$, 
	(2) $\Trans((\state_1, \ldots, \state_n),(\state_1',\ldots, \state_n'))=\prod_{i=1}^n \indexi[\Trans](\indexi[\state], \indexi[\state]')$,
	(3) $\AP=\bigcup_{i=1}^{n}\{\indexi[\ap] \mid \ap \in\indexi[\AP] \}$, and
	(4) $\labelingfct(\state_1,\ldots,\state_n)=\bigcup_{i=1}^n\{\indexi[\ap] \mid \ap \in \indexi[\labelingfct](\indexi[\state])\}$
	for $\indexi[\state],\indexi[\state]' \in \indexi[\states]$ for $i=1,\ldots,n$.
\end{definition}

An \emph{(infinite) path} of a DTMC $\dtmc=\dtmctup$ is a 
sequence of states $\pi = \state_0\state_1\state_2\ldots \in \states^\omega$ with $\Trans(\state_i, \state_{i+1}) > 0$ for all $i \geq 0$. 
Let $\Paths[\dtmc]$ denote the set of all paths of $\dtmc$, and $\Pathsfrom[\dtmc]$ those starting in $\state \in \states$.
For a path $\pi = \state_0 \state_1 \state_2 \ldots$ and $j \in \N$ we use $\pi[j]$ to denote $\state_j$ and $\pi[0, j]$ to denote the finite prefix $\state_0\ldots s_j$. 
A finite path is a finite prefix of an infinite path $\pi$.

For $\state \in \states$, let ${\Pr}^{\dtmc}_{\state}$ denote the probability measure associated with $\dtmc$ and $\state$.
An \emph{objective} for a turn-based game $\tbg$ is a measurable function $X \colon \Paths \to \{0, 1\}$.
Let $\Pr^{\tbg, \alpha}_{\state}(X)$ denote the probability of satisfying an objective $X$ over all paths in $\tbg$ starting in state $\state$ under the joint strategy $\strat$, i.e., $\Pr^{\tbg^\alpha}_{(\state, \start(\state))}(\{ \pi \in \Pathsfrom[\tbg^\strat][(\state, \start(\state))] \mid X(\pi) = 1 \})$.
We refer to \cite{baierPrinciplesModel2008} for a detailed discussion of probability measures.
An \emph{objective profile} for two opposing coalitions $C, \overline{C} \subseteq \agents$ 
is a pair of objectives $(\objone, \objtwo)$.
A \emph{Nash equilibrium} for an objective profile is a joint strategy such that neither coalition can improve their expected outcome by unilaterally deviating from this strategy.
\begin{definition} 
	\label{def:NE}
	Let $\tbg= \tbgtup$ be a TSG, 
	$\state \in \states$,
	$C \subseteq \agents$, 
	and $\objone, \objtwo\colon \Paths \to \{0, 1\}$.
	A joint strategy $\nasheq\in \Strats[][\agents]$ is a \emph{Nash equilibrium (NE)} for $(\objone, \objtwo)$ from $\state$ if
	\begin{itemize}
		\item for all $\gamma \in \Strats[][C]$ we have 
		$\Pr^{\tbg, \nasheq}_{\state}(\objone) \geq \Pr^{\tbg, \alpha}_{\state}(\objone)$ where $\alpha := \nasheq {\leftarrow} \gamma$, and
		
		\item for all $\gamma \in \Strats[][\overline{C}]$ we have 
		  $\Pr^{\tbg, \nasheq}_{\state}(\objtwo) \geq \Pr^{\tbg, \alpha}_{\state}(\objtwo)$ where $\alpha := \nasheq {\leftarrow} \gamma$.
	\end{itemize}
\end{definition}

\section{The Hyperlogic \HyperSGL} 
\label{sec:logic}

We propose the hyperlogic \HyperSGL for strategies in turn-based stochastic games that combines elements of \HyperPCTL and \SGL.
The general structure of the logic is very similar to \HyperPCTL.
However, being designed for one-player games only, \HyperPCTL requires to first quantify over schedulers and then over states, and quantification is restricted to the beginning of a formula.
For multi-player games, a more flexible quantification structure seems more appropriate.
Concretely, \HyperSGL first quantifies over states and then over strategies, and does not restrict strategy quantification to the beginning of the formula but instead allows nested strategy quantification, like \SGL.

\subsection{Syntax}
\label{sec:syntax} 
Our goal is to reason about several independent plays of a game, which can be viewed as probabilistic computation trees from some state under some joint strategy.
\HyperSGL enables this by first quantifying over state variables in order to fix the root states of the computation trees, and then fixing a joint strategy for each play by quantifying over strategies and associating them with the state variables.
Thus, each state variable is associated with a play of the game.

Let $\VarsState$ be an infinite set of state variables, $\agents$ some set of agents, and $\AP$ a set of atomic propositions. 
\HyperSGL formulas are defined by the following grammar:
\[ \arraycolsep=1pt
\begin{array}{lll}
	\text{state-quant.:} 
	&\pstateq &::= 
	\exists \varstate\ \pstateq \mid 
	\forall \varstate\ \pstateq \mid 
	\pschedq
	\\
	\text{strategy-quant.: } 
	& \pschedq &::= 
	\existsStrat{A}{\setofvariables} \pschedq \mid 
	\forallStrat{\setofagents}{\setofvariables} \pschedq \mid 
	\pnonquant  
	\\
	\text{non-quant.:} 
	&\pnonquant &::= 
	\true \mid
	\ap_{\varstate} \mid 
	\pnonquant \wedge \pnonquant \mid 
	\neg \pnonquant \mid 
	\pprob \sim \pprob \mid
	\varstate = \varstate' \mid
	\pschedq
	\\
	\text{probability expr.: }  
	&\pprob &::= 
	\Prob(\ppath) \mid 
	f(\pprob, \ldots, \pprob) 
	\\
	\text{path formula:}  
	&\ppath &::= 
	\Next \pnonquant \mid 
	\pnonquant \Until \pnonquant 
\end{array}
\]
where $\varstate, \varstate' \in \VarsState$, 
$\setofagents \subseteq \agents$, 
$\setofvariables \subseteq \VarsState$, 
$\ap \in \AP$, 
$\sim \in \{\leq, <, =, >, \geq\}$, 
and
$f \colon [0,1]^j \to \R$ is a $j$-ary standard arithmetic operation like addition, subtraction, multiplication, and constants.
We sometimes simplify notation by writing, e.g., $\existsStrat{1,2}{\varstate} \phi$ instead of $\existsStrat{1,2}{\{\varstate\}} \phi$. 
We use standard syntactic sugar like 
$\phi_1 \vee \phi_2 := \neg(\neg\phi_1 \wedge \neg \phi_2)$, 
$\phi_1 \implies \phi_2 := \neg\phi_1 \vee \phi_2$, 
$\Prob(\Finally \phi) := \Prob(\true \Until \phi)$, 
$\Prob(\Globally \phi) := 1 - \Prob(\Finally \neg\phi)$. 
Extending \HyperSGL by adding direct nesting of temporal operators would be straightforward, but would increase the model-checking complexity, analogously to the increased model-checking complexity of \PCTLstar compared to \PCTL.

A state-quantified \HyperSGL formula $\pstateq$ is \emph{well-formed} if 
(C1)
	each occurrence of $\ap_{\varstate}$ in $\pstateq$ is in the scope of a quantifier for $\varstate$,
(C2)
	any quantifier $\existsStrat{\setofagents}{\{\varstate_{1}, \ldots, \varstate_{m}\}}$ is in the scope of quantifiers for $\varstate_{1}$, \ldots, $\varstate_{m}$, and
(C3)
	each probability expression $\Prob(\ppath)$ is in the scope of strategy quantifiers for all agents for all state variables\footnote{It would be sufficient to require this for all agents for all state variables \emph{occurring in $\ppath$}. We make this stronger requirement to simplify notation.}.

The following examples illustrate different kinds of properties expressible in \HyperSGL. 
Whenever we refer to some set of states $T$,
we assume that the states contained in that set are labeled with a corresponding lower-case atomic proposition $t$. 

\begin{example}
	\label{ex:almost-sure}
	The \HyperSGL formula 
	$
	\forall \varstate \
	\existsStrat{1}{\varstate} 
	\forallStrat{2}{\varstate}
	\Prob(\Finally t_{\varstate}) = 1
	$
	expresses that agent 1 has an almost-sure winning strategy to reach a set of target states $T$ against agent 2.
	This example does not utilize the full power of \HyperSGL as a hyperlogic, since it only employs a single state quantifier.
\end{example}

\begin{example}
	\label{ex:nested}	
	Assume agent 1 wants to reach one of two target sets $R$ and $T$, and prefers $R$ over $T$. Reaching $R$ is a possible way for agent 2 to fulfill its objective but there might be other options for agent 2. The following formula expresses that 
	the two agents can collaborate for reaching $R$ almost surely 
	in such a way that 
    (1) the strategy of agent 1 maximizes the probability of reaching $R$ 
    among all strategies for agent 1 (``$\forallStrat{1}{\varstate'}\!$'') and 
    against every strategy of agent 2 (``$\forallStrat{2}{\varstate,\varstate'}\!$''),
	and
	(2) 
	agent 1 can play safe in the sense that it can ensure to almost surely reach a state that is either in $R$, or from which agent 1 has a strategy to almost surely reach $T$, no matter how agent 2 behaves.
	\begin{align*}
		&\forall \varstate \
		\forall \varstate' \
		(\init_{\varstate} \wedge \init_{\varstate'}) 
		\implies
                \existsStrat{1, 2}{\varstate, \varstate'} 
		\Big[ \Prob( \Finally r_{\varstate}) = 1 
		\land {}
		\\ & \quad
		\underbrace{
			\forallStrat{1}{\varstate'}
			\forallStrat{2}{\varstate,\varstate'} \Prob(\Finally r_{\varstate}) \ge \Prob(\Finally r_{\varstate'})
		}_{\footnotesize (1)}
		\land 
		\underbrace{
			\forallStrat{2}{\varstate} \Prob\big[\Finally\big( r_{\varstate} \lor \existsStrat{1}{\varstate}\forallStrat{2}{\varstate} \Prob(\Finally t_{\varstate})= 1\big)\big] {=} 1 
		}_{\footnotesize (2)} \Big]
	\end{align*}
\end{example}

\begin{example}
	\label{ex:eq}
	Consider a turn-based \emph{non-stochastic} game with two agents where the agents have the same objectives as in \cref{ex:nested}.
	The following formula expresses that 
	the two agents can collaborate for reaching $R$ almost surely,
	and at the same time agent 1 can almost-surely ensure to almost-surely reach $T$ in case agent 2 deviates from their joint strategy: 
	\begin{align*}
		&\forall \varstate \ \forall \varstate' .\
		(\init_{\varstate} \wedge \init_{\varstate'}) \implies
		\existsStrat{1,2}{\varstate, \varstate'} 
		\\ & \qquad\qquad 
		\Big( \;
		\Prob(\Finally r_{\varstate}) = 1 \land 
		\forallStrat{2}{\varstate'} 
		\Prob\big[ \Globally \big(\neg (\varstate = \varstate') \implies \existsStrat{1}{\varstate'} \Prob(\Finally t_{\varstate'}) = 1 \big)\big]=1 \Big) .
	\end{align*}
\end{example}

\begin{example}
	\label{ex:nash}
	Assume agents 1 and 2 want to reach goals $T^1$ and $T^2$, respectively. 
	The following \HyperSGL formula expresses that there exists a Nash equilibrium from some fixed initial state labeled with $\init$,
	i.e., a joint strategy such that, if only one of the agents changes its strategy, then the probability of reaching its respective goal will not increase.
	\begin{align*}
		\forall \varstate \
		\forall \varstate' .\
		(\init_{\varstate} \wedge \init_{\varstate'}) 
		\implies \existsStrat{1,2}{\varstate, \varstate'} 
		&\big(\forallStrat{1}{\varstate'} \Prob(\Finally t^1_{\varstate}) \geq \Prob(\Finally t^1_{\varstate'}) \big)
		\\
		& 
		\land {} 
		\big(\forallStrat{2}{\varstate'} \Prob(\Finally t^2_{\varstate}) \geq \Prob(\Finally t^2_{\varstate'})\big)
	\end{align*}
\end{example}

\subsection{Semantics}
\label{sec:semantics}

Without loss of generality, let us now assume that all well-formed \HyperSGL formulas are of the form $Q \varstate_1 \ldots Q \varstate_n .\ \psi$ for some $n \in \N$ and some strategy-quantified \HyperSGL formula $\psi$.
\HyperSGL formulas are evaluated over a context $\Ctx = \Ctxtup$ consisting of
\begin{itemize}
	\item a turn-based game $\tbg = \tbgtup$,

	\item a partial mapping $\mapStrat \colon \VarsState \times \agents \rightharpoonup \bigcup_{\agent \in \agents} \Strats[][\{\agent\}]$ assigning strategies to combinations of state variables and agents 
	s.t.\ $\mapStrat(\varstate, \agent) \in \Strats[][\{\agent\}]$, and

	\item a partial mapping
	$\statetupMap \colon \VarsState \rightharpoonup \states$ assigning states to state variables. 
\end{itemize}

Let $\tbg$ be a turn-based game, $\mapStrat \colon \VarsState \times \agents \rightharpoonup \bigcup_{\agent \in \agents} \Strats[][\{g\}]$ a strategy mapping, $\statetupMap \colon \VarsState \rightharpoonup \states$ a state mapping, and $\Ctx = \Ctxtup$. 
$\tbg$ satisfies a well-formed \HyperSGL formula $\phi$, written $\tbg \models \phi$, iff $(\tbg, \emptymap, \emptymap) \models \phi$.
We define the semantics of well-formed \HyperSGL formulas by structural induction.
State variable quantification is evaluated by quantifying over all states of the game structure. 
A strategy quantifier $\existsStrat{\setofagents}{\setofvariables}\!$ or $\forallStrat{\setofagents}{\setofvariables}\!$ is evaluated by quantifying over all $\setofagents$-strategies and assigning the quantified strategy to all state variables $\varstate \in \setofvariables$.
For each state variable $\varstate$, the current assignments of the mappings thus encode a play of the game as the computation tree from $\statetupMap(\varstate)$ under $\bigoplus_{\agent \in \agents} \mapStrat(\varstate, \agent)$.
Formally, the semantics of state- and strategy-quantified formulas are as follows.
\[ 
\begin{array}{lll}
	\Ctx \models \mathbb{Q} \varstate_i \ \pstateq & \iff 
	& \mathbb{Q} \state_i \in \states .\ 
	\Ctxtup[][][\statetupMap{[{\varstate_i \mapsto \state_i}]}] \models \pstateq  \\
	%
	\Ctx \models \existsStrat{\setofagents}{\setofvariables} \pschedq & \iff 
	& \exists \strat \in \Strats[][\setofagents] .\ 
	\Ctxtup[][{\mapStrat[(\varstate, g) \mapsto \strat|_g \text{ for } g {\in} \setofagents, \varstate {\in} \setofvariables]}][]
	\models \pschedq \\
	\Ctx \models \forallStrat{\setofagents}{\setofvariables} \pschedq & \iff 
	& \forall \strat \in \Strats[][\setofagents] .\ 
	\Ctxtup[][{\mapStrat[(\varstate, g) \mapsto \strat|_g \text{ for } g {\in} \setofagents, \varstate {\in} \setofvariables]}][]
	\models \pschedq \\
\end{array}
\]
where 
$\mathbb{Q} \in \{\exists, \forall\}$, 
$i\in\{1, \ldots, n\}$,
$\setofagents \subseteq \agents$, 
$\setofvariables \subseteq \{\varstate_1, \ldots, \varstate_n\}$.

The semantics of non-quantified formulas is defined as follows.
\[ 
\begin{array}{lll}
	\Ctx \models \true & &  \\
	\Ctx \models \ap_{\varstate_i} & \iff 
	& \ap \in \labelingfct(\statetupMap(\varstate_i)) \\
	\Ctx \models \pnonquant_1 \wedge \pnonquant_2 & \iff & \Ctx \models \pnonquant_1 \text{ and } \Ctx \models \pnonquant_2 \\
	\Ctx \models \neg \pnonquant & \iff & \Ctx \not\models \pnonquant \\
	\Ctx \models \pprob_1 \sim \pprob_2 & \iff & 
	\llbracket \pprob_1 \rrbracket_{\Ctx} \sim \llbracket \pprob_2 \rrbracket_{\Ctx} \\
	\Ctx \models \varstate_i = \varstate_j & \iff & 
	\statetupMap(\varstate_i) = \statetupMap(\varstate_j)
\end{array}
\]
where 
$i,j\in\{1, \ldots, n\}$, 
and
$\ap \in \AP$.

In order to evaluate a probability expression, we build the composition of the DTMCs induced by the strategies fixed in $\mapStrat$, and evaluate path probabilities in the composed DTMC as usual in \PCTL.
Recall that conditions (C1)--(C3) assure that 
whenever we evaluate well-formed \HyperSGL formulas, 
probability expressions are evaluated in a context satisfying $\Dom(\statetupMap)=\{\varstate_1,\ldots,\varstate_n\}$ and
$\Dom(\mapStrat)=\{\varstate_1,\ldots, \varstate_n\}\times \agents$.
Formally, we let $\indexi[\strat] = \strattup[\indexi] := \bigoplus_{\agent \in \agents} \mapStrat(\indexi[\varstate], g)$ for $i \in \{1,\ldots,n\}$, and $\dtmc := \tbg^{\strat_1}\times\ldots\times\tbg^{\strat_n}$ with state set $\states^{\dtmc}$, and define the semantics of probability expressions as follows.
\[
\begin{array}{lcl}
	\llbracket \Prob(\ppath) \rrbracket_{\Ctx} & = &
	\Pr^{\dtmc}_{r}(\pi \in \Pathsfrom[\dtmc][r] \mid \Ctx, \pi \models \ppath ) 
	\\
	\llbracket f(\pprob_1, \ldots, \pprob_\ell) \rrbracket_{\Ctx} & = & f(\llbracket \pprob_1 \rrbracket_{\Ctx}, \ldots, \llbracket \pprob_\ell \rrbracket_{\Ctx}) 
\end{array} 
\]
where 
$\ell \in \N$
and
$r := \left((\statetupMap(\varstate_1), \start_1(\statetupMap(\varstate_1))), \ldots, (\statetupMap(\varstate_n), \start_n(\statetupMap(\varstate_n)))\right) \in \states^{\dtmc}$ is the state corresponding to the current state variable assignment $\statetupMap$ in $\dtmc$.

Temporal operators intuitively allow to make steps in the quantified plays, moving to some sub-tree in each computation tree. 
Formally, the temporal operators are evaluated over paths in the composed induced DTMC $\dtmc$.
In order to evaluate a subformula at some position on a composed path, we
update the state mapping 
and 
shift the strategy mapping such that it remembers the history of the game that was played until that position:
For $\pi \in \Paths[\dtmc]$ and $j \in \N$, we use $\mapStrat^{\pi[0, j]}$ for the update of $\mapStrat$ by $\pi[0, j]$ defined as follows. 
For $i \in \{1,\ldots,n\}$ and $\agent \in \agents$, we let
$(s_{i}, q_{i}) := (\pi[j])_i \in \states^{\tbg^{\indexi[\alpha]}}$
and 
$\mapStrat^{\pi[0, j]}(\indexi[\varstate], g) := \indexi[\strat]^{\pi[0, j]} := (\indexi[\modes], \indexi[\start]^{\indexi[q]}, \indexi[\modef], \indexi[\act])$ where 
	$\indexi[\strat] = (\indexi[\modes], \indexi[\start], \indexi[\modef], \indexi[\act])$
as above and $\indexi[\start]^{\indexi[q]}(\state) = \indexi[q]$ for all $\state \in \states$. 
The semantics of path formulas is then defined as follows, where we use $\projS{\pi[j]}$ 
to denote the updated state variable mapping $\statetupMap'$ with $\statetupMap'(\varstate_i) = \projS{((\pi[j])_i)}$, the projection of $(\pi[j])_i$ to the $\states$-component, for $i = 1, \ldots, n$.
\[
\begin{array}{lcl}
	\Ctx, \pi \models \Next \pnonquant & \iff & \Ctxtup[][\mapStrat^{\pi[0, 1]}][\projS{\pi[1]}] \models \pnonquant \\
	\Ctx, \pi \models \pnonquant_1 \Until \pnonquant_2 & \iff & 
	\exists j_2 \geq 0 .\ \big(\Ctxtup[][\mapStrat^{\pi[0, j_2]}][\projS{\pi[j_2]}] \models \pnonquant_2 \wedge \\
	& & \phantom{E j_2 \geq 0 . }\;\;\; \forall 0 \leq j_1 < j_2 .\ \Ctxtup[][\mapStrat^{\pi[0, j_1]}][\projS{\pi[j_1]}] \models \pnonquant_1 \big)
	\\
\end{array}
\]
\noindent 
The semantics of the probability operator is well-defined as the set $\{\pi \in \Pathsfrom[\dtmc][r] \mid \Ctx, \pi \models \ppath \}$ is measurable with analogous reasoning to \PCTL.  

\section{Relationship to Logics for Stochastic Games}
\label{sec:comp-stoch-games}

Stochastic game logic (\SGL) and probabilistic alternating-time temporal logic (\PATL) are temporal logics for stochastic games. 
Both are evaluated over a fixed initial state $\sinit$ of a game structure, while \HyperSGL is evaluated over structures without dedicated initial states.
Therefore, we can only mimic \SGL and \PATL in \HyperSGL if we can assume that the given state $\sinit$ has a unique label, e.g., $\init$. 
We can then specify that we are only interested in executions where some state variable is initialized with a state labeled with $\init$. 
In fact, if we access the label $\init$ for this purpose only, we do not increase the observability of the initial state compared to \SGL and \PATL.
\iftoggle{extended}{%
	\cref{app:rPATL,app:sgl-proof} give%
}{%
	The extended version contains%
} 
the full syntax and semantics of \SGL and \rPATL.
An overview of the relationships between \HyperSGL, \SGL and \rPATL is presented in \cref{fig:overview-stoch}.

\subsection{PATL and rPATL}
\label{sec:rpatl}

\PATL and \PATLstar \cite{chenProbabilisticAlternatingtime2007} are \PCTL-style (resp.\ \PCTLstar-style) extensions of \ATL \cite{alurAlternatingtimeTemporal2002}.
They allow to specify zero-sum properties of turn-based and concurrent stochastic games. 
\PATL can be directly embedded into \HyperSGL.\footnote{Technically, \PATL allows bounded Until but \HyperSGL does not. Extending \HyperSGL with bounded Until would be straightforward and not increase model-checking complexity.} 
Since \HyperSGL does not allow arbitrary nesting of temporal operators, simulating \PATLstar in \HyperSGL is not possible. 
\rPATL and \rPATLstar \cite{chenAutomaticVerification2013,kwiatkowskaAutomatedVerification2018} extend \PATL and \PATLstar with the ability to reason about expected reward. 
In \cite{kwiatkowskaEquilibriabasedProbabilistic2019}, \rPATL is further extended with the possibility to specify \emph{social welfare subgame perfect $\epsilon$-Nash equilibria (\swspepsne)}.
An \rPATL formula of the form $\existsStrat{C {:} \overline{C}}{max \sim x}(\texttt{P}[\psi^1] + \texttt{P}[\psi^2])$
expresses that the agents have a joint strategy $\nasheq$ such that (1) the sum of the probabilities of satisfying $\psi^1$ and $\psi^2$ relates to $x$ as specified, and (2), $\nasheq$ is a \swspepsne, where $\psi^1$ is the objective for $C$ and $\psi^2$ is the objective for $\overline{C}$.
A \emph{social welfare} Nash equilibrium maximizes the sum of the probability of satisfying the objectives.
A Nash equilibrium is considered to be \emph{subgame perfect} if it is a Nash equilibrium in all subgames, i.e., at each state of the game.
Let \PATLNE be the fragment of \rPATL with \swspepsne but without reward objectives or bounded Until.

In \cref{ex:nash} we have seen that \HyperSGL can express Nash equilibria. 
In order to express \swspepsne, we need to compare different strategies from the same state.
In \HyperSGL, a state variable cannot be associated with several different joint strategies at the same time.
However, we can model the desired behavior by quantifying over several state variables and specifying that we are only interested in cases where the state variables are assigned the same state, using \enquote{$\varstate = \varstate'$}.
Further, the definition of \swspepsne first quantifies over strategies and then over subgames, i.e., over states, while in \HyperSGL we quantify first over states and then over strategies. 
However, if there is no nested strategy quantification, moving the state quantification to the front results in an equivalent formulation. 

\begin{restatable}{theorem}{translateSWSPepsNE}
	\label{th:translateSWSPepsNE}
	For \PATLNE formulas $\psi^1, \psi^2$ without strategy quantifiers (including Nash equilibria), it holds that there exists a \HyperSGL formula $\phi'$ such that
	\begin{align*} 
		\tbg, \state \modelsrPATL \existsStrat{C {:} \overline{C}}{max \sim x}(\mathtt{P}[\psi^1] + \mathtt{P}[\psi^2]) 
		\ \iff \
		\tbg' \modelsHyperSGL \phi'
	\end{align*}
	for all turn-based stochastic games $\tbg$ and states $\state \in \states$, 
	where
	$\tbg'$ is the same as $\tbg$ but $\state$ is additionally labeled with a fresh atomic proposition $\init$.
\end{restatable}

\begin{proof}[Sketch]
	We construct
	\begin{align*}
		&\phi' := \forall \varfixed \forall \varfixedp \
		\forall \varopt \forall \varoptp \
		\exists \varcomp \exists \varcompp .\ 
		(\init_{\varfixed} \wedge \init_{\varfixedp} \wedge \varopt = \varoptp) 
		\implies
		(\varcomp = \varcompp \wedge \phi_{\textit{SWSP}\epsilon} )
		\\ \allowbreak
		&\phi_{\textit{SWSP}\epsilon}
		:= 
		\existsStrat{\agents}{\varfixed, \varopt, \varoptp} 
		\Big[ \Prob(\psi^1_{\varfixed}) + \Prob(\psi^2_{\varfixed}) \sim x
		\wedge 
		\phi_{\textit{SP}\epsilon}(\varopt, \varoptp)  
		\\ & \qquad \qquad \qquad \qquad \qquad \quad \;
		\wedge \forallStrat{\agents}{\varfixedp, \varcomp, \varcompp} 
		\left( \phi_{\textit{SP}\epsilon}(\varcomp, \varcompp) \implies \phi_{\textit{SW}}(\varfixed, \varfixedp) \right) \Big]
		\\ 
		&\phi_{\textit{SP}\epsilon}(\varstate_i, \varstate'_i) := 
		\big(\forallStrat{C}{\varstate'_i} \Prob(\psi^1_{\varstate_i}) \geq \Prob(\psi^1_{\varstate'_i}) - \epsilon \big)
		\wedge 
		\big(\forallStrat{\overline{C}}{\varstate'_i} \Prob(\psi^2_{\varstate_i}) \geq \Prob(\psi^2_{\varstate'_i}) - \epsilon \big)
		\\ \allowbreak
		&\phi_{\textit{SW}}(\varfixed, \varfixedp) := 
		\textstyle\sum_{i=1}^2  \Prob(\psi^i_{\varfixed})
		\geq
		\sum_{i=1}^2  \Prob(\psi^i_{\varfixedp}) 
	\end{align*}
	where $\psi^i_{\varstate}$ corresponds to $\psi^i$ with all atomic propositions indexed by state variable $\varstate$.
	Intuitively, $\phi_{\textit{SWSP}\epsilon}$ asks for a joint strategy $\nasheq$ such that 
	(1) the overall profit under $\nasheq$ from $\fixed$ relates to $x$ as specified,
	(2) $\nasheq$ is a \spepsne, and 
	(3) $\nasheq$ results in a higher overall profit than any other \spepsne $\schedcomp$ from $\fixed$ (social welfare).
	We associate state variables $\varfixed, \varopt, \varoptp$ with $\nasheq$ and state variables $\varfixedp, \varcomp, \varcompp$ with $\schedcomp$.
	We use 
	$\varfixed, \varfixedp$ to compare the overall profit from $\fixed$ under $\nasheq$ and $\schedcomp$.
	We require $\varstate_i$ and $\varstate_i'$ to be instantiated with the same state for $i=1,2$, and use 
	$\varopt, \varoptp$ to check whether $\nasheq$ is a \spepsne, and 
	$\varcomp, \varcompp$ to do the same for $\schedcomp$.
	The size of $\phi'$ is linear in the size of the \PATLNE formula $\existsStrat{C {:} \overline{C}}{max \sim x}(\texttt{P}[\psi^1] + \texttt{P}[\psi^2])$. 
	For the full proof, see \iftoggle{extended}{\cref{app:rPATL}}{the extended version}.
	\qed
\end{proof}

This idea can be extended to translate the fragment \PATLfrag of \PATLNE to \HyperSGL, consisting of those \PATLNE formulas where Nash equilibria are not nested inside strategy quantifiers (including Nash equilibria). 
Translating Boolean operators, atomic propositions, strategy quantifiers and probability operators is straightforward, and Nash equilibria can be translated as shown in \cref{th:translateSWSPepsNE}, using fresh state variables for every Nash equilibrium.
However, it is not obvious how to express formulas with Nash equilibria nested inside strategy quantifiers or Nash equilibria in \HyperSGL, since then state quantifiers are nested inside probability operators and cannot be pulled to the front, see \iftoggle{extended}{\cref{app:general-PATLNE}}{the extended version} for details.
Hence, it is unclear whether \PATLNE can be fully simulated in \HyperSGL.

\begin{openproblem}
	Does \HyperSGL subsume \PATLNE over turn-based stochastic games?
\end{openproblem}

\subsection{Stochastic Game Logic}
\label{sec:sgl}

\SGL \cite{baierStochasticGame2012} is a probabilistic variant of \ATL for turn-based stochastic games that allows to specify $\omega$-regular conditions on paths via \emph{deterministic Rabin automata (DRA)}.
\HyperSGL cannot fully subsume \SGL because \HyperSGL employs \LTL-style path conditions, which are strictly less expressive than DRA.

When evaluating a probability operator, \SGL implicitly quantifies over all memoryful probabilistic strategies for all agents for whom a strategy has not been fixed yet. 
In \HyperSGL, we enforce explicit strategy quantification for all agents 
in order to have a more transparent logic and simpler semantics.
In \SGL, if we restrict the strategies for the explicitly quantified agents to be, e.g., memoryless deterministic, the remaining agents are still allowed memory and randomization, 
while in \HyperSGL, we can only restrict the strategy class for all agents.

Let \SGLfrag be the fragment of \SGL where we restrict path conditions to Boolean and temporal operators and assume explicit strategy quantification. There still is a crucial difference between \SGLfrag and \HyperSGL in the evaluation of temporal operators:
In \HyperSGL, when we evaluate a temporal operator on a path under some strategy mapping, we evaluate the subformula(s) from a position on the path under the \emph{shifted} strategy mapping that remembers the history of the play until that position.
In \SGL, however, all subformula are evaluated under the \emph{same} strategy again, i.e., we do not remember the play so far. 
\iftoggle{extended}{See \cref{app:sgl-proof} for details and an example.}{We refer to the extended version for details and an example.}

\begin{conjecture}
	\SGLfrag cannot be embedded in \HyperSGL over TSGs.
\end{conjecture}
\section{Relationship to Hyperlogics}
\label{sec:comp-hyper}

In this section, we investigate the relationship between \HyperSGL and different hyperlogics for stochastic systems or non-stochastic games, namely \HyperPCTL, \PHL, and \HyperPCTLstar as well as \HyperSL and \HyperATLstar. 
An overview of the relationships between the logics is given in \cref{fig:overview-hyper}.

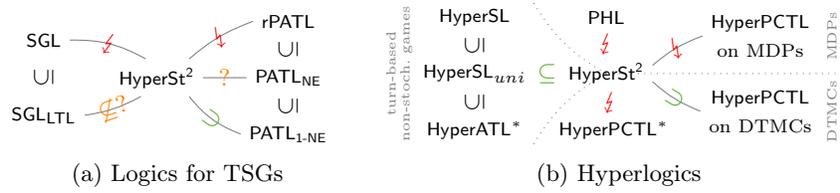
\begin{figure}[t]
	\subfloat[Logics for TSGs\label{fig:overview-stoch}]{%
		\begin{tikzpicture}[every node/.style={sloped}]	
	\node (HyperSt2) at (0,0) {\HyperSGL};
	
	\node (SGL) at (-1.5,0.5) {\SGL};
	\node (SGLLTL) at (-1.5,-0.5) {\SGLfrag};
	
	\path (SGLLTL) to node[midway] {$\subseteq$} (SGL);
	\draw[-, color=gray] (SGL.east) edge[bend left=15] node[text=red]{$\lightning$} (HyperSt2.120);
	\draw[-, color=gray] (SGLLTL.east) edge[bend right=10] node[text=orange]{$\not\subseteq$?} (HyperSt2.230);
	
	\node (rPATL) at (1.75,0.75) {\rPATL};
	\node (PATLNE) at (1.75,0) {\PATLNE};
	\node (PATL1NE) at (1.75,-0.75) {\PATLfrag};
	
	\path (PATL1NE) to node[midway] {$\subseteq$} (PATLNE) to node[midway] {$\subseteq$} (rPATL);
	
	\draw[-, color=gray] (rPATL.west) edge[bend right=10] node[text=red]{$\lightning$} (HyperSt2.30);
	\draw[-, color=gray] (PATLNE.west) edge node[text=orange]{?} (HyperSt2.east);
	\draw[-, color=gray] (PATL1NE.west) edge[bend left=15] node[text=green100,xscale=-1]{$\subset$} (HyperSt2.330);
\end{tikzpicture}
	}\centering
	\hfil
	\subfloat[Hyperlogics \label{fig:overview-hyper}]{%
		\begin{tikzpicture}[every node/.style={sloped}]
	\draw[-, dotted, color=gray] (0,0) -- (3.25,0);
	\node[color=gray, rotate=90] at (3,0.5) {\tiny MDPs};
	\node[color=gray, rotate=90] at (3,-0.5) {\tiny DTMCs};
	
	\draw[-, dotted, color=gray] (0,0) edge[bend left=15] (-1,1);
	\draw[-, dotted, color=gray] (0,0) edge[bend right=15] (-1,-1);
	\node[color=gray, rotate=90, align=center] at (-2.75,0) {\tiny turn-based \\[-1.5ex] \tiny non-stoch.\ games};
	
	\node (HyperSt2) at (0,0) {\HyperSGL};
	
	\node (PHL) at (0,0.75) {\PHL};
	\node (HyperPCTLstar) at (0.1,-0.75) {\HyperPCTLstar};

	\path (PHL) to node[midway,sloped=false,color=red] {$\lightning$} (HyperSt2);
	\path (HyperPCTLstar) to node[midway,sloped=false,color=red] {$\lightning$} (HyperSt2);
	
	\node[align=center] (HyperPCTLMDP) at (2.05,0.5) {\HyperPCTL \\[0ex] \smaller on MDPs};
	\node[align=center] (HyperPCTLDTMC) at (2.05,-0.5) {\HyperPCTL \\[0ex] \smaller on DTMCs};
	
	\draw[-, color=gray] (HyperPCTLMDP.west) edge[bend right=10] node[text=red]{$\lightning$} (HyperSt2.10);
	\draw[-, color=gray] (HyperPCTLDTMC.west) edge[bend left=10] node[text=green100,xscale=-1]{$\subset$} (HyperSt2.350);

	\node (HyperSL) at (-1.75,0.75) {\HyperSL};
	\node (HyperSLuni) at (-1.75,0) {\HyperSLfrag};
	\node (HyperATLstar) at (-1.75,-0.75) {\HyperATLstar};
	
	\path (HyperATLstar) to node[midway] {$\subseteq$} (HyperSLuni) to node[midway] {$\subseteq$} (HyperSL);
	
	\path (HyperSLuni) to node[midway, color=green100] {$\subseteq$} (HyperSt2);
\end{tikzpicture}
	}
	\caption{Overview of the relationship between \HyperSGL and other logics. \textcolor{red}{$\lightning$}: Logics syntactically incompatible. \textcolor{orange}{?}: Open problem / depicted relation is a conjecture.}
	\label{fig:overview-comparison}
\end{figure}

\subsection{Hyperlogics for Stochastic Systems}

The probabilistic hyperlogics \HyperPCTL, \PHL and \HyperPCTLstar take different approaches for extending temporal logics to hyperlogics on stochastic systems.
Apart from syntactical differences, they also make different assumptions about initial states:
\PHL is defined for systems with initial state distributions and \HyperPCTLstar for systems with a unique initial state, and both quantify only over paths from initial states.%
\footnote{Note that an MDP with a fixed initial state is a special case of one with an initial distribution, and that an MDP with an initial distribution can be transformed into one with a fresh initial state whose outgoing transitions match the initial distribution.}
In contrast, \HyperPCTL\xspace --- like \HyperSGL\space --- is defined for systems without specified initial states and allows to quantify over all states.
A technical difference between the logics is thus that \HyperPCTL and \HyperSGL can reason about unreachable states while \HyperPCTLstar and \PHL cannot.

We note that we decided to base \HyperSGL on \HyperPCTL, and the relationship between \PHL, \HyperPCTL, and \HyperPCTLstar is not fully understood yet. 

\paragraph{\HyperPCTL}
\label{sec:hyperpctl}
for DTMCs \cite{abrahamHyperPCTLTemporal2018} extends \PCTL by explicit quantification over states
and allows to express probabilistic relations between several executions of a DTMC.
\HyperPCTL includes the bounded Until operator, which we decided to leave out in \HyperSGL for the sake of simplicity. It would be straightforward to extend \HyperSGL with the bounded Until operator, and this would not increase the model-checking complexity.
\HyperSGL subsumes \HyperPCTL without bounded Until on DTMCs.

\begin{proposition} \label{prop:hyperpctl}
	For a \HyperPCTL formula $\phi := Q_1 \varstate_1 \ldots Q_n \varstate_n .\ \pnonquant$ with $n$ state quantifiers that does not contain bounded Until, it holds that $\dtmc \modelsHyperPCTL \phi$ iff $\tbg_\dtmc \modelsHyperSGL Q_1 \varstate_1 \ldots Q_n \varstate_n .\ \forallStrat{1}{\varstate_1, \ldots, \varstate_n} .\ \pnonquant$ for all DTMCs $\dtmc$ where $\tbg_\dtmc$ corresponds to $\dtmc$ interpreted as a turn-based game $\tbg_\dtmc$ with $\agents = \{1\}$.
\end{proposition}

\HyperPCTL was extended to MDPs \cite{dobeModelChecking2022} by adding quantification over schedulers.
\HyperSGL cannot subsume \HyperPCTL on MDPs, since \HyperPCTL quantifies first over strategies (also called schedulers in the MDP-setting) and then over states, while \HyperSGL quantifies first over states and then over strategies.\footnote{As mentioned in \cref{sec:logic}, while \HyperSGL is based on \HyperPCTL, a more flexible strategy quantification structure, including in particular also nested strategy quantification, seem more appropriate for games, and therefore the quantification order is switched compared to \HyperPCTL.}
Intuitively, in a \HyperSGL formula of the form $\forall \varstate_1 \forall \varstate_2 .\ \existsStrat{1}{\varstate_1} \existsStrat{1}{\varstate_2} \pnonquant$ we can pick the schedulers depending on the instantiations of \emph{both} state variables, 
but in a \HyperPCTL formula of the form $\exists \varsched_1 \exists \varsched_2 .\ \forall \varstate_1 (\varsched_1) \forall \varstate_2 (\varsched_2) .\ \pnonquant$ we have to find schedulers that work for all possible instantiations of the state variables. 
Specifically, the behavior of the scheduler assigned to $\varsched_1$ cannot depend on the initial state assignment of $\varstate_2$. 

\begin{conjecture}
	\HyperSGL does not subsume \HyperPCTL on MDPs and vice versa. 
\end{conjecture}

If \HyperSGL first quantified over strategies and then over states, then \HyperSGL would strictly subsume \HyperPCTL (without bounded Until) over MDPs.

For only a \emph{single} state and scheduler variable, the quantification order does not matter over general schedulers: Intuitively, if for all states $\state$ there exists a scheduler $\sched_\state$ such that $\pnonquant$ holds, we can construct a scheduler satisfying $\pnonquant$ from every state by choosing the appropriate $\sched_\state$ depending on the initial assignment of $\varstate$. 
Over memoryless deterministic schedulers, however, the quantification order does matter since we cannot change behavior depending on where we start.

\paragraph{\PHL}
\label{sec:phl}
\cite{dimitrovaProbabilisticHyperproperties2020} allows to define probabilistic hyperproperties over MDPs with an initial state distribution.
\PHL formulas start with scheduler quantification\footnote{To be precise, every formula can be brought into PNF, i.e., transformed such that scheduler quantification only occurs in a prefix.}, followed by a Boolean combination of 
(1) comparison of probabilistic expressions over \LTL marked with scheduler variables, and 
(2) \HyperCTLstar formulas.
These two types of subformula cannot be mixed.
In the first case, paths start at any of the initial states. 
In the second case, they start at the current state of the path variable that was quantified last (or initially at an initial state).
\HyperSGL allows to restrict state quantification to initial states but does not (natively) allow to start a new experiment at the current state of an existing experiment.

Neither case is syntactically compatible with \HyperSGL:
(1) \HyperSGL allows nesting probability operators but not direct nesting of temporal operators, while \PHL allows the latter but not the former, and
(2) \HyperCTLstar subsumes \CTLstar while \HyperSGL subsumes \PCTL, and \CTLstar and \PCTL are incomparable over MDPs. 

\begin{conjecture}
	\HyperSGL does not subsume \PHL on MDPs and vice versa.
\end{conjecture}

\paragraph{\HyperPCTLstar}
\label{sec:hyperpctlstar}
\cite{wangStatisticalModel2021} is defined for DTMCs with a specified initial state.
\HyperPCTLstar extends \PCTLstar and thus allows direct nesting of temporal operators, which is not allowed in \HyperSGL.
\HyperPCTLstar can be seen as a probabilistic extension of \HyperLTL~\cite{clarksonTemporalLogics2014}, replacing the explicit path quantifiers from \HyperLTL by probability operators annotated with a number of path variables that indicate how to draw the random paths over which the probability operator is evaluated.
Paths are drawn from the initial state of the DTMC or the initial state of the path currently assigned to some path variable.
While a probability operator is evaluated only over the specified paths, one may still refer to paths drawn outside of the current probability operator but the evolution of such previously drawn paths is `set in stone'. 
In contrast, \HyperSGL does not support starting new experiments in some subformula since state quantification is restricted to the beginning of the formula.
A probability operator is evaluated over all paths starting from the current state assignments, i.e., all paths are drawn anew and no path is `set in stone' like in \HyperPCTLstar. 
These differences illustrate that \HyperPCTLstar and \HyperSGL use very different approaches and thus are hard to compare.

\begin{conjecture}
	\HyperPCTLstar does not subsume \HyperSGL on MDPs and 
	vice versa.
\end{conjecture}

\subsection{Hyperlogics for Non-Stochastic Games}

\paragraph{\HyperSL}
\label{sec:hypersl}
\cite{beutnerHyperStrategy2024} is an extension of strategy logic \cite{chatterjeeStrategyLogic2007} that allows to specify strategic hyperproperties for concurrent non-stochastic games over history-dependent deterministic strategies.
\HyperSL state formulas allow to explicitly quantify over named strategies, combine the quantified strategy variables into strategy profiles, bind the plays resulting from these strategy profiles to path variables, and then evaluate a path formula with respect to these path variables.
Path formulas consist of Boolean and temporal operators, atomic propositions indexed by a path variable, and nested state formulas indexed by a path variable.
A nested state formula is evaluated from the current state of the indexing path variable.
The detailed syntax and semantics are given in \iftoggle{extended}{\cref{app:hypersl-proof}}{the extended version}.

\HyperSL and \HyperSGL can combine strategies in similar ways: Several experiments (corresponding to path variables in \HyperSL and state variables in \HyperSGL) can share a strategy and strategy quantification can be nested without overwriting previous quantification. 
However, in \HyperSL, all currently fixed strategies and paths are `forgotten' when nesting strategy quantification in path formulas (by nesting a state formula inside a path formula). 
In \HyperSGL, on the other hand, nested strategy quantification only overwrites the current strategy for the specified agents and state variables. 

\HyperSL is defined for concurrent games that have the same actions available for all agents at each state and allows to assign the same strategy variable to different agents.
For example, a \HyperSL formula of the form $\exists \varstratt .\ \psi[\varpath : (\varstratt, \varstratt)]$ for a concurrent two-player game asks for a strategy $\varstratt$ such that $\psi$ holds if both agents follow $\varstratt$.
However, in turn-based games, 
only one agent can make a move at every state
and thus sharing strategies between agents does not seem meaningful here.
In order to compare the expressivity of \HyperSL and \HyperSGL over turn-based non-stochastic games, it is hence reasonable to only consider the fragment of \HyperSL consisting of all formulas where every strategy variable is only assigned to a unique agent. 
We call this fragment \HyperSLfrag.
We will show that each \HyperSLfrag formula can be transformed into a \HyperSGL formula.
The translation of path formulas is straightforward apart from the translation of a state formula $\phi$ nested inside a path formula $\psi = (\phi)_{\varpath}$.
This nesting starts new experiments (i.e., introduces new path variables) at the current state of another experiment (namely at the first state of the path currently assigned to $\varpath$).
The \HyperSGL syntax does not allow to start a new experiment inside a non-quantified formula (state quantification is restricted to the prefix of a formula). 
However, we can simulate the behavior of \HyperSL by starting all experiments initially and enforcing that experiments occurring in a nested state formula follow the same strategies as the current indexing experiment until the nested formula is reached.

For example, consider the following \HyperSL formula, which serves purely illustrative purposes and does not necessarily express a meaningful property,
\[\arraycolsep=0pt
\begin{array}{rl}
	\phiex &:= 
	\exists \varstratt_1, \varstratt_2, \varstratt_3 
	\Big( \ap_{\varpath_1} \Until \Big[ 
	{\color{cyan}
		\exists \varstratt'_1, \varstratt'_2 \
		\psi'%
		[\varpath_3 : (\varstratt'_1, \varstratt'_2)]
	}
	\Big]_{\hspace{-0.05cm}\varpath_2}
	\Big)%
	[\varpath_1 {:} (\varstratt_1, \varstratt_2), \varpath_2 {:} (\varstratt_1, \varstratt_3)] \\
	{\color{cyan}
	\psi'} &{\color{cyan}:= \Finally \Big( \aptwo_{\varpath_3} \land \Big[ 
	{\color{blue}
		\exists \varstratt''_1, \varstratt''_2 \ ( \overbrace{\ap_{\varpath_4}}^{\psi''})[\varpath_4 : (\varstratt''_1, \varstratt''_2)]}
	\Big]_{\hspace{-0.05cm}\varpath_3} \Big)
	}
\end{array}\]
\begin{wrapfigure}[6]{r}{2.2cm}
	\centering
	\vspace{-2\baselineskip}
	\begin{tikzpicture}[every node/.style={shape=circle, inner sep=0}]
		\node (one0) at (0,0) {};
		\node (one1) at (0,-2) {$\varpath_2$};
		\node (two0) at (0.05,0) {};
		\node (two1) at (0.05,-0.7) {};
		\node[color=cyan] (two3) at (0.9,-2) {$\varpath_3$}; 
		\node (three0) at (0.1,0) {};
		\node (three1) at (0.1,-0.7) {}; 
		\node (three2) at (0.58,-1.4) {};
		\node[color=blue] (three3) at (1.7,-2) {$\varpath_4$};
		
		\tikzset{decoration={snake,amplitude=.2mm,segment length=3mm,
				post length=0mm,pre length=0mm}}
		\draw[decorate] (one0) -- (one1);
		\draw[decorate, color=cyan] (two0) -- (two1) -- (two3);
		\draw[decorate, color=blue] (three0) -- (three1) to[out=-50, in=120] (three2) to[out=-60] (three3);
		
		\node[xshift=0.3cm, color=cyan] (psip) at (three1) {$\psi'$};
		\node[xshift=-0.2cm, color=cyan] (psipl) at (three1) {};
		\node[xshift=0.15cm, color=cyan] (psipr) at (three1) {};
		
		\draw[densely dotted, thick, color=cyan] (psipl) -- (psipr);
		
		\node[xshift=0.4cm, yshift=0.05cm, color=blue] (psipp) at (three2) {$\psi''$};
		\node[xshift=-0.25cm, color=blue] (psippl) at (three2) {};
		\node[xshift=0.175cm, color=blue] (psippr) at (three2) {};
		
		\draw[densely dotted, thick, color=blue] (psippl) -- (psippr);
	\end{tikzpicture}
\end{wrapfigure}
It first starts experiments (paths) $\varpath_1$ and $\varpath_2$ at the initial state using strategy profiles $(\varstratt_1, \varstratt_2)$ and $(\varstratt_1, \varstratt_3)$, respectively.
On the right-hand side of the Until, it then starts \textcolor{cyan}{$\varpath_3$} at the current state of $\varpath_2$ and then \textcolor{blue}{$\varpath_4$} at the current state of \textbf{$\varpath_3$}, as illustrated on the right.
We can mimic this in \HyperSGL by starting four experiments $\varstate_1, \ldots, \varstate_4$ from the initial state and specifying that (1) $\varstate_3$ and $\varstate_4$ initially follow the same strategies as $\varstate_2$, (2) $\varstate_4$ mimics $\varstate_3$ once we step inside $\psi'$, and (3) $\varstate_4$ follows its `own' strategy profile inside $\psi''$:
\[\arraycolsep=0pt
\begin{array}{l}
	\forall \varstate_1, \varstate_2, \varstate_3, \varstate_4 .\ (\bigwedge_{i=1}^{4} \init_{\varstate_i}) \implies {}
	\translation(\phiex) 
	\qquad \text{where}
	\\
	\translation(\phiex) 
	:= 
	\existsStrat{1}{\varstate_1, \varstate_2, \varstate_3, \varstate_4}
	\existsStrat{2}{\varstate_1}
	\existsStrat{2}{\varstate_2, \varstate_3, \varstate_4}
	\Prob \Big( \ap_{\varstate_1} \Until 
	{\color{cyan}
		\existsStrat{1}{\varstate_3, \varstate_4}
		\existsStrat{2}{\varstate_3, \varstate_4}
		\translation(\psi')%
	}
	\Big) = 1
	\\
	{\color{cyan}
		\translation(\psi')} 
		\;\;\;
		{\color{cyan}:=
		\Prob\Big( \Finally ( \aptwo_{\varstate_3} \land 
		{\color{blue}
			\existsStrat{1}{\varstate_4} 
			\existsStrat{2}{\varstate_4} 
				\ap_{\varstate_4}
		}
		) \Big) = 1
	}
\end{array}\]
\indent Since \HyperSL is defined for deterministic memoryful strategies, we also restrict to this strategy class for \HyperSGL.
We interpret a non-stochastic game as a stochastic game with transition probabilities 0 or 1. 
Since \HyperSL is defined for games with an initial state $\sinit$, but \HyperSGL is evaluated over structures without initial states, we assume $\sinit$ to be labeled with $\init$ (see also \cref{sec:comp-stoch-games}).

\begin{restatable}{theorem}{hypersl}
	For every \HyperSLfrag state formula $\phi$, there exists a \HyperSGL formula $\phi'$, such that
	\[\tbg, \sinit \modelsHyperSL \phi \quad \iff \quad \tbg' \modelsHyperSGL \phi' \]
	over \emph{memoryful deterministic strategies},
	for all turn-based non-stochastic games $\tbg$ and states $\sinit$ of $\tbg$,  
	where $\tbg'$ is $\tbg$ interpreted as a stochastic game with transition probabilities 0 and 1, where additionally $\sinit$ is labeled with a fresh atomic proposition $\init$. 
\end{restatable}

\begin{proof}[Sketch]	
	We translate a \HyperSLfrag formula $\phi$ with path variables $\varpath_1$, \ldots, $\varpath_l$ to a \HyperSGL formula of the form $\forall \varstate_1 \ldots \forall \varstate_l .\ (\bigwedge_{\istate=1}^{l} \textit{init}_{\varstate_\istate}) \implies T(\phi)$, where state variable $\varstate_r$ captures $\varpath_r$ for $r=1, \ldots, l$. The size of our translation $\translation(\phi)$ is polynomial in the size of $\phi$.	
	Since each strategy variable is associated with exactly one agent by assumption, each strategy variable quantifier in \HyperSLfrag can be translated to a strategy quantifier in \HyperSGL. The challenge is determining which state variables should be associated with the strategy quantifier since
	new path variables start at the current state of a specified reference variable, yielding a tree of dependencies between the path variables.	
	In order to mimic this behavior in \HyperSGL, 
	we track these dependencies by determining the current \emph{reference variable} of each path variable in the context of each subformula, i.e., its `ancestor' in the dependency tree on the level corresponding to the subformula.
	For example, in the context of $\phiex$, the reference variables of $\varpath_4$ with respect to $\psi'$ and $\phiex$ are
	$\varpath_3$ and $\varpath_2$, respectively.
	%
	In order to translate a strategy quantifier $\exists \varstratt .\ \psi$, we thus collect all path variables $\varpath$ binding $\varstratt$ and additionally all path variables that have such a variable $\varpath$ as a reference variable for this subformula. 
	For example, $\varstratt_3$ in $\phiex$ is bound by $\varpath_2$, and $\varpath_2$ is the reference variable of $\varpath_3$ and $\varpath_4$ with respect to $\phiex$.
	We refer to \iftoggle{extended}{\cref{app:hypersl-proof}}{the extended version} for the full proof.
	\hfill\qed
\end{proof}

\paragraph*{\HyperATLstar}
\cite{beutnerHyperATLLogic2023} is a strategic hyperlogic for concurrent non-stochastic games that combines \ATLstar and \HyperCTLstar.
\HyperATLstar is subsumed by \HyperSL. In particular, existential strategy quantification $\existsStrat{\setofagents}{\!\!} \pi .\ \phi$ in \HyperATLstar can be translated to \HyperSL by existentially quantifying over fresh strategy variables for agents in $\setofagents$, and universally quantifying over fresh strategy variables for agents in $\overline{\setofagents}$. 
This translation does not require agents to share strategies. 
Thus, the resulting \HyperSL formula can be translated to \HyperSGL, and \HyperSGL subsumes \HyperATLstar.
\section{Model-checking}
\label{sec:mc}

Since \HyperSGL subsumes \PCTL and can reason about the existence of strategies, we can reduce the strategy-synthesis problem for single-player stochastic games with \PCTL objectives~\cite{brazdilStochasticGames2006} to \HyperSGL model-checking, yielding the following result.

\begin{theorem}
	\label{th:undecidable}
	The model-checking problem for \HyperSGL over memoryful probabilistic strategies is undecidable. 
\end{theorem}

Over bounded-memory strategies we can decide the model-checking problem by encoding it in \emph{non-linear real arithmetic} \cite{collinsQuantifier1975}\iftoggle{extended}{; see \cref{app:bounded-mem} for an encoding.}{ as shown in the extended version.}  

\begin{restatable}{theorem}{boundedMemDecidable}
	\label{th:bounded-mem-decidable}
	For $k \in \N$, the model-checking problem for \HyperSGL over $k$-memory probabilistic strategies is decidable. 
\end{restatable}

\begin{restatable}{theorem}{exptime}
	\label{th:exptime}
	\label{th:pspace-hard}
	The model-checking problem for \HyperSGL over memoryless deterministic strategies is in \EXPTIME and \PSPACE-hard.
\end{restatable}
\begin{proof}[Sketch]	
	\emph{\EXPTIME-membership:} In \iftoggle{extended}{\cref{app:exptime-membership}}{the extended version}, we give a brute-force \EXPTIME decision procedure that resolves state and strategy quantification by testing all possible states or (memoryless deterministic) strategies, respectively. 
	The number of state and strategy assignments that have to be checked is bounded exponentially in the size of the input.
	In particular, 
	probability expressions/temporal operators do not pose a challenge, but can be handled in the usual way for \PCTL on DTMCs by solving a linear equation system with at most $|\states|^n$ equations where $n$ is the number of state variables in the formula.

	\emph{\PSPACE-hardness:} 
	Model-checking \HyperPCTL on DTMCs is \PSPACE-hard~\cite{abrahamHyperPCTLTemporal2018}. \HyperSGL subsumes \HyperPCTL on DTMCs (\cref{prop:hyperpctl}); over DTMCs the strategy class does not play any role.
	\hfill\qed
\end{proof}

It remains an open question whether tighter upper and lower bounds exist. 

\begin{openproblem}
	Is the model-checking problem for \HyperSGL over memoryless deterministic strategies in \PSPACE? Is it \EXPTIME-hard?
\end{openproblem}

For a fixed number of state quantifiers, however, we can give tight bounds. 

\begin{restatable}{theorem}{fixedNoQuant}
	\label{th:fixedNoQuant}
	For $\numstatequant \in \N$, the model-checking problem for \HyperSGL formulas with $\numstatequant$ state quantifiers is \PSPACE-complete over memoryless deterministic strategies.
\end{restatable}

\begin{proof}[Sketch]
	We adapt the \PSPACE-completeness proof for \SGL over memoryless deterministic strategies \cite{baierStochasticGame2012}.
	We show that \PSPACE-hardness holds already for $n=1$.
	For membership, we inductively define the type of a \HyperSGL formula as a class of the polynomial time hierarchy and prove that an upper bound for the complexity of the model-checking problem for a class of formulas of the same type is given by this type.
	For the full proof, see \iftoggle{extended}{\cref{app:mc-fixed-number-quantifiers}}{the extended version}.
	\hfill\qed
\end{proof}

\section{Conclusion}
\label{sec:conclusion}

We have proposed a hyperlogic for strategies in turn-based stochastic games that allows to express probabilistic relations between several executions of a game, where an execution is a probabilistic computation tree resulting from fixed strategies for the agents, rooted at a state of the game.
To the best of our knowledge, \HyperSGL is the first hyperlogic for stochastic games.
We illustrated that \HyperSGL can express interesting properties of stochastic games, like the existence of optimal strategies or Nash equilibria.
\HyperSGL subsumes a fragment of \rPATL over turn-based stochastic games, and a fragment of \HyperSL that in turn subsumes \HyperATLstar over turn-based non-stochastic games. 
We have established that the \HyperSGL model-checking problem is undecidable over general strategies but decidable for bounded memory.
Over memoryless deterministic strategies, the model-checking problem is in \EXPTIME and \PSPACE-hard, but \PSPACE-complete if we fix the number of state quantifiers.
In future work, it would be interesting to extend \HyperSGL with the ability to reason about (expected) rewards, and to lift the logic to \emph{concurrent} stochastic games.

\begin{credits}
	\subsubsection{\ackname} 
	Lina Gerlach is supported by the DFG RTG 2236/2 \textit{UnRAVeL}. 
	
	\subsubsection{\discintname}
	The authors have no competing interests to declare that are
	relevant to the content of this article.
\end{credits}

%
%
%
\bibliographystyle{splncs04}
\bibliography{../library.bib}

\iftoggle{extended}{
	\appendix 
	\section{\rPATL (\cref{sec:rpatl})}
\label{app:rPATL}

\rPATL \cite{chenProbabilisticAlternatingtime2007,kwiatkowskaEquilibriabasedProbabilistic2019} has been defined for both turn-based and concurrent stochastic games. 
We focus here on the fragment \PATLNE of \rPATL with social welfare subgame perfect $\epsilon$-Nash equilibria without reward objectives and without the bounded Until operator, on turn-based stochastic games.
Let $\agents$ be a set of agents, and $\AP$ a set of atomic propositions.
\PATLNE formulas are constructed according to the following grammar:
\[ 
\begin{array}{lll}
	\text{state formula:} 
	& \phi &::= \true \mid \ap \mid \neg \phi \mid \phi \land \phi 
	\mid \existsStrat{\setofagents}{}\texttt{P}_{\sim q}[\psi] 
	\mid \existsStrat{C {:} \overline{C}}{max \sim x}(\theta) \\
	\text{non-zero sum formula:} 
	&\theta &::= \texttt{P}[\psi] + \texttt{P}[\psi] \\
	\text{path formula:}	
	&\psi &::= \Next \phi \mid \phi \Until \phi 
\end{array}
\]
where $\ap \in \AP$,
$\setofagents, C \subseteq \agents$,
$q \in [0, 1]$, and $x \in \R$.
Note that \rPATL originally also includes the bounded Until operator.
For the sake of simplicity, we decided not to include this operator in \HyperSGL but it would be straightforward to extend \HyperSGL with the bounded Until operator, this would not increase the model-checking complexity.

Let $\tbg = \tbgtup$ be a turn-based stochastic game,
and
$\state \in \states$.
In the following, we use ${\Pr}^{\strat}(s \models \psi^i)$ as a shorthand for $\Pr(\pi \in \Pathsfrom[\tbg^{\strat}][(s, \start^{}(s))] \mid \pi \modelsrPATL \psi^i)$. 
The semantics of \PATLNE is defined inductively as follows:
\[ 
\begin{array}{lll}
	\state \models \true & & \\
	\state \models \ap & \iff 
	&\ap \in \labelingfct(\state) \\
	\state \models \phi_1 \land \phi_2 & \iff
	&\state \models \phi_1 \text{ and } \state \models \phi_2 \\
	\state \models \neg\phi & \iff &
	\state \not\models \phi \\
	\state \models \existsStrat{\setofagents}{} \texttt{P}_{\sim q}[\psi] & \iff & 
	\exists \strat \in \Strats[][\setofagents] .\ 
	\forall \overline{\strat} \in \Strats[][\overline{\setofagents}] .\ 
	{\Pr}^{\strat \oplus \overline{\strat}}(\state \models \psi) \sim q
\end{array}
\]
where 
$\ap \in \AP$,
$\setofagents \subseteq \agents$,
$\sim \in \{\leq, <, =, >, \geq\}$,
and
$q \in [0, 1]$. 
The semantics of path formulas is defined as follows
for $\pi \in \Paths[\tbg^\strat]$ for some $\strat \in \Strats[][\agents]$:
\[ \begin{array}{lcl}
	\pi \models \Next \phi_1 &\iff &\projS{\pi[1] }\models \phi_1 \\
	\pi \models \phi_1 \Until \phi_2 &\iff& \exists j_2 \geq 0 .\ \projS{\pi[j_2]} \models \phi_2 
	\text{ and } \forall 0 \leq j_1 < j_2 .\ \projS{\pi[j_1]} \models \phi_1
\end{array}\]
where for $r = (\state, \mode) \in \states^{\tbg^\strat}$ we use $\projS{r} = \state$ to denote the projection to the $\states$-component.

The semantics of a \swspepsne construct is as follows for 
$\state \in \states$, 
$C \subseteq \agents$,
and
$x \in \R$:
	\label{def:SWSPepsNE}
	\[\state \modelsrPATL \existsStrat{C {:} \overline{C}}{max \sim x}(\texttt{P}[\psi^1] + \texttt{P}[\psi^2]) \iff \exists \nasheq \in \Strats[][\agents] \ \text{ s.t. } \hspace{2cm} \]
	\begin{itemize}
		\item $\Pr^{{\nasheq}}(\state \models \psi^1) + \Pr^{{\nasheq}}(\state \models \psi^2) \sim x$,
		and 
		
		\item $\nasheq$ is a social welfare subgame perfect $\epsilon$-Nash equilibrium (\swspepsne), i.e.,
		\begin{itemize}
			\item $\nasheq$ is a subgame perfect $\epsilon$-Nash equilibrium, 
			i.e., for all states $\state_1 \in \states$ we have 
			\begin{itemize}
				\item for all $\gamma \in \Strats[][C]$ we have 
				$\Pr^{{\nasheq}}(\state_1 \models \psi^1) \geq \Pr^{{\strat}}(\state_1 \models \psi^1) - \epsilon$ where $\alpha := \nasheq|_{\overline{C}} \oplus \gamma$
				
				\item for all $\gamma \in \Strats[][\overline{C}]$ we have 
				$\Pr^{{\nasheq}}(\state_1 \models \psi^2) \geq \Pr^{{\strat}}(\state_1 \models \psi^2) - \epsilon$ where $\alpha := \nasheq|_{C} \oplus \gamma$
			\end{itemize}
			
			\item for all subgame perfect $\epsilon$-Nash equilibria $\strat$ 
			we have \\
			$\Pr^{{\nasheq}}(\state \models \psi^1) + 
			\Pr^{{\nasheq}}(\state \models \psi^2) \geq 
			\Pr^{{\strat}}(\state \models \psi^1) + 
			\Pr^{{\strat}}(\state \models \psi^2)$. 
		\end{itemize} 
	\end{itemize}

For a state formula $\phi$, we let $\tbg, \state \models \phi$ iff $\state \models \phi$.

\translateSWSPepsNE*

\begin{proof}	
		We construct 
	\begin{align*}
		\phi' &:= \forall \varfixed \forall \varfixedp .\
		\forall \varopt \forall \varoptp .\
		\exists \varcomp \exists \varcompp .\ \\
		& \qquad \quad
		(\init_{\varfixed} \wedge \init_{\varfixedp} \wedge \varopt = \varoptp) 
		\implies
		(\varcomp = \varcompp \wedge \phi_{\textit{SWSP}\epsilon} )
		\\
		\phi_{\textit{SWSP}\epsilon}
		&:= 
		\existsStrat{\agents}{\varfixed, \varopt, \varoptp} 
		\Big[ \Prob(\psi^1_{\varfixed}) + \Prob(\psi^2_{\varfixed}) \sim x
		\wedge 
		\phi_{\textit{SP}\epsilon}(\varopt, \varoptp)  
		\\ & \qquad \qquad \qquad \qquad
		\wedge \forallStrat{\agents}{\varfixedp, \varcomp, \varcompp} 
		\left( \phi_{\textit{SP}\epsilon}(\varcomp, \varcompp) \implies \phi_{\textit{SW}}(\varfixed, \varfixedp) \right) \Big]
		\\
		\phi_{\textit{SP}\epsilon}(\varstate_i, \varstate'_i) &:= 
		\left(\forallStrat{C}{\varstate'_i} \Prob(\psi^1_{\varstate_i}) \geq \Prob(\psi^1_{\varstate'_i}) - \epsilon \right)
		\wedge 
		\left(\forallStrat{\overline{C}}{\varstate'_i} \Prob(\psi^2_{\varstate_i}) \geq \Prob(\psi^2_{\varstate'_i}) - \epsilon \right)
		\\
		\phi_{\textit{SW}}(\varfixed, \varfixedp) &:= 
		\sum_{i=1}^2  \Prob(\psi^i_{\varfixed})
		\geq
		\sum_{i=1}^2  \Prob(\psi^i_{\varfixedp}) 
	\end{align*}
	where $\psi^i_{\varstate}$ corresponds to $\psi^i$ with all atomic propositions indexed by $\varstate$ for $i=1,2$ and a state variable $\varstate$.
	The size of $\phi'$ is linear in the size of the \PATLNE formula $\existsStrat{C {:} \overline{C}}{max \sim x}(\texttt{P}[\psi^1] + \texttt{P}[\psi^2])$.
	
	Let $\tbg$ be a turn-based stochastic game $\tbg = \tbgtup$ and $\state \in \states$.
	By definition of the semantics of \PATLNE, we have
	\begin{align*}
		& \tbg, \state \modelsrPATL \existsStrat{C {:} \overline{C}}{max \sim x}(\texttt{P}[\psi^1] + \texttt{P}[\psi^2]) \\
		\iff \; & \exists \nasheq \in \Strats[][\agents] .\ 
		{\Pr}^{\nasheq}(\state \models \psi^1) + 
		{\Pr}^{\nasheq}(\state \models \psi^2) \sim x
		\\
		& \; \wedge \nasheq \text{ is a subgame perfect } \epsilon\text{-NE: }
		\\
		& \; \hspace{1.2em} 
		\forall \opt \in \states .\
		(\forall \gamma \in \Strats[][C] .\ {\Pr}^{\nasheq}(\opt \models \psi^1) \geq {\Pr}^{\nasheq|_{\overline{C}} \oplus \gamma}(\opt \models \psi^1) - \epsilon) \wedge
		\\
		& \; \hspace{5.2em} 
		(\forall \overline{\gamma} \in \Strats[][\overline{C}] .\ {\Pr}^{\nasheq}(\opt \models \psi^2) \geq {\Pr}^{\nasheq|_{C} \oplus \overline{\gamma}}(\opt \models \psi^2) - \epsilon)
		\\
		& \; \wedge \forall \schedcomp \in \Strats[][\agents] .\ 
		(\schedcomp \text{ is a subgame perfect } \epsilon\text{-NE: } \forall \comp \in \states .\ \ldots )
		\implies \\
		& \; \hspace{7em}
		{\Pr}^{\nasheq}(\state \models \psi^1) {+} 
		{\Pr}^{\nasheq}_{\state}(\pi \models \psi^2)
		\geq 
		{\Pr}^{\schedcomp}(\state \models \psi^1) {+} 
		{\Pr}^{\schedcomp}(\state \models \psi^2)
	\end{align*}
	By pulling out quantifiers, we can transform this to an equivalent formulation of the following form:
	\[\exists \nasheq \in \Strats[][\agents] .\
	\forall \schedcomp \in \Strats[][\agents] .\ \forall \opt \in \states .\
	\exists \comp \in \states .\ 
	\psi , \] 
	where 
	$\psi$ does not contain state quantifiers 
	and
	$\comp$ is used to verify that $\schedcomp$ is a \spepsne and is now quantified existentially since it was previously quantified universally in the premise of an implication.
	We can further transform this by pulling the state quantifiers to the front:
	\begin{align*}%
		&\exists \nasheq \in \Strats[][\agents] .\
		\forall \schedcomp \in \Strats[][\agents] .\
		\forall \opt .\
		\exists \comp .\ 
		\psi 
		\\
		\iff \; & 
		\exists \nasheq \in \Strats[][\agents] .\
		\forall \opt .\
		\forall \schedcomp \in \Strats[][\agents] .\
		\exists \comp .\ 
		\psi 
		\\
		\overset{\labelcref{item:move-opt}}{\iff} \; & 
		\forall \opt .\
		\exists \nasheq \in \Strats[][\agents] .\
		\forall \schedcomp \in \Strats[][\agents] .\
		\exists \comp .\ 
		\psi 
		\\
		\overset{\labelcref{item:move-comp}}{\iff} \; & 
		\forall \opt .\
		\exists \nasheq \in \Strats[][\agents] .\
		\exists \comp .\ 
		\forall \schedcomp \in \Strats[][\agents] .\
		\psi 
		\\
		\iff \; & 
		\forall \opt .\
		\exists \comp .\ 
		\exists \nasheq \in \Strats[][\agents] .\
		\forall \schedcomp \in \Strats[][\agents] .\
		\psi .
	\end{align*}
	
	\begin{enumerate}[(a)] 	
		\item \label{item:move-opt}
		`$\Rightarrow$': The claim holds trivially.
		
		`$\Leftarrow$': 
		Assume that for all $\opt \in \states$ there exists some $\nasheq_{\opt} = \schedtup[\superstaropt] \in \Strats[][\agents]$ such that for all $\schedcomp \in \Strats[][\agents]$ there exists some $\comp \in \states$ such that $\psi$ holds.
		Then, we can construct a joint strategy $\nasheq$ that behaves as $\nasheq_{\state'}$ if we start with a state $\state' \in \states$. 
		Formally, $\nasheq = \schedtup[\superstar]$ with
		\begin{multicols}{2}
			\begin{itemize}
				\item $\superstar[\modes] = \bigcup_{\state \in \states} \superstars[\modes] \times \{s\}$,
				\item $\superstar[\start](s) = (\superstars[\start](s), s)$,
				\item $\superstar[\modef]((\mode, s), \state') = \superstars[\modef](\mode, \state')$,
				\item $\superstar[\act]((\mode, s), \state', \action) = \superstars[\act](\mode, \state', \action)$ 
			\end{itemize}
		\end{multicols}
		for $s, \state' \in \states$, $\mode \in \superstars[\modes]$ and $\action \in \Act$.
		By assumption, for this $\nasheq$ it must hold that for all $\opt \in \states$ for all $\schedcomp \in \Strats[][\agents]$ there must exist $\comp \in \states$ such that $\psi$ holds.
		
		\item \label{item:move-comp}
		`$\Leftarrow$': The claim holds trivially.
		
		`$\Rightarrow$': We show the claim by contraposition. 
		Assume that there exists some $\opt \in \states$ such that for all $\nasheq \in \Strats[][\agents]$ for all $\comp \in \states$, there exists some $\schedcomp_{\opt, \nasheq, \comp}$ such that $\psi$ does not hold.
		Then, there exists some $\opt \in \states$ such that for all $\nasheq \in \Strats[][\agents]$ we can construct a witness $\schedcomp$ such that for all $\comp \in \states$, $\psi$ does not hold, namely by defining $\alpha$ such that it behaves as $\schedcomp_{\opt, \nasheq, s}$ if we start with state $s \in \states$. 
	\end{enumerate}
	
	It remains to show that this matches the semantics of the \HyperSGL formula.
	Let $\tbg'$ be the same as $\tbg$ but $\state \in \states$ is labeled with a fresh atomic proposition $\init$.
	By definition of the semantics of \HyperSGL, we have
	\begin{align*}
		&\tbg' \models \forall \varfixed \forall \varfixedp .\
		\forall \varopt \forall \varoptp .\
		\exists \varcomp \exists \varcompp .\ 
		(\init_{\varfixed} \wedge \init_{\varfixedp} \wedge \varopt = \varoptp) 
		\implies
		(\varcomp = \varcompp \wedge \phi_{\textit{SWSP}\epsilon} ) 
		\\ & \iff \\
		& \forall \opt .\
		\exists \comp .\ 
		\exists \nasheq \in \Strats[][\agents] .\
		\forall \schedcomp \in \Strats[][\agents] .\
		\\
		& \quad \Ctxtup[\tbg'] \models
		\textstyle\sum_{i=1}^2  \Prob(\psi^i_{\varfixed}) \sim x \wedge 
		\phi_{\textit{SP}\epsilon}(\varopt, \varoptp) 
		\\ & \qquad \qquad \qquad \qquad \quad
		\wedge
		\left( \phi_{\textit{SP}\epsilon}(\varcomp, \varcompp) \implies \phi_{\textit{SW}}(\varfixed, \varfixedp) \right)
	\end{align*}
	where 
	$\mapStrat$ and $\statetupMap$ are defined as follows for $\agent \in \agents$
	\[\begin{array}{wl{0.3\textwidth} wl{0.3\textwidth} wl{0.3\textwidth}}
		\mapStrat(\varfixed, \agent) = \nasheq|_\agent & \mapStrat(\varfixedp, \agent) = \schedcomp|_\agent &  \statetupMap(\varfixed) = \statetupMap(\varfixedp) = \fixed \\
		\mapStrat(\varopt, \agent) = \nasheq|_\agent & \mapStrat(\varoptp, \agent) = \nasheq|_\agent & \statetupMap(\varopt) = \statetupMap(\varoptp) = \opt \\
		\mapStrat(\varcomp, \agent) = \schedcomp|_\agent & \mapStrat(\varcompp, \agent) = \schedcomp|_\agent & \statetupMap(\varcomp) = \statetupMap(\varcompp) = \comp .
	\end{array}\]
	\hfill\qed
\end{proof}

\subsection{Extending to general \PATLfrag and \PATLNE}

\paragraph{\PATLfrag}
\label{app:PATL-1NE}
formulas are built according to the grammar
\[
\begin{array}{lll}
	\text{top-level formula:}&\Phi &::= \true \mid \ap \mid \neg \Phi \mid \Phi \wedge \Phi \mid \theta \mid \rho \\
	\text{Nash equilibrium:}&\theta &::= \existsStrat{C : \overline{C}}{max \sim x}(\texttt{P}[\psi] + \texttt{P}[\psi]) \\
	\text{state formula w/o NE:}&\rho &::= \true \mid \ap \mid \neg \rho \mid \rho \wedge \rho \mid \existsStrat{\setofagents}{} \texttt{P}_{\sim q}[\psi] \\
	\text{path formula:} &\psi &::= \Next \rho \mid \rho \Until \rho 
\end{array}
\]
where $\ap \in \AP$,
$\setofagents, C \subseteq \agents$,
$q \in [0, 1]$, 
and
$x \in \R$. 
Embedding top-level formulas, state formulas without Nash equilibria, and path formulas in \HyperSGL is straightforward. 
Note that, \rPATL existential strategy quantification implicitly quantifies universally over all opponent strategies, which needs to be made explicit in the translation to \HyperSGL.
Nash equilibria can be translated as in \cref{th:translateSWSPepsNE}.

\paragraph{\PATLNE.}
\label{app:general-PATLNE}
It is not obvious whether general \PATLNE can be simulated in \HyperSGL. 
More precisely, it is unclear how to express formulas with Nash equilibria nested inside strategy quantification, including Nash equilibria nested inside other Nash equilibria, in \HyperSGL. 
For example, consider the \PATLNE formula 
\[ \phinested := \llangle \agents \rrangle \texttt{P}_{\sim q}\big[\Finally \existsStrat{C {:} \overline{C}}{\textit{max} \sim x} (\texttt{P}[\Finally \ap] + \texttt{P}[\Finally \aptwo]) \big] . \]
A state $\state$ satisfies $\phinested$ if and only if 
\[\exists \strat \in \Strats[][\agents] .\ \Pr(\pi \in \Pathsfrom[{\tbg^{\strat}}][(\state, \start(\state))] \mid \exists i .\ \exists \nasheq \in \Strats[][\agents] .\ \forall \opt .\ \exists \comp .\  \ldots ) \sim q \]
where `\ldots' has to be replaced with the details of the evaluation of the Nash equilibrium.
If we translate $\phinested$ to \HyperSGL using the construction presented above for expressing the nested Nash equilibrium in \HyperSGL, we need to introduce two state variables $\varfixed, \varfixedp$ to track the current state, and four state variables $\varopt, \varoptp, \varcomp, \varcompp$ to define the \swspepsne. 
We need to quantify over these state variables at the beginning of the formula. 
Hence, the semantics of the translated \HyperSGL formula would be of the form 
\[\forall \opt .\ \exists \comp .\ \exists \strat \in \Strats[][\agents] .\ \Pr(\pi \in \Pathsfrom[{\tbg^{\strat}}][(\state, \start(\state))] \mid \exists i .\ \exists \nasheq \in \Strats[][\agents] .\ \ldots ) \sim q \] 
which does not match the semantics of the \PATLNE formula.
	\section{\SGL (\cref{sec:sgl})}
\label{app:sgl-proof}

For a DRA $\dra$, we use $\mathcal{L}(\dra)$ to denote the language accepted by $\dra$. 
We refer to~\cite{gradelAutomataLogics2002} for details on deterministic Rabin automata. 
Let $\agents$ be a set of agents, and
$\AP$ a set of atomic propositions.
\SGL~\cite{baierStochasticGame2012} formulas are formed according to the following syntax:
\[
\phi := 
\true \mid 
\ap \mid 
\neg \phi \mid 
\llangle \setofagents \rrangle \phi \mid 
\ProbSGL_{\sim x} (\dra; \phi, \ldots, \phi ) 
\]
where 
$\ap \in \AP$,
$\setofagents \subseteq \agents$,
$\sim \in \{\leq, <, =, >, \geq\}$,
$x \in [0,1]$,
and
$\dra$ is a DRA.
The semantics of \SGL is defined inductively as follows,
for a turn-based stochastic game $\tbg$, 
$\state \in \states$, 
$\setofagents \subseteq \agents$, and 
$\strat \in \Strats$:
\[ 
\begin{array}{lcl}
	\state, \setofagents, \strat \models \true & & \\
	\state, \setofagents, \strat \models \ap & \iff 
	&\ap \in \labelingfct(\state) \\
	\state, \setofagents, \strat \models \neg\phi & \iff &
	\state, \setofagents, \strat \not\models \phi \\
	\state, \setofagents, \strat \models \llangle \setofagentstwo \rrangle \phi & \iff &
	\exists \strattwo \in \Strats[][\setofagentstwo] .\
	\state, \setofagents \cup \setofagentstwo, \strat \leftarrow \strattwo \models \phi
	\\
	\state, \setofagents, \strat \models \ProbSGL_{\sim x} (\dra; \vec{\phi}) & \iff & 
	\forall \overline{\strat} \in \Strats[][\overline{\setofagents}] .\ 
	\\ & & \quad 
	{\Pr}^{\tbg^{\strat \oplus \overline{\strat}}}_{r}(\{ \pi \in \Pathsfrom[\tbg^{\strat \oplus \overline{\strat}}][r] \mid \tilde{\pi}^{\vec{\phi}}_{\setofagents, \strat} \in \mathcal{L}(\dra) \}) \sim x 
\end{array}
\]
where 
$r = (\state, \start^{\strat \oplus \overline{\strat}}(\state))$,
$\vec{\phi} = (\phi_1, \ldots, \phi_j)$ 
and
$\tilde{\pi}^{\phi_1, \ldots, \phi_j}_{\setofagents, \strat} \in (2 ^{\{1, \ldots, j\}})^\omega$ is defined as follows for $i \geq 0$:
\[
\tilde{\pi}^{\phi_1, \ldots, \phi_j}_{\setofagents, \strat}(i) :=
\{ l \in \{1, \ldots, j\} \mid \pi[i], \setofagents, \strat \models \phi_l \}
\]
We write $\Prob_{\sim x}(\Next \phi)$ and $\Prob_{\sim x}(\phi_1 \Until \phi_2)$ for $\Prob_{\sim x}(\dra_{\Next}; \phi)$ and $\Prob_{\sim x}(\dra_{\Until}; \phi_1, \phi_2)$.
Further, we write $\phi_1 \wedge \phi_2$ for $\Prob_{> 0}(\dra_{\wedge}; \phi_1, \phi_2)$ and analogously for other Boolean operators.

There is a crucial difference between \SGL and \HyperSGL in the evaluation of temporal operators.
In \HyperSGL, when we evaluate a temporal operator on a path under some strategy mapping, we evaluate the subformula(s) from a position on the path under the \emph{shifted} strategy mapping that remembers the history of the play until that point.
In \SGL, however, the subformula(s) are evaluated under the \emph{same} strategy again, i.e., we do not remember the play so far but act as if we start a new game.
The following example illustrates the difference between shifting strategies in the evaluation of subformulas, and using the original strategy.%
\begin{example}
	{\makeatletter\let\par\@@par 
		\begin{wrapfigure}[3]{r}{0.3\textwidth}%
			\vspace{-1.5\baselineskip}%
			\begin{tikzpicture}[node distance=2cm, inner sep=1pt]
				\node[state, label=35:$\ap$, minimum size=20pt] (s0) {$s_0$}; 
				\node[state, right of=s0, label=145:$\aptwo$, minimum size=20pt] (s1) {$s_1$}; 
				
				\path[->]
				(s0) edge [loop left, looseness=6] node[pos=0.75, above] {$\action$} (s0)
				(s1) edge [loop right, looseness=6] node[pos=0.25, above] {$\actionthree$} (s1)
				(s0) edge node[pos=0.5, above] {$\actiontwo$} (s1);
			\end{tikzpicture}%
		\end{wrapfigure}%
		Consider 
		the single-player stochastic game $\tbgpq$ 
		depicted on the right
		and
		the \SGL formula $\existsStrat{1}{} \ProbSGL_{=1}(\Next (\ap \land \ProbSGL_{=1}(\Next \aptwo)))$.
		For $\strat \in \Strats[\tbgpq]$, we use $\Pr^{\tbgpq^\strat}_{\state}$ and $\Pathsfrom[\tbgpq^\strat]$ to denote $\Pr^{\tbgpq^\strat}_{(\state, \start^\strat(\state))}$ and $\Pathsfrom[\tbgpq^\strat][(\state, \start^\strat(\state))]$, under abuse of notation.
		Further, we use $\{\}$ to denote the unique $\emptyset$-strategy.
		It holds that 
		\begin{align*}
			&s_0, \emptyset, \emptymap \modelsSGL \existsStrat{1}{} \ProbSGL_{=1}(\Next (\ap \land \ProbSGL_{=1}(\Next \aptwo))) \\
			\iff& \exists \strat \in \Strats[\tbgpq][\{1\}] .\ {\Pr}^{\tbgpq^\strat}_{s_0}(\{\pi \in \Pathsfrom[\tbgpq^\strat][s_0] \mid \projS{\pi[1]}, \{1\}, \strat \modelsSGL (\ap \land \ProbSGL_{=1}(\Next \aptwo))  \}) = 1 
		\end{align*}
		In order to satisfy the first conjunct ($\ap$) at the next state, we must choose a strategy $\strat$ that deterministically selects $\action$ at the current state, $s_0$.
		Hence, $\projS{\pi[1]} = s_0$ for all $\pi \in \Pathsfrom[\tbgpq^\strat][s_0]$. 
		The subformula $\ProbSGL_{=1}(\Next \aptwo)$ is evaluated over $\strat$ again (instead of over $\strat^{\pi[0,1]}$). 
		Hence, we know that $\projS{\pi'[1]} = s_0$ for all $\pi' \in \Pathsfrom[\tbgpq^\strat][\projS{\pi[1]}]$ since $\projS{\pi[1]} = s_0$. This directly implies that the second conjunct ($\ProbSGL_{=1}(\Next \aptwo)$) cannot be satisfied while also satisfying the first conjunct and thus $s_0, \emptyset, \emptymap \not\models_{\SGL} \existsStrat{1}{} \ProbSGL_{=1}(\Next (\ap \land \ProbSGL_{=1}(\Next \aptwo)))$.
		\par}%
	
	In contrast, in \HyperSGL, it holds that
	\begin{align*}
		&\Ctxtup[\tbgpq][\emptymap][\{\varstate \mapsto s_0\}] \modelsHyperSGL \existsStrat{1}{\varstate} \Prob(\Next \underbrace{(\ap_{\varstate} \wedge \Prob(\Next \aptwo_{\varstate})=1)}_{\psi})=1 \\[-3ex]
		\iff& \exists \strat \in \Strats[\tbgpq][\{1\}] .\ \\ &{\Pr}^{\tbgpq^\strat}_{s_0}(\{\pi \in \Pathsfrom[\tbgpq^\strat][s_0] \mid \Ctxtup[\tbgpq][\{(\varstate,1) \mapsto \strat^{\pi[0,1]}\}][\{\varstate \mapsto \projS{\pi[1]}\}] \modelsHyperSGL \psi \}) = 1 .
	\end{align*}
	Since
	the subformula $(\ap_{\varstate} \land \Prob(\Next \aptwo_{\varstate})=1)$ is evaluated under the shifted strategy $\strat^{\pi[0,1]}$,
	and
	there does exist a strategy $\strat$ that chooses to first stay in $s_0$ once and then visit $s_1$,
	it holds that $\Ctxtup[\tbgpq][\emptymap][\{\varstate \mapsto s_0\}] \modelsHyperSGL \existsStrat{1}{\varstate} \Prob(\Next (\ap_{\varstate} \wedge \Prob(\Next \aptwo_{\varstate})=1))=1$.
\end{example}

	\section{\HyperSL (\cref{sec:hypersl})}
\label{app:hypersl-proof}

\HyperSL \cite{beutnerHyperStrategy2024} is defined for concurrent stochastic games.
Since \HyperSGL is defined for turn-based games, we only define \HyperSL with respect to turn-based stochastic games here. 

\HyperSL is defined over memoryful deterministic strategies. 
Let $\tbg = \tbgtup$ be a turn-based stochastic game and let $\finPaths$ (and $\finPathsfrom$) be the set of all finite paths in $\tbg$ (from $\state \in \states$, respectively).
We define the application of a memoryful deterministic strategy $\strat = \strattup \in \Strats[][\agents]$ to a finite path $\pi \in \finPaths$ recursively as follows:
\begin{itemize}
	\item If $\pi = \state$, then we let $\strat(\pi)$ be the unique $\action \in \Act$ s.t.\ $\act(\start(\state), \state, \action) = 1$, and
	we let $\modef(\pi) = \modef(\start(\state), \state)$, 
	and
	\item If $\pi = \pi' \state$, then we let $\strat(\pi)$ be the unique $\action \in \Act$ s.t.\ $\act(\modef(\pi'), \state, \action) = 1$,
	and we let $\modef(\pi) = \modef(\modef(\pi'), \state)$.
\end{itemize}

Let $\agents$ be some set of agents,
$\AP$ a set of atomic propositions,
$\VarsPath$ an infinite set of path variables, 
and
$\VarsStrat$ an infinite set of strategy variables.
\HyperSL formulas are built according to the following grammar:
\[\begin{array}{lll}
	\text{state formula:} 
	&\phi &::= \forall \varstrat .\ \phi \mid \exists \varstrat.\ \phi \mid \psi[\varpath_1 : \vec{\varstrat}_1, \ldots, \varpath_m : \vec{\varstrat}_m] \\
	\text{path formula:}
	&\psi & ::= \ap_{\varpath} \mid \phi_{\varpath} \mid \psi \land \psi \mid \neg \psi \mid \Next \psi \mid \psi \Until \psi 
\end{array}\]
where 
$\ap \in \AP$,
$\varpath, \varpath_1, \ldots, \varpath_m \in \VarsPath$,
$\varstrat \in \VarsStrat$,
and
$\vec{\varstrat}_1, \ldots, \vec{\varstrat}_m \colon \agents \to \VarsStrat$ are strategy profiles that assign a strategy variable to each agent.
It is assumed 
for each formula $\psi[\varpath_p : \vec{\varstrat}_p]_{p=1}^{m}$
that all path variables that are free in $\psi$ belong to $\{\varpath_1, \ldots, \varpath_m \}$, i.e., all used path variables are bound to some strategy profile. All nested state formulas are assumed to be closed.

Let 
$\tbg = \tbgtup$ be a turn-based stochastic game.
State formulas are evaluated over a state $\state \in \states$ and a \emph{strategy assignment}
$\Delta \colon \VarsStrat \rightharpoonup \bigcup_{\agent \in \agents} \Strats[][\{\agent\}]$, and path formulas are evaluated over a \emph{path assignment}
$\pathmap \colon \VarsPath \rightharpoonup \Paths$. 
The semantics of \HyperSL path formulas is defined inductively as follows: 
\[ 
\begin{array}{lll}
	\pathmap \models \ap_{\varpath} & \iff 
	&\ap \in \labelingfct(\pathmap(\varpath)[0]) \\
	\pathmap \models \phi_{\varpath} & \iff 
	&\pathmap(\varpath)[0], \emptymap \models \phi \\
	\pathmap \models \psi_1 \land \psi_2 & \iff
	&\pathmap \models \psi_1 \text{ and } \pathmap \models \psi_2 \\
	\pathmap \models \neg\psi & \iff &
	\pathmap \not\models \psi \\
	\pathmap \models \Next \psi & \iff & 
	\pathmap[1, \infty] \models \psi \\
	\pathmap \models \psi_1 \Until \psi_2 & \iff &
	\exists j_2 \in \N .\ \pathmap[j_2, \infty] \models \psi_2 \text{ and } \forall 0\leq j_1 < j_2 .\ \pathmap[j_1, \infty] \models \psi_1 
\end{array}
\]
where $\pathmap[j, \infty]$ denotes the shifted path assignment defined by $\pathmap[j, \infty](\varpath) := \pathmap(\varpath)[j, \infty]$ for $\varpath \in \VarsPath$, where for $\pi = \state_0 \state_1 \state_2 \ldots \in \Paths$ we let $\pi[j, \infty] = \state_j \state_{j+1} \ldots$.
The semantics of \HyperSL state formulas is defined as follows:
\[
\begin{array}{lll}
	\state, \Delta \models \forall \varstrat.\ \phi & \iff 
	&\forall \strat \in \bigcup_{\agent \in \agents} \Strats[][\{\agent\}] .\ \state, \Delta[\varstrat \mapsto \strat] \models \phi \\
	\state, \Delta \models \exists \varstrat.\ \phi & \iff 
	&\exists \strat \in \bigcup_{\agent \in \agents} \Strats[][\{\agent\}] .\ \state, \Delta[\varstrat \mapsto \strat] \models \phi \\
	\state, \Delta \models \psi[\varpath_p : \vec{\varstrat}_p]_{p=1}^{m} & \iff 
	&\left[\varpath_p \mapsto Play_{\tbg}(\state, \bigoplus_{\agent \in \agents} \Delta(\vec{\varstrat}_{p}(\agent)))\right]_{p=1}^{m} \models \psi 
\end{array}\]
where for a state $\state \in \states$ and a joint strategy $\strat \in \Strats[][\agents]$, we define $Play_{\tbg}(\state, \strat)$ as the unique path $\pi \in \Paths$ such that 
$\pi[0] = \state$
and
for all $j \in \N$ it holds that
$\pi[j+1]$ is the unique state with $\Trans(\pi[j], \strat(\pi[0, j]), \pi[j+1]) = 1$.
We use $\tbg, \state \models \phi$ to denote $\state, \emptymap \models \phi$.

\hypersl*

\begin{proof}
	W.l.o.g.\ we can assume that all path and strategy variable names are unique, i.e., no path variable is bound to a strategy profile twice and no strategy variable is quantified twice.
	We further assume that all strategy variables are used in a strategy profile at some point, and that all path variables are used to index an atomic proposition or nested state formula at some point.
	
	Let $\phi$ be a \HyperSLfrag formula and let $\varpath_1$, \ldots, $\varpath_l$ be the path variables in $\phi$ for some $l \in \N$.
	We will define a \HyperSGL formula $\phi'$ using state variables $\varstate_1$, \ldots, $\varstate_l$ where $\varstate_r$ captures $\varpath_r$ for $r=1, \ldots, l$.
	All state variables are quantified at the beginning of $\phi'$ and we require that all of them should initially be set to $\sinit$, 
	so $\phi'$ is of the form $\phi' := \forall \varstate_1 \ldots \forall \varstate_l .\ \big( \bigwedge_{\istate=1}^{l} \textit{init}_{\varstate_\istate} \implies \ldots \big)$. 
	Note that, since we consider non-stochastic games here, fixing a joint strategy determines a unique path from every state.
	Before we define a translation from \HyperSLfrag to \HyperSGL, we introduce some preliminary concepts.
	
	By assumption, each strategy variable is associated with exactly one agent. 
	For each strategy variable $\varstratt$ in $\phi$, we define $\Agt(\varstratt)$ to be the unique agent $\agent \in \agents$ such that there exists some path variable $\varpath_k$ in $\phi$ with $\varstratt = \vec{\varstrat}_{k}(\agent)$. 
	For example, in the context of $\phiex$ defined in \cref{sec:hypersl}, we have $\Agt(\varstratt_1) = 1$, $\Agt(\varstratt_2) = 2$, and $\Agt(\varstratt_3) = 2$.

	In \HyperSL, it is possible to nest state formulas inside path formulas via the construct $(Q_1 \varstratt_1 \ldots Q_n \varstratt_n .\ \psi[\varpath_{i_p} : \vec{\varstrat}_{i_p}]_{p=1}^{m})_{\varpath'}$, where the new path variables $\varpath_{i_1}, \ldots, \varpath_{i_m}$ start at the current state of $\varpath'$.	
	In order to be able to mimic this behavior in \HyperSGL, we define the \emph{reference variable} $\RefVar(\phi_1, \varpath)$ for some path variable $\varpath$ in the context of a \HyperSL formula $\phi_1$ inductively as follows: 
	\[\begin{array}{rcl}
		\RefVar(\cdot, \varpath) \colon 
		& \ap_{\varpath'} &\mapsto 
		\begin{cases}
			\varpath & \text{if }\varpath = \varpath' \\
			\valundefined & \text{otherwise}
		\end{cases} \\
		& \begin{rcases}
			\neg \psi_1 \\
			\Next \psi_1 
		\end{rcases}&\mapsto \RefVar(\psi_1, \varpath) \\
		& \begin{rcases}
			\psi_1 \wedge \psi_2 \\
			\psi_1 \Until \psi_2 \end{rcases}&\mapsto 
		\begin{cases}
			\valundefined & \begin{aligned} 
				\text{if } &\RefVar(\psi_1, \varpath) = \valundefined \land {} \\
				 & \RefVar(\varpath, \psi_2) = \valundefined
				\end{aligned} \\
			\varpath & 
			\begin{aligned}
				\text{if } &\RefVar(\psi_1, \varpath) = \varpath \vee {} \\ 
				 & \RefVar(\psi_2, \varpath) = \varpath 
			\end{aligned} \\
			\varpath' & 
			\begin{aligned} 
				\text{if } \exists \varpath' \neq \varpath .\ &\RefVar(\psi_1, \varpath) = \varpath'
				\vee {} \\
				 & \RefVar(\psi_2, \varpath) = \varpath' 
			\end{aligned}
		\end{cases} \\
		& (\phi_1)_{\varpath'} & \mapsto \begin{cases}
			\valundefined & \text{if } \varpath \neq \varpath' \land \RefVar(\phi_1, \varpath) = \valundefined \\
			\varpath' & \text{otherwise}
		\end{cases} \\
		& Q \varstratt .\ \phi_1 &\mapsto \RefVar(\phi_1, \varpath) \\
		& \psi_1 [\varpath_k : \vec{\varstrat}_k]_{k=1}^{m} & \mapsto  
		\RefVar(\psi_1, \varpath)
	\end{array}\]
	
	Observe that $\RefVar(\phi_1, \varpath)$ is undefined if and only if $\varpath$ does not occur in $\phi_1$.
	If a path variable $\varpath$ occurs in a formula of the form $\psi_1 \land \psi_2$, then by assumption $\varpath$ is either bound in exactly one of the subformulas or in neither. 
	In the latter case it follows that $\varpath$ must have already been bound outside of $\psi_1 \wedge \psi_2$ and $\RefVar(\psi_1 \land \psi_2, \varpath) = \varpath$. 
	Otherwise, w.l.o.g.\ we can assume that $\varpath$ is bound in $\psi_1$. Then, it must hold that $\RefVar(\psi_2, \varpath) = \valundefined$ and there must exist some path variable $\varpath'$ with $\RefVar(\psi_1, \varpath) = \varpath'$. Thus, $\RefVar(\psi_1 \land \psi_2, \varpath) = \varpath'$.
	
	Consider a formula of the form $(\phi_1)_{\varpath'}$.
	We know that $\varpath'$ cannot occur in $\phi_1$ since path variables may not be used in nested state formulas and we assume that path variables are not bound twice.
	Hence, $\RefVar(\phi_1, \varpath') = \valundefined$ and $\RefVar((\phi_1)_{\varpath'}, \varpath') = \varpath'$.
	For $\varpath \neq \varpath'$, if $\varpath$ does not occur in $\phi_1$, then also $\RefVar((\phi_1)_{\varpath'}, \varpath) = \valundefined$. 
	Otherwise, if $\varpath$ does occur in $\phi_1$, then $\varpath$ should follow the evolution of $\varpath'$ outside of $(\phi_1)_{\varpath'}$, so $\RefVar((\phi_1)_{\varpath'}, \varpath) = \varpath'$. 
	
	For example, consider again $\phiex$. 
	The reference variable of $\varpath_4$ with respect to $\psi'' = a_{\varpath_4}$ is
	$\RefVar(a_{\varpath_4}, \varpath_4) = \varpath_4$.
	Thus, $\RefVar(\psi', \varpath) = \RefVar(\big[ \exists \varstratt'' (\ap_{\varpath_4})[\varpath_4 : \vec{\varstrat}_4]\big]_{\varpath_3}, \varpath_4) = \varpath_3$.
	It follows that $\RefVar(\phiex, \varpath_4) = \varpath_2$.
	
	For some state variable $\varstate_r$, we also call $\RefVar(\phi_1, \varpath_r)$ the reference variable of $\varstate_r$ in $\phi_1$.
	We mimic the nested quantification over path variables in \HyperSGL by enforcing that each state variable always evolves in the same way as its current reference variable, i.e., follows the same strategy.
	Once we reach the subformula where the corresponding path variable is quantified, the reference variable is that path variable itself. 
	For every strategy variable quantifier $Q \varstratt$ in $\phi$, we thus need to collect all state variables whose reference variable is bound to a strategy profile containing $\varstratt$.
	Let $\psi_{\varstratt}[\varpath_{i_p} : \vec{\varstrat}_{i_p}]_{p=1}^{m}$ be the unique subformula of $\varphi$ that uses $\varstratt$ in assignments of paths to strategy profiles.
	We define
	$\Vars(\varstratt) := $ 
	\[\begin{array}{l}
		\{ \varstate_\istate \mid  
		\exists \varpath_{i_p} \in \{\varpath_{i_1}, \ldots, \varpath_{i_m}\} .\ 
		\RefVar(\psi_{\varstratt}, \varpath_\istate) = \varpath_{i_p} 
		\wedge 		
		\exists \agent \in \agents .\ \varstratt = \vec{\varstrat}_{i_p}(\agent) 
		\} 
	\end{array} . \]
	Intuitively, $\Vars(\varstratt)$ is the set of all state variables for which $\Agt(\varstratt)$ should follow the strategy assigned to $\varstratt$ in $\psi_{\varstratt}$ until a nested state formula is reached.
	This in particular includes all state variables corresponding to a path variable whose strategy profile assigns $\varstratt$ to some agent itself.
	For example, for the formula $\phiex$ from \cref{sec:hypersl}, we have $\Vars(\varstratt_1) = \{\varstate_1, \varstate_2, \varstate_3, \varstate_4\}$, $\Vars(\varstratt_2) = \{\varstate_1\}$, $\Vars(\varstratt_3) = \{\varstate_2, \varstate_3, \varstate_4\}$, and $\Vars(\varstratt'_1) = \{\varstate_3, \varstate_4\} = \Vars(\varstratt'_2)$.

	Using these definitions, we can now define the translation of $\phi$ to \HyperSGL as the formula 
	$\phi' := \forall \varstate_1 \ldots \forall \varstate_l .\ \big( \bigwedge_{\istate=1}^{l} \textit{init}_{\varstate_\istate} \implies \translation(\phi) \big)$
	where $\translation$ maps \HyperSLfrag formulas to \HyperSGL formulas as follows:
	\[\begin{array}{ll >{\centering\arraybackslash$} p{1.0cm} <{$} l}
		\translation \colon &\text{\HyperSLfrag} & \to &\HyperSGL \\
		&\ap_{\varpath_r} & \mapsto & 
		\ap_{\varstate_r} \\
		&\neg \psi_1 & \mapsto & 
		\neg \translation(\psi_1) \\
		&\psi_1 \wedge \psi_2 & \mapsto & 
		\translation(\psi_1) \wedge \translation(\psi_2) \\
		&\Next \psi_1 & \mapsto & 
		\Prob(\Next \translation(\psi_1)) = 1 \\
		&\psi_1 \Until \psi_2 & \mapsto & 
		\Prob(\translation(\psi_1) \Until \translation(\psi_2)) = 1\\
		&(\phi_{1})_{\varpath} & \mapsto & \translation(\phi_{1}) \\
		&\psi_1[\varpath_{i_p} : \vec{\varstrat}_{i_p}]_{p=1}^{m} & \mapsto & 
		\translation(\psi_1) 
		\\
		&\exists \varstratt .\ \phi_1 & \mapsto & 
		\llangle 
		\Agt(\varstratt) 
		\rrangle_{\Vars(\varstratt)} \translation(\phi_1) 
		\\
		&\forall \varstratt .\ \phi_1 & \mapsto & 
		\| 
		\Agt(\varstratt) 
		\|_{\Vars(\varstratt)} \translation(\phi_1) 
	\end{array}\]
	
	Note that $\phi'$ is well-formed and that the size of $\phi'$ is polynomial in the size of $\phi$ since the number of state variables $l$ is bounded by the size of $\phi$ and for all strategy variables $\varstratt$ we know that $|\Agt(\varstratt)|=1$ and $|\Vars(\varstratt)| \leq l$.
	
	Recall that \HyperSL is defined over memoryful deterministic strategies.
	For \HyperSGL, we use $\modelsHyperSGLHD$ to denote that we restrict the quantification over strategies to memoryful (i.e., \textbf{h}istory-dependent) \textbf{d}eterministic strategies.
	Let $\tbg = \tbgtup$ be a turn-based non-stochastic game with $\sinit \in \states$, and let $\tbg'$ be like $\tbg$ interpreted as a stochastic game with transition probabilities 0 and 1, where $\sinit$ is additionally labeled with a fresh atomic proposition $\init$.
	Proving $\sinit, \emptymap \modelsHyperSL \phi$ iff $\tbg' \modelsHyperSGLHD \phi'$ corresponds to showing 
	\[\sinit, \emptymap \modelsHyperSL \phi 
	\iff 
	\tbg', \emptymap, [\varstate_\istate \mapsto \sinit \text{ for } \istate=1, \ldots, l] \modelsHyperSGLHD \translation(\phi). \]
	We show the following stronger claims: 
	\begin{enumerate}
		\item \label{enum:path}
		Let $\psi$ be a \HyperSL path formula, 
		$\pathmap$ a path assignment,
		$\mapStrat$ a strategy mapping, 
		and $\statetupMap$ a state variable mapping.
		If the following conditions hold
		\begin{enumerate}
			\item \label{enum:path:pathmap}
			$\pathmap$ is compatible with $\psi$, i.e.,
			\begin{enumerate}
				\item 
				If $\varpath$ is free in $\psi$,
				then $\varpath \in \Dom(\pathmap)$, and
				
				\item 
				If $\varpath$ is bound in $\psi$,
				then $\varpath \not\in \Dom(\pathmap)$, 
			\end{enumerate}
			
			\item \label{enum:path:mapsched} 
			$\mapStrat$ and $\pathmap$ are compatible, i.e.,
			
			$\mapStrat(\varstate_\istate, \agent)(\pathmap(\varpath_\istate)[0, z]) = \action$ such that $\Trans(\pathmap(\varpath_\istate)[z], \action, \pathmap(\varpath_\istate)[z{+}1]) = 1$ for all $z \in \N$, $\varpath_\istate \in \Dom(\pathmap)$ and $\agent \in \agents$,
			and
			
			\item \label{enum:path:statetup} 
			$\statetupMap$ and $\pathmap$ are compatible with respect to $\psi$, i.e.,
			
			if $\varpath_\istate$ occurs in $\psi$, then $\statetupMap(\varstate_\istate) = \pathmap(\RefVar(\psi, \varpath_\istate))[0]$,
		\end{enumerate}
		then it holds that
		\[\pathmap \modelsHyperSL \psi \iff \tbg', \mapStrat, \statetupMap \modelsHyperSGLHD \translation(\psi) . \] 
		
		\item \label{enum:state}
		Let $\phi_1$ be a \HyperSL state formula of the form
		$Q_{1} \varstratt_{1} \ldots Q_{n} \varstratt_{n} \psi_1[\varpath_{i_p} : \vec{\varstrat}_{i_p}]_{p=1}^{m}$,
		$\mapStrat$ a strategy mapping,
		$s$ a state in $\tbg$,
		and
		$\statetupMap$ a state variable mapping (for $\tbg'$).
		If the following condition holds
		\begin{enumerate}
			\item \label{enum:state:statetup}
			$\statetupMap$ and $\state$ are compatible with respect to $\phi_1$, i.e.,
			
			if $\varpath_\istate$ occurs in $\psi_1[\varpath_{i_p} : \vec{\varstrat}_{i_p}]_{p=1}^{m}$, then $\statetupMap(\varstate_\istate) = s$
		\end{enumerate}
		then it holds that
		\[ s, \emptymap \modelsHyperSL \phi_1 \iff \tbg', \mapStrat, \statetupMap \modelsHyperSGLHD T(\phi_1) . \]
	\end{enumerate}
	We proceed by induction on the structure of $\psi$.
	
	\paragraph{If $\psi := \ap_{\varpath_\istate}$:} 
	Fix $\pathmap$, $\mapStrat$, $\statetupMap$ satisfying the conditions of claim \ref{enum:path}. 
	Then, 
	\[\begin{array}{rl}
	\pathmap \modelsHyperSL \psi 
	&\xLeftrightarrow{\text{Def.}} \ap \in \labelingfct(\pathmap(\varpath_\istate)[0]) \\
	&\xLeftrightarrow{\text{Ass.}} \ap \in \labelingfct(\statetupMap(\varstate_\istate)) \\
	&\xLeftrightarrow{\text{Def.}} \tbg', \mapStrat, \statetupMap \modelsHyperSGLHD \translation(\psi) . 
	\end{array}\]

	\paragraph{For Boolean operators,} the claim can be easily verified.
	
	\paragraph{If $\psi := \Next \psi_1$:}
	Fix $\pathmap$, $\mapStrat$, $\statetupMap$ satisfying the conditions of claim \ref{enum:path} w.r.t.\ $\psi$.
	Let $\dtmc = \bigtimes_{\varstate_\istate \in Dom(\statetupMap)} (\tbg')^{\strat_\istate}$ for $\strat_\istate = \bigoplus_{\agent \in \agents} \mapStrat(\varstate_\istate, \agent)$. 
	Since $\tbg'$ is non-stochastic, there is a unique path $\pif \in \Paths[\dtmc]$ from the unique state corresponding to $\statetupMap$ in $\dtmc$.
	Let $\statetupMap'$ be the state variable assignment corresponding to the composed state $\pif[1]$ in $\dtmc$. 
	Then, $\psi_1$, $\pathmap[1, \infty]$, $\mapStrat^{\pif[0, 1]}$, $\statetupMap'$ satisfy the conditions of the induction hypothesis:
	\begin{itemize}
		\item Condition \ref{enum:path:pathmap}: Since $\pathmap$ is compatible with $\psi$, it follows directly that $\pathmap[1, \infty]$ is compatible with $\psi_1$. 
		
		\item Condition \ref{enum:path:mapsched}: 
		Let $z \in \N$ with $z>0$, $\varpath_\istate \in \Dom(\pathmap)$ and $\agent \in \agents$. 
		By definition and induction hypothesis, it holds that
		\[\mapStrat^{\pif[0, 1]}(\varstate_\istate, \agent)(\pathmap[1, \infty](\varpath_\istate)[0, z{-}1]) = \mapStrat(\varstate_\istate, \agent)(\pathmap(\varpath_\istate)[0, z]) = \action\] for a unique $\action \in \Act$ with $\Trans(\pathmap(\varpath_\istate)[z], \action, \pathmap(\varpath_\istate)[z{+}1]) = 1$.
		Hence, it follows that also $\Trans(\pathmap[1, \infty](\varpath_\istate)[z{}-1], \action, \pathmap[1, \infty](\varpath_\istate)[z]) = 1$ and thus $\mapStrat^{\pif[0, 1]}$ and $\pathmap[1, \infty]$ are compatible.
		
		\item Condition \ref{enum:path:statetup}: Assume $\varpath_\istate$ occurs in $\psi_1$. Then $\varpath_\istate$ also occurs in $\psi$ and $\statetupMap'(\varstate_\istate)$ is the unique successor state of $\statetupMap(\varstate_\istate)$ in $\tbg'$ under $\mapStrat$. Since $\statetupMap$ and $\pathmap$ are compatible with respect to $\psi$, it thus follows that $\statetupMap'$ and $\pathmap[1, \infty]$ are compatible with respect to $\psi_1$.
	\end{itemize}

	Based on these observations, we can conclude that
	\[\begin{array}{rl}
		\pathmap \modelsHyperSL \psi
		\xLeftrightarrow{\text{Def.\ Sem.}} & \pathmap[1, \infty] \modelsHyperSL \psi_1 \\
		\xLeftrightarrow{\text{IH}} & \tbg', \mapStrat^{\pif[0, 1]}, \statetupMap' \modelsHyperSGLHD \translation(\psi_1) \\
		\xLeftrightarrow{\text{Def.\ Sem.}} & \tbg', \mapStrat, \statetupMap, \pif \modelsHyperSGLHD \Next \translation(\psi_1) \\
		\xLeftrightarrow{\tbg' \text{ non-stoch.}} & \tbg', \mapStrat, \statetupMap \modelsHyperSGLHD \Prob(\Next T(\psi_1)) = 1 \\
		\xLeftrightarrow{\text{Def.\ }\translation} & \tbg', \mapStrat, \statetupMap \modelsHyperSGLHD \translation(\psi) .
	\end{array}\]

	\paragraph{If $\psi := \psi_1 \Until \psi_2$:}
	The reasoning follows analogously.

	\paragraph{If $\psi := (\phi_1)_{\varpath}$:} 
	Fix some $\pathmap$, $\mapStrat$, $\statetupMap$ satisfying the conditions of claim \ref{enum:path} w.r.t.\ $\psi$. 
	By definition of the semantics we have 
	$\pathmap \modelsHyperSL (\phi_1)_{\varpath} 
	\iff \pathmap(\varpath)[0], \emptymap \modelsHyperSL \phi_1$.
	Assume $\varpath_\istate$ occurs in $\phi_1$, then $\RefVar(\psi, \varpath_\istate) = \varpath$ and by condition \ref{enum:path:statetup} we know that $\statetupMap(\varstate_\istate) = \pathmap(\varpath)[0]$. 
	Thus, $\statetupMap$ and $\pathmap(\varpath)[0]$ are compatible w.r.t.\ $\phi_1$ (condition \ref{enum:state:statetup}) and we can apply the induction hypothesis, yielding
	\[\begin{array}{lrl}
		\pathmap \modelsHyperSL (\phi_1)_{\varpath} 
		& \xLeftrightarrow{\text{Def.\ Sem.}} & \pathmap(\varpath)[0], \emptymap \modelsHyperSL \phi_1 \\
		& \xLeftrightarrow{\text{IH}} & \tbg', \mapStrat, \statetupMap \modelsHyperSGLHD \translation(\phi_1) \\
		& \xLeftrightarrow{\text{Def.\ T}} & \tbg', \mapStrat, \statetupMap \modelsHyperSGLHD \translation(\psi) .
	\end{array}\]

	\paragraph{If $\psi := \phi_1 = Q_{1} \varstratt_{1} \ldots Q_{n} \varstratt_{n} .\ \psi_1[\varpath_{i_p} : \vec{\varstrat}_{i_p}]_{p=1}^{m}$:}
	Fix $\mapStrat$, $\state$, $\statetupMap$ satisfying the conditions of claim \ref{enum:state} w.r.t.\ $\phi_1$.
	In the following, we let $\setofagents_{t} := \{\Agt(\varstratt_{t})\}$ for $t=1, \ldots, n$.
	By definition of the semantics we have   
	\[\begin{array}{lcl}
		&  & \state, \emptymap \modelsHyperSL \phi_1 \\[-12pt]
		& \iff & Q_{1} \strat_{1} \in \Strats[\tbg][\setofagents_{1}] \ldots Q_{n} \strat_{n} \in \Strats[\tbg][\setofagents_{n}] .\ 
		\state, \overbrace{[\varstratt_{t} \mapsto \strat_{t}]_{t=1}^{n}}^{\Deltafixed}
		\modelsHyperSL 
		\overbrace{\psi_1[\varpath_{i_p} : \vec{\varstrat}_{i_p}]_{p=1}^{m}}^{\phipath} \\
		& \iff & Q_{1} \strat_{1} \in \Strats[\tbg][\setofagents_{1}] \ldots Q_{n} \strat_{n} \in \Strats[\tbg][\setofagents_{n}] .\ 
		\pathmapfixed
		\modelsHyperSL 
		\psi_1 
	\end{array}\]
	where 
	$\pathmapfixed(\varpath_{i_p}) = Play_{\tbg}(\state, \bigoplus_{\agent \in \agents} \Deltafixed(\vec{\varstrat}_{i_p}(\agent)))$ for $p=1, \ldots, m$.
	
	In order to be able to apply the induction hypothesis, we need to construct $\mapStratfixed$ and $\statetupMapfixed$ satisfying conditions \ref{enum:path:pathmap}-\ref{enum:path:statetup} with respect to $\psi_1$ and $\pathmapfixed$.
	For a given $\vec{\strat} \in \Strats[\tbg][\setofagents_{1}] \times \ldots \times \Strats[\tbg][\setofagents_{n}]$, we 
	let $\statetupMapfixed = \statetupMap$ and 
	define the update $\mapStratfixed$ of $\mapStrat$ by the strategies in $\vec{\strat}$ as 
	\[\begin{array}{lll}
		\mapStratfixed(\varstate_{\istate}, \agent) &:= \begin{cases}
			\Deltafixed(\vec{\varstrat}_{i_p}(\agent)) & 
			\text{if there ex.\ } p \text{ with } \RefVar(\phipath, \varpath_\istate) = \varpath_{i_p}
			\\
			\mapStrat(\varstate_\istate, \agent) & \text{otherwise}
		\end{cases} 
	\end{array}\]
	for $\istate=1, \ldots, l$ and $\agent\in \agents$.
	
	Let us now show that these choice satisfy conditions \ref{enum:path:pathmap}-\ref{enum:path:statetup}.
	We show the two parts of condition \ref{enum:path:pathmap} separately:
	\begin{itemize}
		\item If $\varpath_\istate$ is free in $\psi_1$, then 
		$\istate \in \{i_1, \ldots, i_m\}$ and thus $\varpath_\istate \in \Dom(\pathmapfixed)$.
		
		\item If $\varpath_\istate$ is bound in $\psi_1$, 
		then
		$\istate \not\in \{i_1, \ldots, i_m\}$ and thus $\varpath_\istate \not\in \Dom(\pathmapfixed)$.
	\end{itemize}
	Condition \ref{enum:path:mapsched} holds since for
	$p=1, \ldots, m$, $z \in \N$, and the unique $\agent \in \agents$ with $\pathmapfixed(\varpath_{i_p})[z] \in \states_\agent$ we have
	\[\mapStratfixed(\varstate_{i_p}, \agent)(\pathmapfixed(\varpath_{i_p})[0, z]) = \Deltafixed(\vec{\varstrat}_{i_p}(\agent))(\varpath_{i_p})[0, z] . \]  
	Lastly, condition \ref{enum:path:statetup} holds with the following reasoning:	
	If $\varpath_\istate$ occurs in $\psi_1$, then $\RefVar(\psi_1, \varpath_\istate) \in \{\varpath_{i_1}, \ldots, \varpath_{i_m}\}$ and by construction $\pathmapfixed(\RefVar(\psi_1, \varpath_\istate))[0] = \state$.
	By assumption (condition \ref{enum:state:statetup}) on $\statetupMap$ and $\state$, we have $\statetupMap(\varstate_\istate) = \state$.
	Putting everything together, we get $\statetupMapfixed(\varstate_\istate) = \statetupMap(\varstate_\istate) = \state = \pathmapfixed(\RefVar(\psi_1, \varpath_\istate))[0]$.
	
	Hence, by induction hypothesis and the definition of the semantics of \HyperSGL, it follows that
	\[\begin{array}{rl}
		&Q_{1} \strat_{1} \in \Strats[\tbg][\setofagents_{1}] \ldots Q_{n} \strat_{n} \in \Strats[\tbg][\setofagents_{n}] .\ 
		\pathmapfixed
		\modelsHyperSL 
		\psi_1 \\
		\iff &Q_{1} \strat_{1} \in \Strats[\tbg][\setofagents_{1}] \ldots Q_{n} \strat_{n} \in \Strats[\tbg][\setofagents_{n}] .\ 
		\tbg', \mapStratfixed, \statetupMapfixed
		\modelsHyperSGLHD 
		\translation(\psi_1) 
		\\
		\iff &
		\tbg', \mapStrat, \statetupMap \modelsHyperSGLHD 
		\llcombined 
		\setofagents_{1}
		\rrcombined^{Q_{1}}_{\Vars(\varstratt_{1})} 
		\ldots 
		\llcombined 
		\setofagents_{n} 
		\rrcombined^{Q_{n}}_{\Vars(\varstratt_{n})} 
		\translation(\psi)
	\end{array}	\]
	where $\llcombined \ldots \rrcombined^{\exists} = \existsStrat{\ldots}{}$ and $\llcombined \ldots \rrcombined^{\forall} = \forallStrat{\ldots}{}$.
	
	Thus, $\pathmap \modelsHyperSL (\phi_1)_{\varpath} \iff \tbg', \mapStrat, \statetupMap \modelsHyperSGLHD \translation((\phi_1)_{\varpath})$.
	\hfill\qed
\end{proof}
	\section{Bounded Memory (\cref{th:bounded-mem-decidable} from \cref{sec:mc})}
\label{app:bounded-mem}

\boundedMemDecidable*

\begin{figure}[b]
	\scalebox{0.95}{
		\begin{minipage}{1.04\linewidth}
			\begin{algorithm}[H]
				\caption{Main SMT encoding algorithm for memoryless strategies}
				\label{alg:main}
				\label{alg:truth}
				
				\KwIn{$\tbg=\tbgtup$: TSG, \newline
					$\phi = \phifull$: 
					\HyperSGL formula. 
					}
				\KwOut{
					Whether $\tbg$ satisfies $\phi$ over memoryless probabilistic strategies.
				}
				\Fn{\FMain{$\tbg, \phifull$}}{
					\ForEach(\tcp*[f]{encode state quantifiers}){$i=1, \ldots, \numstatequant$}{
						\label{line:truthstart}
						\label{line:truthstatestart}
						\lIfElse{$Q_{i} {=} \forall$}
						{$E_i := \bigwedge_{s_i \in S}$}
						{$E_i := \bigvee_{s_i \in S}$}\label{line:truthstateend}
					}
					$ E_{\textit{tru}} := E_1 \ldots E_{\numstatequant} \
					\Bigl(
					\holdsWith[(s_1, \ldots, s_\numstatequant)][\pschedq]
					\Bigr)$\tcp*{truth of input formula}
					\label{line:truthend}
					$ E_{\textit{sem}}, \_ := \FSem(\tbg, \pschedq, \numstatequant, 0, [])$ 
					\tcp*{semantics of $\pschedq$}\label{line:sem}
					\lIfElse{$check(E_{\textit{tru}}\wedge E_{\textit{sem}} )=\mathit{SAT}$}{\Return \textit{TRUE}}{\Return \textit{FALSE}}
				}
			\end{algorithm}
		\end{minipage}
	}
\end{figure}

\begin{figure}[p]
	\scalebox{0.95}{
		\begin{minipage}{1.04\linewidth}
			\begin{algorithm}[H]
				\caption{SMT encoding for the meaning of strategy- or non-quantified \HyperSGL formulas}
				\label{alg:semEnc}
				
				\KwIn{%
					$\tbg = \tbgtup$: TSG, \newline
					$\phi$: Strategy- or non-quantified \HyperSGL formula, \newline
					$\numstatequant$: Number of state quantifiers, \newline
					$\countQuant$: Counter for indexing strategy quantifications, \newline
					$\currMap$: Mapping pairs of state variable indices and agents to strategy quantification indices.
				}
				\KwOut{SMT encoding of the meaning of $\phi$ for $\tbg$.}
				
				\Fn{\FSem{$\tbg, \phi, \numstatequant, \countQuant, \currMap$}}{
					\If{$\phi$ is $\existsStrat{\setofagents}{\setofvariables} \phi'$}{
						\label{line:stratEnc}
						{$\countQuant := \countQuant + 1$\;}
						\tcp{In `experiment' $i$, agent $\agent$ should now use the strategy 
							encoded by the variables $\stratencWith[\countQuant][\cdot][\cdot]$
						}
						\lForEach{$(\varstate_i, \agent) \in \setofvariables \times \setofagents$}{
							$\currMap[i, \agent] := \countQuant$
						}
						\label{line:stratEncMap}
						\tcp{Translate strategy quantif.\ to SMT variable quantif.\ }
						{$E_{q} := \ $\enquote{ }\;}
						\label{line:stratEncQuantStart}
						\lForEach{$(\state, \action) \in \states_{\setofagents} \times \Act$ with $\action \in \Act(\state)$}{
							$E_q := E_q\ \exists \stratencWith[\countQuant][][]$
						}
						\label{line:stratEncQuantEnd}
						{$E_{\textit{str}}:= \bigwedge_{\state \in \states_\setofagents} 
							\left(\bigwedge_{\action \in \Act(\state)} 0\, {\leq}\, \stratencWith[\countQuant][][]\, {\leq}\, 1\right) \wedge 
							\sum_{\action \in \Act(\state)} \stratencWith[\countQuant][][] {=} 1 $\;}
						\label{line:stratEncValid}
						{$E', \countQuant' := \FSem(\tbg, \phi', \numstatequant, \countQuant, \currMap)$\;}
						\label{line:stratEncE'}
						{$E := E_q\ (E_{\textit{str}} \wedge E') \wedge 
							\bigwedge_{\statetup \in \states^{\numstatequant}} \holdsWith \iff \holdsWith[][\phi']
						$\;}
						\label{line:stratEncE}
						{\Return $E, \countQuant'$\;}
						\label{line:stratEncEnd}
					}
					\lElseIf(\tcp*[f]{Analogously to the existential case}){$\phi$ is $\forallStrat{\setofagents}{\setofvariables} \phi'$}{
						\ldots
						\label{line:stratEncUnivEnd}
					}
					\ElseIf{$\phi$ is $\true$}{
						{$E := \bigwedge_{\statetup \in \states^{\numstatequant}} \holdsWith$\;}
						{\Return $E, \countQuant$\;}
					}
					\ElseIf{$\phi$ is $a_{\varstate_i}$}{
						$E := (\bigwedge_{\statetup \in \states^{\numstatequant}, a \in L(s_{i})} \holdsWith) \wedge (\bigwedge_{\statetup \in \states^{\numstatequant}, a \not\in L(s_{i})} \neg\holdsWith) $\;
						\Return $E, \countQuant$\;
					}
					\ElseIf{$\phi$ is $\phi_1 \wedge \phi_2$}{
						{$E_1, \countQuant_1 : = \FSem(\tbg, \phi_1, \numstatequant, \countQuant, \currMap)$\;}
						{$E_2, \countQuant_2 : = \FSem(\tbg, \phi_2, \numstatequant, \countQuant_1, \currMap)$\;}
						{$E := E_1 \wedge E_2 \wedge 
							\bigwedge_{\statetup \in \states^{\numstatequant}} (\holdsWith \iff (\holdsWith[][\phi_1] \wedge \holdsWith[][\phi_2]))$\;}
						{\Return $E, \countQuant_2$\;}
					}
					\ElseIf{$\phi$ is $\neg\phi'$}{
						$E', \countQuant' : = \FSem(\tbg, \phi', \numstatequant, \countQuant, \currMap)$\;
						$E := E' \wedge \bigwedge_{\statetup \in \states^\numstatequant} (\holdsWith \oplus \holdsWith[][\phi'])$\;
						\Return $E, \countQuant'$\;
					}
					\ElseIf(\tcp*[f]{$\sim \in \{\leq,<,=,>,\geq\}$}){$\phi$ is $\phi_1 \sim \phi_2$}{
						{$E_1, \countQuant_1 := \FSemProb(\tbg, \phi_1, \numstatequant, \countQuant, \currMap)$\;}
						{$E_2, \countQuant_2 := \FSemProb(\tbg, \phi_2, \numstatequant, \countQuant_1, \currMap)$ \;}
						{$E := E_1 \wedge E_2 \wedge \bigwedge_{\statetup \in \states^{\numstatequant}} 
							(\holdsWith \iff (\probWith[][\phi_1] \sim \probWith[][\phi_2]))$\;}
						{\Return $E, \countQuant_2$\;}
					}
					\ElseIf{$\phi$ is $\varstate_i = \varstate_j$}{
						\label{line:s=s'start}
						{$E := 
							(\bigwedge_{\statetup \in \states^{\numstatequant}, s_i = s_j} \holdsWith) \wedge 
							(\bigwedge_{\statetup \in \states^{\numstatequant}, s_i \neq s_j} \neg\holdsWith)$\;}
						{\Return $E, \countQuant$\;}
						\label{line:s=s'end}
					}
				}
			\end{algorithm}
		\end{minipage}
	}
\end{figure}

\begin{figure}[t]
	\scalebox{0.95}{
		\begin{minipage}{1.04\linewidth}
			\begin{algorithm}[H]
				\caption{SMT encoding for the meaning of \HyperSGL probability expressions}
				\label{alg:semEncProb}
				
				\KwIn{%
					$\tbg = \tbgtup$: TSG, \newline
					$\phi$: \HyperSGL probability epxression, \newline
					$\numstatequant$: Number of state quantifiers, \newline
					$\countQuant$: Counter for indexing strategy quantifications, \newline
					$\currMap$: Mapping pairs of state variable indices and agents to strategy quantification indices.
				}
				\KwOut{SMT encoding of the meaning of $\phi$ for $\tbg$.}
				
				\Fn{\FSemProb{$\tbg, \phi, \numstatequant, \countQuant, \currMap$}}{
					\If{$\phi$ is $c$}{
						\label{line:nextstart}
						{$E := \bigwedge_{\statetup \in \states^{\numstatequant}} (\probWith = c)$\;}
						{\Return $E, \countQuant$\;} 
					}
					\ElseIf{$\phi$ is $f(\phi_1, \ldots, \phi_k)$}{
						{$\countQuant_0 := \countQuant$\;}
						\ForEach{$i = 1, \ldots, k$}{
							{$E_i, \countQuant_i := \FSemProb(\tbg, \phi_i, \numstatequant, \countQuant_{i-1}, \currMap)$\;}
						}
						{$E := (\bigwedge_{i=1}^{k} E_i) \wedge 
							(\bigwedge_{\statetup \in \states^{\numstatequant}} (\probWith\ {=}\ f(\probWith[][\phi_1], \ldots, \probWith[][\phi_k])))$\;}
						{\Return $E, \countQuant_k$\;}
					}
					\ElseIf{$\phi$ is $\Prob(\Next \phi')$}{
						{$E', \countQuant' := \FSem(\tbg, \phi', \numstatequant, \countQuant, \currMap)$\; }
						\label{line:nextE'}
						{$E := E' \wedge \bigwedge_{\statetup \in \states^\numstatequant} \Bigl((\holdsToIntWith[][\phi'] {=} 1 \wedge \holdsWith[][\phi']) \vee (\holdsToIntWith[][\phi'] {=} 0 \wedge \neg\holdsWith[][\phi'])\Bigr)$\;}
						\label{line:nextToInt}
						\ForEach{$\statetup = (s_1, \ldots, s_\numstatequant) \in \states^{\numstatequant}$}{
							\tcp{Retrieve index of currently relevant strategy for each $s_i$}
							{$\currInd := [\currMap[i, \agents(s_i)] \textit{ for } i=1, \ldots, \numstatequant]$\;}
							\label{line:nextRetrieve}
							{$E := E \wedge \displaystyle 
								\probWith = \displaystyle
								\sum_{\actiontup \in \Act(\statetup)^\numstatequant}
								\sum_{\statetup' \in \fsucc(\statetup, \actiontup)}
								\left(\prod_{i=1}^\numstatequant
								\stratencWith[{\currInd[i]}][\state_i][\action_i]
								\cdot
								\Trans(\state_i, \action_i, \state'_i)\right)
								\cdot
								\holdsToIntWith[\statetup'][\phi']
								$\;}
							\label{line:nextE}
						}
						{\Return $E, \countQuant'$\;}
						\label{line:nextend}
					}
					\ElseIf{$\phi$ is $\Prob(\phi_1 \Until \phi_2)$}{
						{\Return $\FSemProbUntil(\tbg, \phi, \numstatequant, \countQuant, \currMap)$\;}
					}	
				}
		\end{algorithm}
	\end{minipage}
	}
\end{figure}

\begin{figure}[t]
	\scalebox{0.95}{
	\begin{minipage}{1.04\linewidth}
		\begin{algorithm}[H]
			\caption{SMT encoding for the meaning of \HyperSGL probability expressions involving Until operators}
			\label{alg:semEncProbUntil}
			
			\KwIn{%
				$\tbg = \tbgtup$: TSG, \newline
				$\phi := \Prob(\phi_1 \Until \phi_2)$: \HyperSGL probability epxression with Until, \newline
				$\numstatequant$: Number of state quantifiers, \newline
				$\countQuant$: Counter for indexing strategy quantifications, \newline
				$\currMap$: Mapping pairs of state variable indices and agents to strategy quantification indices.
			}
			\KwOut{SMT encoding of the meaning of $\phi$ for $\tbg$.}
			
			\Fn{\FSemProbUntil{$\tbg, \phi, \numstatequant, \countQuant, \currMap$}}{
				\label{line:untilStart}
				{$E_1, \countQuant_1 := \FSem(\tbg, \phi_1, \numstatequant, \countQuant, \currMap)$\;}
				{$E_2, \countQuant_2 := \FSem(\tbg, \phi_2, \numstatequant, \countQuant_1, \currMap)$\;}
				{$E := E_1 \wedge E_2$\;}
				\ForEach{$\statetup=(s_1, \ldots, s_\numstatequant)\in \states^\numstatequant$}{
					\tcp{Base cases where the probability of $\phi_1 \Until \phi_2$ is 0 or 1}
					{$E := E \wedge (\holdsWith[][\phi_2] \implies \probWith{=}1)\wedge \bigl((\neg\holdsWith[][\phi_1] \wedge \neg\holdsWith[][\phi_2]) \implies \probWith{=}0 \bigr) $\;}
					\label{line:untilBase}
					\tcp{Otherwise, the probability depends on the successors}
					{$\currInd := [\currMap[\state_i, \agents(s_i)] \textit{ for } i=1, \ldots, \numstatequant]$\;} 
					{$E_{pr} := \probWith = \displaystyle
						\sum_{\actiontup \in \Act(\statetup)^\numstatequant}\;
						\sum_{\statetup' \in \fsucc(\statetup, \actiontup)}\;
						\Big(\prod_{i=1}^\numstatequant
						\stratencWith[{\currInd[i]}][\state_i][\action_i]
						\cdot
						\Trans(\state_i, \action_i, \state'_i)\Big)
						\cdot \probWith[\statetup'][\phi]$\;}
					\label{line:untilEpr}
					\tcp{Loop condition: Ensure that a $\phi_2$-state will be reached with positive probability if $\probWith{>}0$}
					{$E_{loop} := 
						\probWith{>}0 \implies 
						\displaystyle
						\bigvee_{\actiontup \in \Act(\statetup)^\numstatequant} 
						\bigvee_{\statetup' \in \fsucc(\statetup, \actiontup)}
						\Bigl(\displaystyle 
						\prod_{i=1}^\numstatequant
						\stratencWith[{\currInd[i]}][\state_i][\action_i]
						{>}0
						\wedge 
						\Big(\holdsWith[\statetup'][\phi_2] 
						\vee 
						d_{\statetup,\phi_2} {>} d_{\statetup',\phi_2} \Big)\Bigr)
						$\;}
					\label{line:untilEloop}
					{$E := E \wedge
						\Biggl[\Bigl[
						\holdsWith[][\phi_1] \wedge 
						\neg\holdsWith[][\phi_2] 
						\Bigr] 
						\implies 
						\biggl[
						E_{pr}
						\;\wedge$ $
						E_{loop}
						\biggr]\Biggr]$\;}
					\label{line:untilE}
				}
				{\Return $E, \countQuant_2$\;}
				\label{line:untilEnd}
			}
		\end{algorithm}
	\end{minipage}
	}
\end{figure}				

\begin{proof}
	We can encode the bounded-memory model-checking problem in non-linear real arithmetic (NRA), analogously to \HyperPCTL~\cite{abrahamProbabilisticHyperproperties2020} and \AHyperPCTL \cite{gerlachIntroducingAsynchronicity2023a}. 
	NRA is decidable in time doubly exponential in the number of variables~\cite{collinsQuantifier1975,davenportVarietiesDoublyExponential2021}.
	We first present how to create the encoding for memoryless (probabilistic) strategies and then explain how this approach can be extended to bounded memory. 
	
	\paragraph{Memoryless Strategies.}
	\cref{alg:truth,alg:semEnc,alg:semEncProb,alg:semEncProbUntil} create an SMT encoding $E$ for a given TSG $\tbg$ and \HyperSGL formula $\phifull$ such that $E$ is satisfiable iff $\tbg \models \phifull$ over memoryless probabilistic strategies.
	We use 
	\begin{itemize}
		\item Boolean variables $\holdsWith$ to encode the evaluation of a strategy- or non-quantified subformula $\phi$ of $\pschedq$ from state tuple $\statetup \in \states^\numstatequant$,
		
		\item real variables $\holdsToIntWith$ to encode the integer representation of $\holdsWith$, and
		
		\item real variables $\probWith$ encoding the evaluation of a probability expression subformula $\phi$ from state tuple $\statetup \in \states^\numstatequant$.
	\end{itemize}
	
	\textbf{\cref{alg:truth}} translates each existential (or universal) state quantification to a conjunction (or disjunction) over all states and recursively encodes the semantics of $\pschedq$ in $\tbg$ using~\cref{alg:semEnc,alg:semEncProb,alg:semEncProbUntil}. 
	
	\textbf{\cref{alg:semEnc}} encodes the meaning of non-state-quantified \HyperSGL formulas.
	For strategy- or non-quantified formulas, this works the same as for \HyperPCTL, with the addition of strategy quantification and comparison of state variables.
	The encoding of subformulas $\varstate_i = \varstate_j$ is straightforward (Lines \ref{line:s=s'start}-\ref{line:s=s'end}).
	For the encoding of strategy quantification (Lines \ref{line:stratEnc}-\ref{line:stratEncUnivEnd}), we keep a counter $\countQuant$ in order to give a unique index to each strategy quantification. 
	We mimic the $j^{th}$ strategy quantification $\existsStrat{\setofagents_j}{\setofvariables}$ (or $\forallStrat{\setofagents_j}{\setofvariables}$) by quantifying over variables $\stratencWith$ for $\state \in \states_{\setofagents_j}, \action \in \Act(\state)$ (Lines \ref{line:stratEncQuantStart}-\ref{line:stratEncQuantEnd}) and specifying that they should encode a valid strategy (\cref{line:stratEncValid}).
	We further keep an array $\currMap$ to store that all $\varstate \in \setofvariables$ should use  the strategy encoded by the variables $\stratencWith$ 
	for $\agent \in \setofagents_j$ from now on, until this assignment is possibly overwritten by a nested strategy quantification (\cref{line:stratEncMap}). 
	Note that $\countQuant$ and $\currMap$ are part of the book-keeping for the creation of the SMT encoding, not encoded as SMT variables themselves. 
	In particular, we can decide `statically' which strategy variables have to be used in the encoding of the semantics of a subformula.
	We show only the existential case explicitly here; the universal case is analogous with corresponding changes in \cref{line:stratEncQuantEnd,line:stratEncE}.

	\textbf{\cref{alg:semEncProb,alg:semEncProbUntil}} encode the meaning of \HyperSGL probability expressions.
	The encoding of constants and arithmetic functions is straightforward.
	Recall that we only allow standard arithmetic operations like addition and multiplication for the arithmetic operator $f$. Those operations are also allowed in non-linear real arithmetic.
	
	For the encoding of the Next and the Until operator, we simplify notation by defining 
	$\Act(\statetup) = \bigtimes_i \Act(\state_i)$ and
	$\fsucc(\statetup, \actiontup)$ as the set of all $\statetup' = (s'_1, \ldots, s'_\numstatequant) \in \states^{\numstatequant}$ such that $\prod_{i=1}^{\numstatequant} \Trans(\state_i, \action_i, \state'_i) > 0$, for $\statetup = (s_1, \ldots, s_\numstatequant) \in \states^{\numstatequant}$ and $\actiontup = (\action_1, \ldots, \action_\numstatequant) \in \Act^\numstatequant$.
	Further, for $\state \in \states$, we use $\agents(s)$ to denote the unique agent $\agent \in \agents$ such that $\state \in \states_\agent$.
	The differences of our encoding to the encodings for \HyperPCTL or \AHyperPCTL~\cite{abrahamProbabilisticHyperproperties2020,gerlachIntroducingAsynchronicity2023a} stem from 
	(1)~the extension to stochastic games, meaning we have to track strategies for several agents,
	(2)~allowing nested strategy quantification, and
	(3)~considering probabilistic schedulers.
	
	For $\phi := \Prob(\Next \phi')$, we first encode the semantics of $\phi'$ (Alg.\ \ref{alg:semEncProb}, \cref{line:nextE'}) and enforce that $\holdsToIntWith[][\phi']$ equals the integer representation of $\holdsWith[][\phi']$ (Alg.\ \ref{alg:semEncProb}, \cref{line:nextToInt}).
	For each state tuple $\statetup \in \states^\numstatequant$, we encode the probability of $\Next \phi'$ as the sum over all action tuples $\actiontup \in \Act^\numstatequant$ and successor state tuples $\statetup'$ of the product of
	(1) the probability of choosing $\actiontup$ in $\statetup$ ($\prod_{i=1}^\numstatequant
	\stratencWith[{\currInd[i]}][\state_i][\action_i]$),
	(2) the probability of going from $\statetup$ to $\statetup'$ via $\actiontup$ ($\prod_{i=1}^{\numstatequant} \Trans(\state_i,\action_i,\state'_i)$), 
	and
	(3) the integer representation of whether $\phi'$ holds in $\statetup'$ ($\holdsToIntWith[\statetup'][\phi']$).
	For (1), we need to retrieve $\currInd[i] = \currMap[\state_i, \agents(s_i)]$, the index of the currently relevant strategy variable for state variable $\varstate_i$ and agent $\agent = \agents(\state_i)$, for $i=1,\ldots,\numstatequant$ (Alg.\ \ref{alg:semEncProb}, \cref{line:nextRetrieve}).

	For $\phi := \Prob(\phi_1 \Until \phi_2)$, the encoding works similarly.
	We first encode the semantics of $\phi_1$ and $\phi_2$ and then encode the evaluation of $\phi$ for every state tuple.
	In \cref{alg:semEncProbUntil}, \cref{line:untilBase}, we encode the base cases where it is known that $\phi_1 \Until \phi_2$ holds with probability 0 or 1.
	If the base cases do not hold, we define the evaluation of $\phi$ at $\statetup \in \states^\numstatequant$ 
	by encoding the linear equation system characterizing $\Prob(\phi_1 \Until \phi_2)$ (Alg.\ \ref{alg:semEncProbUntil}, \cref{line:untilEpr})
	and adding a loop condition to ensure that we only accept the least solution of the equation system, employing real variables $d_{\statetup,\phi_2}$ (Alg.\ \ref{alg:semEncProbUntil}, \cref{line:untilEloop}). This approach was introduced in the encoding for \HyperPCTL~\cite{abrahamProbabilisticHyperproperties2020}, for details on the least fixed point characterization we refer to~\cite{baierPrinciplesModel2008}.

	Note that only the variables encoding the strategies are explicitly quantified, while all remaining variables are implicitly existentially quantified.
	The constructed SMT encoding is hence not necessarily in prenex normal form, but can easily be transformed into it. 
	Note further that the non-linearity of the encoding stems from the fact that we consider probabilistic strategies.

	\paragraph{Bounded Memory.}
	We can extend the SMT encoding from the memoryless case to bounded memory, similarly to how the SMT encoding for \AHyperPCTL treats a specific kind of bounded-memory `stutter-schedulers'~\cite{gerlachIntroducingAsynchronicity2023a}. 
	Let $\tbg = \tbgtup$ be a turn-based game and $\phi := \phifull$ be a \HyperSGL formula.
	For memory bound $k \in \N$, we assume $\modes_j = \modes = \{0, \ldots, k-1\}$ for every strategy quantifier index $j$. 
	The state space of the induced product DTMC over which $\pschedq$ is evaluated is $(\states \times \modes)^\numstatequant$ (in contrast to $\states^\numstatequant$ for memoryless schedulers).
	For every strategy quantifier, 
	we additionally introduce pseudo-Boolean variables $\modefencWith$ encoding the mode transition choice, where $\modefencWith = 1$ represents $\modef_j(\mode,\state) = \mode'$, for $\state \in \states, \mode, \mode' \in \modes_j$, and enforce that for every $\state \in \states, \mode \in \modes_j$ it holds that $\sum_{\mode' \in \modes_j} \modefencWith = 1$.
	We also need to 
	adjust the strategy encoding to variables $\stratencWith[][\state,\mode][]$ encoding the probability of choosing $\alpha \in \Act$ in $\state \in \states$ when the mode is $\mode \in \modes_j$. 
	Further, the variables $\holdsWith, \holdsToIntWith, \probWith, d_{\statetup, \phi}$ need to be replaced by variables $\holdsWith[\statetup, \modetup], \holdsToIntWith[\statetup, \modetup], \probWith[\statetup, \modetup], d_{\statetup, \modetup, \phi}$ for $\statetup \in \states^\numstatequant, \modetup \in \modes^\numstatequant$ and $\phi$ a subformula of $\pschedq$.
	Lastly, we need to update the calculation of probabilities in the encoding of probability expressions in \cref{alg:semEncProbUntil}, \cref{line:untilEpr,line:untilEloop} and \cref{alg:semEncProb}, \cref{line:nextE} as shown exemplarily for the Next operator in \cref{alg:semEncProbBounded}.
	\hfill\qed
\end{proof}

\begin{figure}[t]
\scalebox{0.95}{
	\begin{minipage}{1.04\linewidth}
		\begin{algorithm}[H]
			\caption{SMT encoding for the meaning of \HyperSGL probability expressions for bounded memory}
			\label{alg:semEncProbBounded}
			
			\KwIn{
				$\tbg = \tbgtup$: TSG, \newline
				$\phi$: \HyperSGL probability epxression, \newline
				$k$: Memory bound, \newline
				$\numstatequant$: Number of state quantifiers, \newline
				$\countQuant$: Counter for indexing strategy quantifications, \newline
				$\currMap$: Mapping pairs of state variables and agents to strategy quantification indices.
			}
			\KwOut{SMT encoding of the meaning of $\phi$ for $\tbg$.}
			
			\Fn{\FSemProb{$\tbg, \phi, k, \numstatequant, \countQuant, \currMap$}}{
				{\ldots\;}
				\ElseIf{$\phi$ is $\Prob(\Next \phi')$}{
					{$E', \countQuant' := \FSem(\tbg, \phi', k, \numstatequant, \countQuant, \currMap)$\; }
					{$E := E' \wedge \bigwedge_{\statetup \in \states^\numstatequant}
						\bigwedge_{\modetup \in \modes^\numstatequant} \Bigl((\holdsToIntWith[\statetup, \modetup][\phi'] {=} 1 \wedge \holdsWith[\statetup, \modetup][\phi']) \vee (\holdsToIntWith[\statetup, \modetup][\phi'] {=} 0 \wedge \neg\holdsWith[\statetup, \modetup][\phi'])\Bigr)$\;}
					\ForEach{$\statetup = (s_1, \ldots, s_\numstatequant) \in \states^{\numstatequant}$}{
						\tcp{Retrieve index of currently relevant strategy for each $s_i$}
						{$\currInd := [\currMap[i, \agents(s_i)] \textit{ for } i=1, \ldots, \numstatequant]$\;}
						\ForEach{$\modetup = (\mode_1, \ldots, \mode_n) \in \modes^\numstatequant$}{
							{$E := E \wedge \displaystyle 
								\probWith = \displaystyle
								\sum_{\actiontup \in \Act(\statetupMap)^\numstatequant}
								\sum_{\statetup' \in \fsucc(\statetup, \actiontup)}
								\sum_{\modetup' \in \modes^\numstatequant}$ 
								\qquad 
								$\left(\prod_{i=1}^\numstatequant
								\stratencWith[{\currInd[i]}][\state_i, \mode_i][\action_i]
								\cdot 
								\modefencWith[{\currInd[i]}][\state_i][\mode_i][\mode'_i]
								\cdot
								\Trans(\state_i, \action_i, \state'_i)\right)
								\cdot
								\holdsToIntWith[\statetup', \modetup'][\phi']
								$\;}
						}
					}
					{\Return $E, \countQuant'$\;}
				}
				{\ldots\;}
			}
		\end{algorithm}
	\end{minipage}
}
\end{figure}
	\section{\EXPTIME-Membership over Memoryless Deterministic Strategies (\cref{th:exptime} from \cref{sec:mc})}
\label{app:exptime-membership}

\exptime*

The lower bound follows from the fact that model-checking \HyperPCTL on DTMCs is \PSPACE-hard \cite{abrahamHyperPCTLTemporal2018} and \HyperSGL subsumes \HyperPCTL on DTMCs 
over any scheduler class by the construction given in \cref{sec:hyperpctl}.

We show membership in \EXPTIME by presenting a model-checking algorithm and showing that its worst-case running time is exponential in the size of the inputs.
\cref{alg:check,alg:calc} show a naive, brute-force decision procedure which resolves state and strategy quantification by testing all possible states or (memoryless deterministic) strategies, respectively. 
Both algorithms take as input 
a turn-based game $\tbg$, 
a dictionary $\mapStrat$ mapping tuples of states variables and agents to strategies,
a dictionary $\statedict$ mapping state variable indices to states,
and a \HyperSGL formula or probability expression, respectively.
The helper function \FSucc{$\tbg, \mapStrat, \statetupMap$} returns the set of all tuples $(\statedict[t], p)$ of state variable mappings $\statedict[t]$ and probabilities $p \in [0,1]$ such that $\statedict[t]$ corresponds to a successor state of $\statedict$ in the DTMC $\dtmc$ induced by $\mapStrat$ on $\tbg$, and $p$ is the probability of transitioning from $\statedict$ to $\statedict[t]$ in $\dtmc$. 

\begin{figure}[p]
	\begin{algorithm}[H]
		\caption{Simple model-checking algorithm for memoryless deterministic strategies}
		\label{alg:check}
		
		\KwIn{$\tbg=\tbgtup$: TSG, \newline
			$\mapStrat$: Mapping tuples of states variables and agents to strategies, \newline
			$\statedict$: Mapping state variable indices to states, \newline
			$\phi$: \HyperSGL formula.
		}
		\KwOut{Whether $\tbg$ satisfies $\phi$ over memoryless deterministic strategies.}
		\Fn{\FCheck{$\tbg, \mapStrat, \statedict, \phi$}}{
			\If{$\phi$ is $\exists \varstate_i \ \psi $}{
				\ForEach(\tcp*[f]{finitely many states}){$\state \in \states$}{
					\lIf{\FCheck{$\Ctxtup[][][{\statedict{[i \mapsto \state]}}], \psi$}}{\Return{true}}
				}
				\Return{false}
			}
			\ElseIf{$\phi$ is $\forall \varstate_i \ \psi $}{
				\ForEach(\tcp*[f]{finitely many states}){$\state \in \states$}{
					\lIf{\algorithmicnot \FCheck{$\Ctxtup[][][{\statedict{[i \mapsto \state]}}], \psi$}}{\Return{false}}
				}
				\Return{true}
			}
			\ElseIf{$\phi$ is $\existsStrat{\setofagents}{\{\varstate_{i_1}, \ldots, \varstate_{i_m}\}} \psi $}{
				\ForEach(\tcp*[f]{finitely many strategies}){$\strat \in \Strats[]$}{
					{$\mapStrat^{\strat} := \mapStrat[(\varstate_{i_l}, j) \mapsto \strat_j \text{ for } j\in \setofagents, l{=}1, \ldots, m]$ \;}
					\lIf{\FCheck{$\Ctxtup[][{\mapStrat^{\strat}}], \psi$}}{\Return{true}}
				}
				\Return{false}
			}
			\tcp{$\forallStrat{A}{\varstate} \psi$ can be covered by $\neg \existsStrat{A}{\varstate} \neg \psi$ or analogous to $\existsStrat{}{}$}
			\lElseIf{$\phi$ is $\true$}{
				\Return{true}
			}
			\lElseIf{$\phi$ is $\ap_{\varstate_i}$}{
				\Return{$(\ap \in \labelingfct(\statedict{[i]}))$}
			}
			\lElseIf{$\phi$ is $\neg \psi_1$}{
				\Return{\algorithmicnot \FCheck{$\Ctxtup, \psi_1$}}
			}
			\ElseIf{$\phi$ is $\psi_1 \wedge \psi_2$}{
				\Return{\FCheck{$\Ctxtup, \psi_1$} \algorithmicand \FCheck{$\Ctxtup, \psi_2$}}
			}
			\ElseIf{$\phi$ is $\pprob_1 \sim \pprob_2$}{
				\Return{$($\FCalc{$\Ctxtup, \pprob_1$}$) \sim ($\FCalc{$\Ctxtup, \pprob_2$}$)$}
			}
			\lElseIf{$\phi$ is $\varstate_i = \varstate_j$}{
				\Return{$\statedict{[i]} == \statedict{[j]}$}}
		}
	\end{algorithm}
\end{figure}

\begin{figure}[t]
	\begin{algorithm}[H]
		\caption{Evaluating probability expressions}
		\label{alg:calc}
		
		\KwIn{$\tbg=\tbgtup$: TSG, \newline
			$\mapStrat$: Mapping tuples of states variables and agents to strategies, \newline
			$\statedict$: Mapping state variable indices to states, \newline
			$\pprob$: \HyperSGL probability expression.
		}
		\KwOut{$\llbracket \pprob \rrbracket_{\tbg, \mapStrat, \statedict} \in [0,1]$}
		
		\Fn{\FCalc{$\tbg, \mapStrat, \statedict, \pprob$}}{
			\If{$\pprob$ is $\Prob(\Next \psi)$}{
				{$succ := $ \FSucc{$\Ctxtup$} \;} \label{line:calc-succ}
				{$\textit{sum} := 0$ \;}
				\ForEach(\tcp*[f]{$p$: transition probability}){$(\statedict[t], p) \in succ$}{
					\lIf{\FCheck{$\Ctxtup[][][{\statedict[t]}], \psi$}}{
						$\textit{sum} := \textit{sum} + p$
					}
				}
				\Return{$\textit{sum}$}
			}
			\ElseIf{$\pprob$ is $\Prob(\psi_1 \Until \psi_2)$}{ 
				{$\Sat(\psi_1) := \{\vec{t} \in \states^n \mid $ \FCheck{$\Ctxtup[][][{\statedict[t]}], \psi_1$}$\}$ \tcp*{$n = |\statedict|$}} \label{line:psi1} 
				{$\Sat(\psi_2) := \{\vec{t} \in \states^n \mid $ \FCheck{$\Ctxtup[][][{\statedict[t]}], \psi_2$}$\}$ \;} \label{line:psi2}
				{Calculate transition matrix $\Trans^{\mapStrat}$ of $\tbg^{\mapStrat} := \dtmc := \bigtimes_{\varstate_i \in \Dom(\statetupMap)} \tbg^{\indexi[\strat]}$ where $\indexi[\strat] := \bigoplus_{\agent \in \agents} \mapStrat(\indexi[\varstate], g)$ for $i \in \Dom(\statetupMap)$ \;} \label{line:matrix}
				{$\states_{=0} := \{ \vec{t} \in \states^n \mid \Pr(\vec{t} \models \Sat(\psi_1) \Until \Sat(\psi_2)) = 0\}$ \;} \label{line:S0}
				{$\states_{=1} := \{ \vec{t} \in \states^n \mid \Pr(\vec{t} \models \Sat(\psi_1) \Until \Sat(\psi_2)) = 1\}$ \;} \label{line:S1}
				{$\states_{?} := \states^n \setminus (\states_{=0} \cup \states_{=1})$ \;} \label{line:S?}
				{Solve $x = (\Trans^{\mapStrat}(\vec{t}, \vec{t'}))_{\vec{t}, \vec{t'} \in \states_{?}} \cdot x + (\Trans^{\mapStrat}(\vec{t}, \states_{=1}))_{\vec{t} \in \states_{?}}$\;} \label{line:lgs}
				\Return{$x$}
			}
			\ElseIf{$\pprob$ is $f(\psi_1, \ldots, \psi_j)$}{
				\Return{$f($\FCalc{$\Ctxtup, \psi_1$}, \ldots, \FCalc{$\Ctxtup, \psi_j$}$)$}
			}
		}
	\end{algorithm}
\end{figure}

\begin{lemma}[Time Complexity]
	For a turn-based game $\tbg$ with $|\states|>1$ and $|\Act|>1$, and a well-formed \HyperSGL formula $\phi$, calling $\FCheck(\tbg, (), (), \phi)$ results in a running time complexity of 
	\[O \left( |\phi| \cdot |\states|^{n |\phi|} \cdot |\Act|^{(|\states| + n) |\phi|} \cdot |\states|^n \right), \] 
	where 
	$n$ is the number of state quantifiers in $\phi$.
\end{lemma}

\begin{proof}
	Let $\tbg$ be a turn-based game, $\phi' = Q_1 \varstate_{1} \ldots Q_n \varstate_n .\ \pschedq$ a well-formed \HyperSGL formula. 
	In the following, we use $time(\psi)$ to denote the running time of a recursive call to $\FCheck$ or $\FCalc$, respectively, for a subformula $\psi$ of $\phi'$ in the context of $\FCheck(\tbg, (), (), \phi')$. 
	
	In the worst case, calling $\FCheck(\tbg, (), (), \phi')$ results in a recursive call to $\FCheck(\tbg, (), \statedict, \pschedq)$ for each possible mapping $\statedict$ of the $n$ state variables. 
	The runtime of $\FCheck(\tbg, (), (), \phi')$ can hence be bounded by $time(\pschedq) \cdot |\states|^n$. 
	
	Let us now analyze the running time of a call to $\FCheck(\tbg, \mapStrat, \statedict, \phi)$ or $\FCalc(\tbg, \mapStrat, \statedict, \phi)$ for a subformula $\phi$ of $\phi'$ that does not contain state quantifiers, 
	a mapping of the strategies $\mapStrat$, 
	and a mapping of the state variables $\statedict$. 
	We show that there exists a constant $\kappa$ such that $time(\phi) \leq \kappa \cdot |\phi| \cdot |\states|^{n |\phi|} \cdot  |\Act|^{|\states| m_{\phi} + n |\phi|}$, where $n$ is the number of state quantifiers in $\phi'$ and $m_{\phi}$ is the maximal number of strategy quantifications along any path in the syntax tree of $\phi$.
	Clearly, $m_{\phi} \leq |\phi|$.
	We proceed by structural induction on $\phi$.	
	
	If $\phi = \true$ or $\phi = \ap_{\varstate}$, then the claim follows directly.
	
	If $\phi = \existsStrat{\setofagents}{\varstate_1, \ldots, \varstate_n} \psi$, then 
	$m_{\phi} = m_\psi + 1$ and $|\phi| > |\psi|$.
	In the worst case, calling $\FCheck(\tbg, \mapStrat, \statedict, \phi)$ results in a call to $\FCheck(\tbg, \mapStrat^\strat, \statedict, \psi)$ for all memoryless deterministic $\strat \in \Strats$. 
	The number of memoryless deterministic strategies for $\setofagents$ can be bounded by $|\Act|^{|\states|}$.
	Hence, we have
	\begin{align*}time(\phi) 
		&\leq time(\psi) \cdot |\Act|^{|\states|}
		\\&\leq \kappa \cdot |\psi| \cdot |\states|^{n |\psi|} \cdot |\Act|^{|\states| m_{\psi} + n |\psi|} \cdot |\Act|^{|\states|}
		\\&= \kappa \cdot |\psi| \cdot |\states|^{n |\psi|} \cdot |\Act|^{|\states| m_{\phi} + n |\psi|}
		\\&\leq \kappa \cdot |\phi| \cdot |\states|^{n |\phi|} \cdot |\Act|^{|\states| m_{\phi} + n |\phi|} .
	\end{align*}
	The reasoning for universal strategy quantification follows analogously.

	If $\phi = \psi_1 \wedge \psi_2$, then $|\phi| > |\psi_1| + |\psi_2|$
	and
	$m_{\phi} = \max_i{m_{\psi_i}}$, and we have
	\begin{align*}time(\phi) 
		&= time(\psi_1) + time(\psi_2) + 1
		\\&\leq \kappa \cdot |\psi_1| \cdot |\states|^{n |\psi_1|} \cdot |\Act|^{|\states| m_{\psi_1} + n |\psi_1|}  
		\\&\quad 		
		+ \kappa \cdot |\psi_2| \cdot |\states|^{n |\psi_2|} \cdot |\Act|^{|\states| m_{\psi_2} + n |\psi_2|}  
		+ 1
		\\&\leq \kappa \cdot \big(|\psi_1| + |\psi_2|\big) \cdot |\states|^{n |\phi|} \cdot |\Act|^{|\states| m_{\phi} + n |\phi|} + 1
		\\&\leq \kappa \cdot |\phi| \cdot |\states|^{n |\phi|} \cdot  |\Act|^{|\states| m_{\phi} + n |\phi|} .
	\end{align*}
	The reasoning for $\phi = \neg \psi$, $\phi = \pprob_1 \sim \pprob_2$ and $\phi = f(\pprob_1, \ldots, \pprob_j)$ follows analogously.

	If $\phi = \Prob(\psi_1 \Until \psi_2)$, then $|\phi| > |\psi_1| + |\psi_2|$
	and $m_{\phi} = \max_i{m_{\psi_i}}$. 
	In \cref{alg:calc}, \cref{line:psi1,line:psi2}, we need to call \FCheck{$\tbg, \mapStrat, \statedict[t], \psi_1$} and \FCheck{$\tbg, \mapStrat, \statedict[t], \psi_2$}, respectively, for all possible state variable mappings $\statedict[t]$ of size $n$.
	Calculating the transition matrix of the induced DTMC (Alg.\ \ref{alg:calc}, \cref{line:matrix}) takes $O( n |\states|^n |\Act|^n)$. 
	The qualitative sets (Alg.\ \ref{alg:calc}, \cref{line:S0,line:S1,line:S?}) can be computed in $O(|\states|^{2 n})$. 
	The linear equation system created in \cref{alg:calc}, \cref{line:lgs} has at most $|\states|^n$ equations and can be solved in at most $O(|\states|^{3 n})$.
	In summary, the running time of \cref{alg:calc}, \cref{line:matrix,line:S0,line:S1,line:S?,line:lgs} can be bounded by $O(|\states|^{3 n} |\Act|^n)$ since $n \leq |\states|^n$ for $|\states|>1$.
	Hence, there exists a constant $c$ such that
	\begin{align*}time(\phi) 
		&= |\states|^{n} \cdot time(\psi_1) + |\states|^{n} \cdot time(\psi_2) + c \cdot |\states|^{3 n} \cdot |\Act|^n
		\\&\leq |\states|^n \cdot \kappa \cdot |\psi_1| \cdot |\states|^{n |\psi_1|} \cdot |\Act|^{|\states| m_{\psi_1} + n |\psi_1|} 
		\\&\quad + |\states|^n \cdot \kappa \cdot |\psi_2| \cdot |\states|^{n |\psi_2|} \cdot |\Act|^{|\states| m_{\psi_2} + n |\psi_2|}  
		+ c \cdot |\states|^{3 n}\cdot |\Act|^n 
		\\&\leq \kappa \cdot |\psi_1| \cdot |\states|^{n |\phi|} \cdot |\Act|^{|\states| m_{\phi} + n |\phi|} 
		\\&\quad 
		+ \kappa \cdot |\psi_2| \cdot |\states|^{n |\phi|} \cdot |\Act|^{|\states| m_{\phi} + n |\phi|}  
		+ c \cdot |\states|^{n |\phi|} \cdot |\Act|^{|\states| m_{\phi} + n |\phi|} .
	\end{align*}
	If we choose $\kappa \geq c$, we can conclude that
	\begin{align*}time(\phi) 
		&\leq \kappa \cdot (|\psi_1| + |\psi_2| + 1) \cdot |\states|^{n |\phi|} \cdot |\Act|^{|\states| m_{\phi} + n |\phi|}  
		\\&\leq \kappa \cdot |\phi| \cdot |\states|^{n |\phi|} \cdot |\Act|^{|\states| m_{\phi} + n |\phi|} .
	\end{align*}
	
	If $\phi = \Prob(\Next \psi)$, then $|\phi| > |\psi|$ and thus $|\phi|\geq 2$
	and $m_{\phi} = m_{\psi}$.
	The successor state tuples of $\statedict$ in the induced composed DTMC (Alg.\ \ref{alg:calc}, \cref{line:calc-succ}) can be computed in $O( n |\states|^n |\Act|^n)$.
	Hence, there exists a constant $c'$ such that 
	\begin{align*}time(\phi) 
		&\leq |\states|^{n} \cdot (time(\psi_1) + 1) + c' \cdot n \cdot |\states|^{n} \cdot |\Act|^n 
		\\&\leq |\states|^n \cdot \kappa \cdot |\psi_1| \cdot |\states|^{n |\psi_1|} \cdot |\Act|^{|\states| m_{\psi_1} + n |\psi_1|} 
		+ |\states|^n + c' \cdot |\states|^{2 n} \cdot |\Act|^n
		\\&\leq |\states|^n \cdot \kappa \cdot |\psi_1| \cdot |\states|^{n |\psi_1|} \cdot |\Act|^{|\states| m_{\psi_1} + n |\psi_1|}  
		+ c' \cdot |\states|^{2 n} \cdot (|\Act|^{n} +1)
		\\&\leq \kappa \cdot |\psi_1| \cdot |\states|^{n (|\psi_1| + 1)} \cdot |\Act|^{|\states| m_{\psi_1} + n |\psi_1|} 
		+ c' \cdot |\states|^{2 n} \cdot |\Act|^{|\states| m_{\phi} + n |\phi|} 
		\\&\leq \kappa \cdot |\psi_1| \cdot |\states|^{n |\phi|} \cdot |\Act|^{|\states| m_{\phi} + n |\phi|} 
		+ c' \cdot |\states|^{n |\phi|} \cdot |\Act|^{|\states| m_{\phi} + n |\phi|} .
	\end{align*}
	If we choose $\kappa \geq c'$, we can conclude that
	\begin{align*}time(\phi) 
		&\leq \kappa \cdot |\psi_1| \cdot |\states|^{n |\phi|} \cdot |\Act|^{|\states| m_{\phi} + n |\phi|} 
		+ \kappa \cdot |\states|^{n |\phi|} \cdot |\Act|^{|\states| m_{\phi} + n |\phi|} 
		\\&\leq \kappa \cdot |\phi| \cdot |\states|^{n |\phi|} \cdot |\Act|^{|\states| m_{\phi} + n |\phi|} .
	\end{align*}
	\hfill\qed	
\end{proof}
	\section{Fixed Number of State Quantifiers over Memoryless Deterministic Strategies (\cref{th:fixedNoQuant} from \cref{sec:mc})}
\label{app:mc-fixed-number-quantifiers}

It remains an open questions whether model-checking \HyperSGL over memoryless deterministic schedulers is possible in \PSPACE. 
Naive approaches fail because, in the worst case, we have to check all possible combinations for state variable assignments, and the number of these combinations is exponential in the number of state quantifiers.
For a fixed number of state quantifiers, however, we can show that the model-checking problem for \HyperSGL over memoryless deterministic strategies is in \PSPACE and \PSPACE-hard.

We use $\modelsMD$ to denote that we restrict the quantification over strategies to \textbf{m}emoryless \textbf{d}eterministic strategies.

\fixedNoQuant*

\begin{proof} 
	\emph{Hardness:}
	We can show that the \HyperSGL model-checking problem over memoryless deterministic strategies is \PSPACE-hard even for a single state quantifier by reduction from the \emph{quantified Boolean formula} problem (\emph{QBF}).
	We can use the construction used for showing \PSPACE-hardness of \SGL model-checking over memoryless deterministic strategies presented in \cite{baierStochasticGame2012}, because the formula constructed there can be directly translated to \HyperSGL with a single state quantifier.
	The semantical differences between \SGL and \HyperSGL detailed in \cref{sec:sgl} and \cref{app:sgl-proof} do not matter here because 
	(1) there is no implicit strategy quantification and 
	(2) we only quantify over memoryless deterministic strategies, so it does not matter whether strategies are \enquote{shifted} by temporal operators.
	
	\paragraph{Membership:}
	Let $\tbg$ be a turn-based game, and $\phi' = Q_1 \varstate_{1} \ldots Q_\numstatequant \varstate_\numstatequant .\ \pschedq$ a well-formed \HyperSGL formula for some $\numstatequant \in \N$. 
	We can decide whether $\tbg \models \phi'$ by iterating over all $(\state_1, \ldots, \state_\numstatequant) \in \states^n$ and 
	checking
	whether $\Ctxtup[][\emptymap][\statetupMap_{\state_1, \ldots, \state_\numstatequant}] \modelsMD \pschedq$, where $\statetupMap_{\state_1, \ldots, \state_\numstatequant}(i) = \state_i$ for $i=1, \ldots, \numstatequant$. 
	If we store the model-checking result for state tuple $(\state_1, \ldots, \state_\numstatequant) \in \states^n$ in a Boolean variable $\textit{holds}_{\pschedq, (\state_{1}, ..., \state_{\numstatequant})}$, we can resolve the quantification by computing 	
	\[ \textit{holds}_{\phi'} := \circ^{Q_1}_{\state_{1} \in \states} \ldots \circ^{Q_\numstatequant}_{\state_{\numstatequant} \in \states} \textit{holds}_{\pschedq, (\state_{1}, ..., \state_{\numstatequant})}\] where 
	$\circ^{Q_i} = \sum_{}$ if $Q_i = \exists$ and $\circ^{Q_i} = \prod_{}$ otherwise.
	Then, $\textit{holds}_{\phi'} >0$ if and only if $\tbg \modelsMD \phi'$.
	The number of recursive model-checking queries and the number of performed arithmetic operations are both polynomial in the size of the input (since $\numstatequant$ is fixed).
	
	Let us now consider the model-checking complexity for subformulas of $\pschedq$. 
	We adapt the proof for showing that model-checking \SGL over memoryless deterministic strategies is in \PSPACE~\cite{baierStochasticGame2012}.
	We will recursively define the type of a non-state-quantified \HyperSGL formula as a class of the polynomial time hierarchy, and then show that this type yields an upper bound for the model-checking complexity of the class of formulas of this type. 
	Since the polynomial hierarchy is contained in \PSPACE, this shows that the model-checking problem for \HyperSGL formulas with a fixed number of state quantifiers over memoryless deterministic strategies is in \PSPACE.
	
	Recall the polynomial hierarchy with $\Delta_0 = \Sigma_0 = \Pi_0 = \P$ and $\Delta_{i+1} = \P^{\Sigma_i}$, $\Sigma_{i+1} = \NP^{\Sigma_i}$, $\Pi_{i+1} = \coNP^{\Sigma_i}$
	and that the polynomial hierarchy is contained in \PSPACE
	\cite{papadimitriouComputationalComplexity1994}.
	We define the maximum of a set of classes $\{\Lambda_{i_1}, \ldots, \Lambda_{i_j}\}$ with $\Lambda \in \{\Delta, \Sigma, \Pi \}$ and $m := \Max{i_1, \ldots, i_j}$ as follows: 
	\[\Max{\Lambda_{i_1}, \ldots, \Lambda_{i_j}} = \begin{cases}
		\Delta_m & \text{if } \Lambda_{i_l} = \Delta_{i_l} \text{ for all } l \text{ with } i_l = m \\
		\Sigma_m & 
		\text{if } \Lambda_{i_l} \in \{\Pi_{i_l}, \Sigma_{i_l}\} \text{ for some } l \text{ with } i_l = m \footnotemark.
	\end{cases}\] 
	\footnotetext{We do not need to distinguish $\Pi_l$ and $\Sigma_l$ because $\Max{}$ will only be used in $\P^{\Max{}}$ and $\P^{\Sigma_l} = \P^{\Pi_l}$.} 
	
	\noindent The type of a non-state-quantified \HyperSGL formula is defined inductively:	
	\[\begin{array}{lcl}
		\Type(\true) & = & \Delta_0 \\
		\Type(\ap_{\varstate}) & = & \Delta_0 \\
		& & \\
		\Type(\pnonquant_1 \wedge \pnonquant_2) & = & 
		\P^{\Max{\Type(\pnonquant_1), \Type(\pnonquant_2)}} \\
		& = & 
		\begin{cases}
			\Delta_l & \text{if } \Max{\Type(\pnonquant_1), \Type(\pnonquant_2)} = \Delta_l \\
			\Delta_{l+1} & \text{if } \Max{\Type(\pnonquant_1), \Type(\pnonquant_2)} = \Sigma_l
		\end{cases}\\
		\Type(\pprob_1 \sim \pprob_2) & = & \P^{\Max{\Type(\pprob_1), \Type(\pprob_2)}} \\
		\Type(\Prob(\Next \pnonquant)) & = & \P^{\Type(\pnonquant)} \\
		\Type(\Prob(\pnonquant_1 \Until \pnonquant_2)) & = & \P^{\Max{\Type(\pnonquant_1), \Type(\pnonquant_2)}} \\
		\Type(f(\pprob_1, \ldots, \pprob_j)) & = & \P^{\Max{\Type(\pnonquant_1), \ldots, \Type(\pnonquant_j)}} \\
		& & \\
		\Type(\existsStrat{\setofagents}{\varstate_{i_1}, \ldots, \varstate_{i_m}} \pschedq) & = & \NP^{\Type(\pschedq)} \\
		& = & 
		\begin{cases}
			\Sigma_l & \text{if } \Type(\pschedq) = \Delta_l \\
			\Sigma_{l+1} & \text{if } \Type(\pschedq) = \Sigma_l \text{ or } \Type(\pschedq) = \Pi_l
		\end{cases} \\
		& & \\
		\Type(\neg\pnonquant) & = & \begin{cases}
			\Delta_l & \text{if } \Type(\pnonquant) = \Delta_l \\
			\Pi_l & \text{if } \Type(\pnonquant) = \Sigma_l  \\
			\Sigma_l & \text{if } \Type(\pnonquant) = \Pi_l			
		\end{cases}\\
		\Type(\forallStrat{\setofagents}{\varstate_{i_1}, \ldots, \varstate_{i_m}} \pschedq) & = & \coNP^{\Type(\pschedq)} \\
		& = & \begin{cases}
			\Pi_l & \text{if } \Type(\pschedq) = \Delta_l \\
			\Pi_{l+1} & \text{if } \Type(\pschedq) = \Sigma_l \text{ or } \Type(\pschedq) = \Pi_l 
		\end{cases}  \\
	\end{array}\]
	
	The set of all \HyperSGL formulas can thus be divided into sets of classes of the same type.
	We analyze the model-checking complexity for each such class separately, and show that it corresponds exactly to the type of this class. 
	More specifically, we show the following claim: 
	For a subformula $\phi$ of $\pschedq$, a turn-based game $\tbg = \tbgtup$, a strategy mapping $\mapStrat \colon \VarsState \times \agents \rightharpoonup \bigcup_{\agent \in \agents} \Strats[][\{\agent\}]$, and a state variable mapping $\statetupMap \colon \N \rightharpoonup \states$ with $|\Dom(\statetupMap)| = n$, 
	\begin{enumerate}
		\item \label{enum:fixed:quant}
		If $\phi$ is a 
		strategy-quantified or non-quantified formula, then deciding whether $\Ctxtup \modelsMD \phi$ is possible in $\Type(\phi)$, and 
		
		\item \label{enum:fixed:prob}
		If $\phi$ is a probability expression, then calculating the value $\llbracket \phi \rrbracket_{\Ctxtup}$ is possible in $\Type(\phi)$. 
	\end{enumerate}
	We proceed by induction on the structure of $\phi$.
	
	\begin{itemize}
	\item \emph{Cases $\phi = \true$, $\phi = \ap_{\varstate}$.} 
	The claim follows directly.

	\item \emph{Case $\phi = \pnonquant_1 \wedge \pnonquant_2$.} 
	By induction hypothesis, the problem of deciding whether $\Ctxtup \modelsMD \pnonquant_i$ is in $\Type(\pnonquant_i)$ for $i=1, 2$. 
	
	If $\Max{\Type(\pnonquant_1), \Type(\pnonquant_2)} = \Delta_l$, then also $\Ctxtup \modelsMD \pnonquant_1 \wedge \pnonquant_2$ can be solved in $\Delta_l$.
	
	Otherwise, $\Max{\Type(\pnonquant_1), \Type(\pnonquant_2)} = \Sigma_l$.
	Then, $\Ctxtup \modelsMD \pnonquant_1 \wedge \pnonquant_2$ can be solved by a deterministic polynomial-time Turing machine with an oracle for $\Sigma_l$, i.e., $\Type(\phi) = \P^{\Sigma_l} = \Delta_{l+1}$.

	\item \emph{Cases $\phi = \pprob_1 \sim \pprob_2$, $\phi = f(\pprob_1, \ldots, \pprob_j)$.} Analogously.

	\item \emph{Case $\phi = \Prob(\Next \pnonquant_1)$.}
	We can construct the set $\Sat(\pnonquant_1)$ of all $\statetupMap[t] \colon \N \to \states$ with $\Dom(\statetupMap[t]) = \Dom(\statetupMap)$ such that $\Ctxtup[][][{\statetupMap[t]}] \modelsMD \pnonquant_1$ using $|\states|^n$ calls to an oracle for $\Type(\pnonquant_1)$.
	Then, we can evaluate $\phi$ in $\tbg^{\mapStrat}_{\statetupMap}$ by summing over the probability of transitioning from $\statetupMap$ to $\statetupMap[t]$ for all $\statetupMap[t] \in \Sat(\pnonquant_1)$.
	Hence, $\Type(\phi) = \P^{\Type(\pnonquant_1)}$ since $n$ is fixed.

	\item \emph{Case $\phi = \Prob(\pnonquant_1 \Until \pnonquant_2)$.}
	For all $\statetupMap[t] \colon \N \to \states$ with $\Dom(\statetupMap[t]) = \Dom(\statetupMap)$ and $i=1,2$, we can check whether $\Ctxtup[][][{\statetupMap[t]}] \modelsMD \pnonquant_i$ using an oracle for $\Type(\pnonquant_i)$.
	Using these results, we can evaluate $\phi$ in $\tbg^{\mapStrat}_{\statetupMap}$ by solving a linear equation system with $|\states|^n$ equations\footnote{Here the reasoning would break if $n$ is not fixed or we do not bound the memory of the strategies.}, which is possible in time $O(|\states|^{3n})$.
	Hence, we can calculate $\llbracket \phi \rrbracket_{\Ctxtup}$ using $2 \cdot |\states|^n$ calls to an oracle for $\Max{\Type(\pnonquant_1), \Type(\pnonquant_2)}$ and thus $\Type(\phi) = \P^{\Max{\Type(\pnonquant_1), \Type(\pnonquant_2)}}$.

	\item \emph{Case $\phi = \existsStrat{\setofagents}{\varstate_{i_1}, \ldots, \varstate_{i_m}} \ \pschedq$.} 
	We can guess a memoryless deterministic $\setofagents$-strategy $\alpha$ and check whether $\Ctxtup[][\mapStrat'] \modelsMD \pschedq$, 
	where $\mapStrat'$ corresponds to $\mapStrat$ updated by $\alpha$ for $\varstate_{i_1}, \ldots, \varstate_{i_m}$. 
	By induction hypothesis, $\pschedq$ can be model-checked in $\Type(\pschedq)$ and hence $\Type(\phi) = \NP^{\Type(\pschedq)}$.

	\item \emph{Case $\phi = \neg \pnonquant$.}
	If $\Type(\pnonquant) = \Delta_l$ for some $l \in \N$, then there exists a deterministic polynomial-time Turing machine with an oracle for $\Sigma_{l-1}$ solving $\Ctxtup \modelsMD \pnonquant$. By negating this output, we can construct a deterministic polynomial-time Turing machine with an oracle for $\Sigma_{l-1}$ for solving $\Ctxtup \modelsMD \phi$ and $\Type(\phi) = \Delta_l$ in this case.
	
	If $\Type(\pnonquant) = \Sigma_l$, then we can verify yes-instances of $\Ctxtup \modelsMD \pnonquant$ using a nondeterministic polynomial-time Turing machine with an oracle for $\Sigma_{l-1}$. 
	Hence, we can use this Turing machine to verify no-instances of $\Ctxtup \modelsMD \phi$ and $\Type(\phi) = \Pi_l$.
	
	If $\Type(\pnonquant) = \Pi_l$, the reasoning follows analogously.

	\item \emph{Case $\phi = \forallStrat{\setofagents}{\varstate_{i_1}, \ldots, \varstate_{i_m}} \pschedq$.} 
	We can guess a memoryless deterministic $\setofagents$-strategy $\strat$ and check whether $\Ctxtup[][\mapStrat'] \not\modelsMD \pschedq$, 
	where $\mapStrat'$ corresponds to $\mapStrat$ updated by $\alpha$ for $\varstate_{i_1}, \ldots, \varstate_{i_m}$. 
	By induction hypothesis, we can solve $\Ctxtup[][\mapStrat'] \modelsMD \pschedq$ using an oracle for $\Type(\pschedq)$.
	If $\Type(\pschedq) = \Delta_l$, then solving $\Ctxtup[][\mapStrat'] \not\modelsMD \pschedq$ is also possible in $\Delta_l$ and $\Type(\phi) = \coNP^{\Delta_l}$.
	If $\Type(\pschedq) = \Sigma_l$, then solving $\Ctxtup[][\mapStrat'] \not\modelsMD \pschedq$ is possible in $\Pi_l$, and $\Type(\phi) = \coNP^{\Pi_l} = \coNP^{\Sigma_l}$.
	Analogously for $\Type(\pschedq) = \Pi_l$.
	In summary, $\Type(\phi) = \coNP^{\Type(\pschedq)}$.
	\end{itemize}
	
	In conclusion, deciding whether $\tbg \modelsMD Q_1 \varstate_{1} \ldots Q_\numstatequant \varstate_\numstatequant \ \pschedq$ is possible in $\P^{\Type(\pschedq)}$.	
	\hfill\qed
\end{proof}
}{}

\end{document}